\documentclass{ieeeaccess}
\usepackage{etoolbox}
\DeclareRobustCommand{\authororcid}[1]{\href{https://orcid.org/#1}{\raisebox{3pt}{\includegraphics[height=7pt]{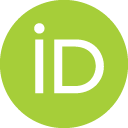}}}}
\makeatletter
\patchcmd{\@maketitle}{Digital Object Identifier\space\@doi}{\vphantom{Digital Object Identifier}}{}{\errmessage{DOI display patch failed}}
\def\ps@preprintdisplay{%
  \def\@oddhead{}\def\@evenhead{}%
  \def\@oddfoot{\hfil{\footerpagefont\thepage}\hfil}%
  \let\@evenfoot\@oddfoot}
\def\ps@titlepage{%
  \ps@preprintdisplay
  \def\@oddhead{\parbox[b]{\textwidth}{\normalfont\fontsize{9}{11}\selectfont\color{black}%
    Published in IEEE Access (2026). DOI: \href{https://doi.org/10.1109/ACCESS.2026.3737467}{10.1109/ACCESS.2026.3737467}.\\[3pt]
    {\fontsize{10}{12}\selectfont\bfseries
    \href{https://github.com/yifeipet/AFP_GIC}{\raisebox{-1.5pt}{\includegraphics[height=10pt]{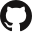}}\enspace Code and Model}%
    \quad\textnormal{\textperiodcentered}\quad
    \href{https://huggingface.co/spaces/yifeipet/AFP-GIC}{\raisebox{-1.5pt}{\includegraphics[height=10pt]{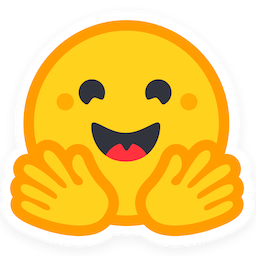}}\enspace Live Demo}%
    \quad\textnormal{\textperiodcentered}\quad
    \href{https://github.com/yifeipet/AFP_GIC/releases}{\raisebox{-1.5pt}{\includegraphics[height=10pt]{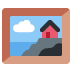}}\enspace Reconstructions and Metrics}}\par
    \vspace{4pt}{\color[gray]{0.8}\hrule height 0.3pt}}}%
  \let\@evenhead\@oddhead}
\let\ps@headings\ps@preprintdisplay
\AtBeginDocument{\definecolor{accessblue}{rgb}{0,0,0}\pagestyle{preprintdisplay}}
\makeatother

\AtBeginDocument{%
  \advance\pdfpagewidth by 5mm
  \advance\pdfhorigin by 5mm}

\usepackage{cite}
\usepackage{amsmath,amssymb,amsfonts}
\usepackage{amsthm}
\usepackage{algorithmic}
\usepackage{graphicx}
\usepackage{float}
\usepackage{multirow}
\usepackage{textcomp}
\usepackage[colorlinks=true,linkcolor=blue,citecolor=blue,urlcolor=blue]{hyperref}

\usepackage{bm}
\newcommand{\EqLink}[1]{\hyperref[#1]{Eq.~(\ref*{#1})}}
\newcommand{\EquationLink}[1]{\hyperref[#1]{Eq.~(\ref*{#1})}}
\newcommand{\EqsLink}[2]{\hyperref[#1]{Eqs.~(\ref*{#1})} and \hyperref[#2]{(\ref*{#2})}}
\newcommand{\EqRangeLink}[2]{\hyperref[#1]{Equations~(\ref*{#1})--(\ref*{#2})}}
\DeclareRobustCommand{\StageILink}{Stage~\hyperlink{stage-training-i}{I}}
\DeclareRobustCommand{\StageIILink}{Stage~\hyperlink{stage-training-ii}{II}}
\DeclareRobustCommand{\StageIIILink}{Stage~\hyperlink{stage-training-iii}{III}}
\newcommand{\alphabf}{\bm{\alpha}}
\newtheorem{proposition}{Proposition}

\makeatletter
\AtBeginDocument{\DeclareMathVersion{bold}
\SetSymbolFont{operators}{bold}{T1}{times}{b}{n}
\SetSymbolFont{NewLetters}{bold}{T1}{times}{b}{it}
\SetMathAlphabet{\mathrm}{bold}{T1}{times}{b}{n}
\SetMathAlphabet{\mathit}{bold}{T1}{times}{b}{it}
\SetMathAlphabet{\mathbf}{bold}{T1}{times}{b}{n}
\SetMathAlphabet{\mathtt}{bold}{OT1}{pcr}{b}{n}
\SetSymbolFont{symbols}{bold}{OMS}{cmsy}{b}{n}
\renewcommand\boldmath{\@nomath\boldmath\mathversion{bold}}}
\makeatother

\def\BibTeX{{\rm B\kern-.05em{\sc i\kern-.025em b}\kern-.08em
    T\kern-.1667em\lower.7ex\hbox{E}\kern-.125emX}}

\begin{document}
\bstctlcite{access_bstctl}
\hypersetup{colorlinks=true,linkcolor=blue,citecolor=blue,urlcolor=blue}
\history{\vphantom{Date of publication}}
\doi{}

\title{Adaptive Fused Prior Transfer for Controllable Generative Image Compression}
\author{\uppercase{Yifei Pei}\authororcid{0009-0001-4159-0882},
\uppercase{Ying Liu}\authororcid{0000-0003-3380-4243}, and
\uppercase{Nam Ling}\authororcid{0000-0002-5741-7937}\IEEEmembership{Life Fellow, IEEE}}

\address{Department of Computer Science and Engineering, Santa Clara University, Santa Clara, CA 95053 USA}

\markboth
{Pei \headeretal: Adaptive Fused Prior Transfer for Controllable Generative Image Compression}
{Pei \headeretal: Adaptive Fused Prior Transfer for Controllable Generative Image Compression}

\corresp{Corresponding authors: Yifei Pei (e-mail: ypei@scu.edu; yifeipet@gmail.com) and Nam Ling (e-mail: nling@scu.edu).}

\begin{abstract}
Learned image compression has achieved competitive rate-distortion performance through end-to-end optimized transforms, quantization, and entropy modeling. At very low bitrates, however, the compressed representation often cannot preserve fine textures and local structures, and distortion-oriented reconstruction can produce over-smoothed images with reduced naturalness. Perceptual and generative codecs address this problem by synthesizing missing details with reconstruction priors. Controllable codecs further allow one model to cover different bitrate and reconstruction preferences. However, controllability alone does not resolve the decoder-side reconstruction-prior problem: under severe bit constraints, the decoder must infer missing details from limited transmitted information, and existing codebook-based controllable designs generally rely on single-codebook token-based reconstruction priors. This paper proposes Adaptive Fused Prior Transfer for Controllable Generative Image Compression (AFP-GIC), a controllable codec that transfers an adaptive fused prior from a frozen pretrained AdaCode model. Encoder-side fused-prior features guide latent formation, while the decoder predicts a compatible fused prior from the compressed representation and selected control variables, enabling prior-guided reconstruction without transmitting the fused prior itself. A motivating analysis shows that better decoder-side fused-prior alignment tightens a reconstruction-error upper bound and that the fused-prior family contains single-codebook choices as special cases. Under the unified benchmark, AFP-GIC achieves 18.1\% lower decoder latency and uses 31.10 million (20.5\%) fewer inference parameters than DC-VIC. Experiments on Kodak, CLIC2020, and DIV2K show competitive PSNR and SSIM, with the clearest perceptual gains in NIQE scores and very-low-bitrate visual comparisons. Code is available at \url{https://github.com/yifeipet/AFP_GIC}.
\end{abstract}

\begin{keywords}
Adaptive fused-prior transfer, controllable image compression, generative image compression, learned image compression, low-bitrate reconstruction, pretrained prior transfer.
\end{keywords}

\titlepgskip=-21pt

\maketitle

\section{Introduction}
\label{sec:introduction}

\PARstart{L}{earned} image compression (LIC) has made substantial progress under conventional rate-distortion evaluation and is now frequently compared with established image codecs. Rate is typically measured in bits per pixel (bpp); distortion measures the discrepancy between the original and decoded images. Most LIC systems follow the nonlinear transform coding framework, in which learned analysis and synthesis transforms, quantization, and entropy modeling are optimized end to end \cite{balle2021ntc}. Early end-to-end and hyperprior-based methods established the main learned transform-coding paradigm \cite{balle2017e2e,balle2018hyperprior}. Subsequent joint autoregressive-hierarchical entropy priors and representative conditional probability models further improved entropy modeling \cite{minnen2018joint,mentzer2018conditional}. More recent architectures, such as ELIC, improved transform design, context modeling, and decoding efficiency \cite{he2022elic}. Benchmark studies also report continued improvements of learned codecs relative to conventional reference codecs under rate-distortion evaluation \cite{hu2022benchmark}. Still, strong rate-distortion performance does not by itself address the challenges of very-low-bitrate compression. At such rates, reconstruction quality depends on both efficient transmission and the decoder's ability to infer missing detail.

This difficulty is most evident in very-low-bitrate image compression. For the low-bitrate operating points considered in this paper, including rates below roughly 0.2 bpp, the compressed representation often lacks sufficient information to preserve fine textures and local structures. Conventional distortion-oriented codecs are often trained with pixel-domain losses such as mean squared error (MSE) and evaluated using metrics such as peak signal-to-noise ratio (PSNR), which is derived from MSE \cite{wang2009mse}. Optimizing pixel-fidelity objectives can therefore favor averaged outputs in low-bitrate regions where local structure is not fully encoded. This behavior reflects the perception-distortion tradeoff described by Blau and Michaeli in the rate-distortion-perception framework \cite{blau2018perceptiondistortion,blau2019rdp}: improving distortion metrics such as MSE or PSNR does not necessarily improve perceived image quality. In this paper, perceptual quality denotes the visually perceived quality of the decoded image. Realism is used more specifically for naturalness and plausibility, especially in synthesized textures and local structures. Thus, high PSNR does not necessarily imply visually sharp reconstructions. This motivates reporting perceptual metrics together with distortion metrics in the very-low-bitrate regime.

Perceptual and generative compression methods therefore use learned reconstruction priors to recover details not fully represented in the bitstream. HiFiC \cite{mentzer2020hific} combines learned compression with adversarial and perceptual objectives. MS-ILLM \cite{muckley2023illm} targets improved distributional fidelity through learned perceptual modeling. Recent diffusion-based codecs, such as PerCo \cite{careil2024perco} and conditional diffusion compression \cite{yang2023cdc}, have reported improved visual realism at low bitrates, particularly under perceptual or distribution-oriented evaluation. A remaining difficulty is that many perceptual codecs are trained for a limited set of operating preferences, such as a target bitrate range or a particular distortion-perception balance. Adapting the same model to different deployment requirements can therefore be difficult.

Controllable image compression offers one way to reduce this dependence on fixed operating preferences, since a single model can operate across different bitrate, distortion, and perceptual-quality targets. Multi-Realism demonstrated that one compressed representation can be decoded with different balances between pixel fidelity and realism \cite{agustsson2023multirealism}. CRDR \cite{iwai2024crdr} considered joint control of bitrate and reconstruction behavior. Control-GIC \cite{li2025controlgic} pursued a once-for-all controllable setting through dynamic granularity adaptation. DC-VIC combines controllable compression with a pretrained codebook-based generative model and explicitly models the allocation of information between token-based generation and feature-level refinement \cite{iwai2024dcvic}. These studies indicate the importance of reconstruction priors in controllable generative compression. They also leave open how such priors should be represented and made available at the decoder.

Codebook-based generative priors provide a discrete representation for token modeling and can be integrated with entropy coding. In many such designs, each location is represented by a token selected from a single learned codebook. VQ-VAE \cite{oord2017vqvae} and VQGAN \cite{esser2021vqgan} follow this discrete-codebook formulation. DC-VIC adopts a pretrained codebook-based reconstruction prior for controllable compression. AdaCode generalizes the single-codebook representation by adaptively fusing multiple basis codebooks according to image content \cite{liu2023adacode}. The resulting fused representation can describe a broader range of structures and textures than a single-codebook choice. This leads to the question considered in this paper: can an adaptive fused prior be incorporated into a controllable generative codec without transmitting the fused prior itself or violating entropy-constrained decoding?

To address this question, we investigate adaptive fused-prior transfer for low-bitrate controllable generative image compression. AFP-GIC builds on the dual-conditioned controllable codec structure used in recent controllable generative compression and replaces single-codebook prior modeling with prior transfer from a frozen pretrained AdaCode model. A key challenge is asymmetric prior availability. At the encoder, latent formation can be guided by fused-prior features extracted from the input image. At the decoder, reconstruction has access only to the transmitted compressed representations and compact control/header information. AFP-GIC handles this asymmetry by using fused-prior guidance at the encoder and prior prediction at the decoder. The bitstream therefore remains entropy-constrained. At the same time, the decoder can use a broader adaptive-prior representation during low-bitrate reconstruction.

Our contributions are:
\begin{itemize}
\item We introduce adaptive fused-prior transfer for controllable generative image compression. The proposed formulation generalizes prior-guided controllable compression from single-codebook reconstruction priors to a frozen adaptive fused-prior model.
\item We design an asymmetric encoder-decoder prior-transfer mechanism that mitigates the prior-availability mismatch in entropy-constrained compression. Adaptive fused-prior features guide latent formation at the encoder. At the decoder, a compatible fused prior is predicted from compressed representations and used to guide reconstruction through the frozen pretrained decoder.
\item We provide an analytical motivation under a Lipschitz decoder assumption, relating improved prior alignment to a tighter reconstruction-error upper bound and showing that the adaptively fused prior family contains the single-codebook alternative as a special case.
\item We pair adaptive fused-prior transfer with a lightweight fully convolutional decoder, which reduces the parameter footprint and measured decoding latency relative to DC-VIC in our evaluation \cite{iwai2024dcvic}.
\end{itemize}

\section{Related Work}
\label{sec:related_work}

\subsection{Learned Image Compression}
Most learned image compression methods follow the nonlinear transform coding framework, where an analysis transform maps the input image to a latent representation, quantization produces discrete codes, an entropy model estimates their probability, and a synthesis transform reconstructs the decoded image \cite{balle2021ntc}. Early end-to-end optimized codecs established this transform-coding formulation for learned compression \cite{balle2017e2e}. Scale hyperpriors were then introduced to transmit side information for more accurate entropy modeling \cite{balle2018hyperprior}, and joint autoregressive-hierarchical entropy priors further improved probability estimation for the quantized latents \cite{minnen2018joint}. Conditional probability models provide another representative direction for learned entropy modeling \cite{mentzer2018conditional}. More recent systems, such as Cheng et al.'s attention-based codec and ELIC, improve transform design, context modeling, and coding efficiency \cite{cheng2020gaussianmix,he2022elic}. Benchmark studies report that learned codecs can approach or exceed conventional reference codecs under conventional rate-distortion evaluation \cite{hu2022benchmark}.

The above codecs often serve as the starting point for perceptual and generative compression studies, but they are primarily optimized for rate-distortion objectives and are not specifically targeted at perceptual reconstruction quality at very low rates. When the bit budget is insufficient, minimizing distortion alone tends to favor over-smoothed reconstructions and loss of high-frequency texture. This limitation has motivated perceptual and generative compression methods that improve visual reconstruction while keeping distortion competitive, rather than focusing only on entropy modeling or transform design.

\subsection{Perceptual and Generative Image Compression}
Perceptual and generative image compression methods are designed to improve visual quality when the bit budget is very limited. Early work on generative compression studied distribution-preserving lossy compression \cite{tschannen2018dplc}. GAN-based extreme learned compression introduced adversarial training for extreme low-bitrate reconstruction \cite{agustsson2019extremegan}, and HiFiC established a high-fidelity generative compression framework based on distortion, perceptual, and adversarial objectives \cite{mentzer2020hific}. Other approaches improve perceptual quality through stronger statistical objectives, as in MS-ILLM \cite{muckley2023illm}, or through highly generative low-bitrate latent modeling, as in Generative Latent Coding \cite{jia2024glc}. A common difficulty in this area is that GAN-based training can introduce structured artifacts, and the relative weights of the losses often depend on bitrate and image content.

Perceptual and generative image compression is commonly framed in rate-distortion-perception terms, with different methods varying mainly in how the perceptual objective is defined. More recently, diffusion-based perceptual codecs have shown strong visual quality at low bitrates, including PerCo \cite{careil2024perco}, conditional diffusion compression \cite{yang2023cdc}, and foundation-diffusion compression \cite{relic2024foundationdiffusion}. Recent efforts improve decoding efficiency through compressed-feature-initialized relay residual diffusion in RDEIC \cite{li2025rdeic} and one-step diffusion in OSDiff \cite{jia2025osdiff}. However, these systems often operate at a fixed training point, can be computationally expensive at decoding time, or emphasize distributional realism even when reference fidelity is relaxed.

\begin{figure*}[!t]
\centering
\includegraphics[width=0.82\textwidth]{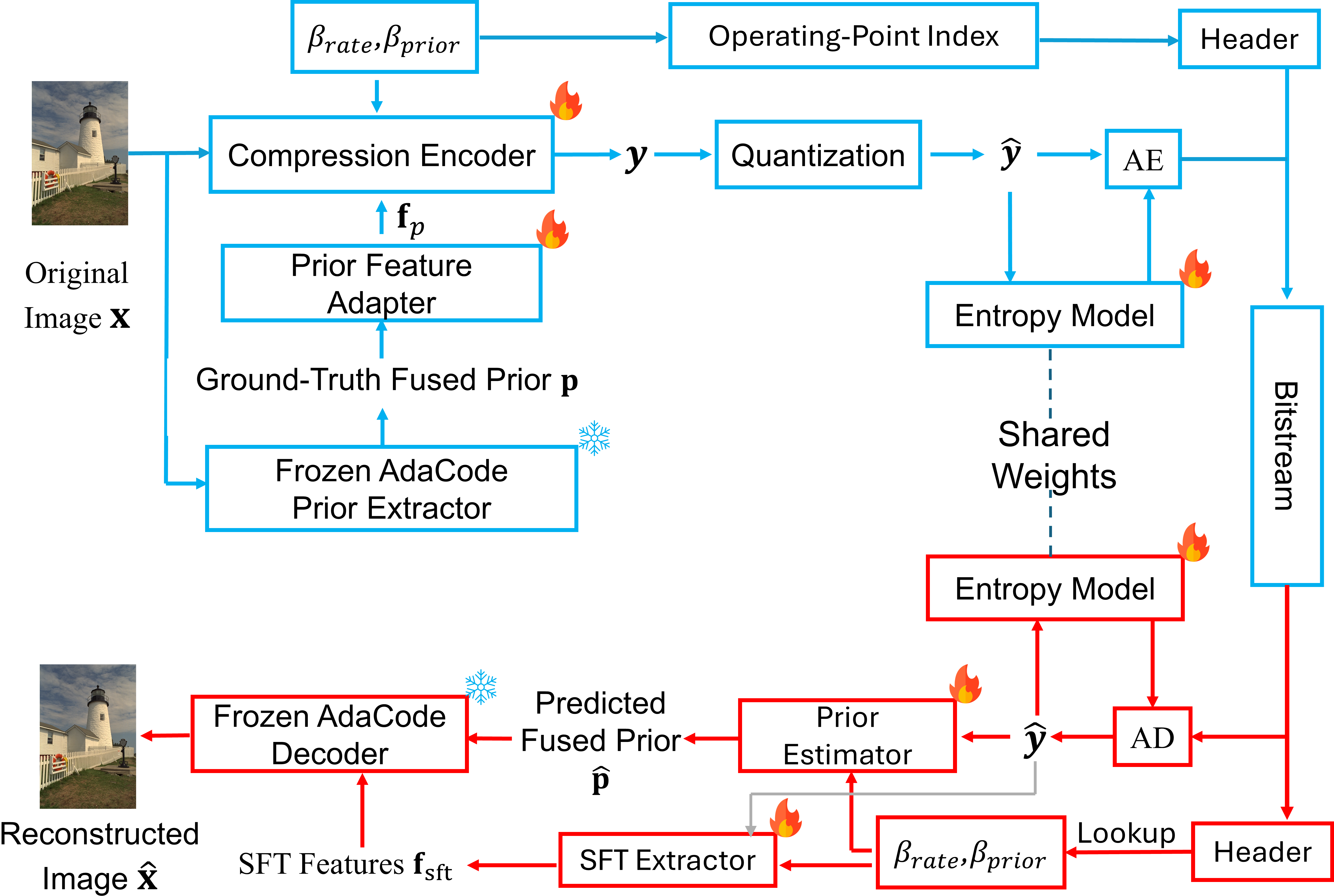}
\caption{Overview of the proposed Adaptive Fused Prior Transfer for Controllable Generative Image Compression (AFP-GIC). Blue denotes the encoding process, and red denotes the decoding process. Snowflake icons denote frozen modules, whereas fire icons denote trainable modules. The selected control pair is conveyed through a compact operating-point index in the header.}
\label{fig:overview}
\end{figure*}

\begin{figure*}[!t]
\centering
\includegraphics[width=0.88\textwidth]{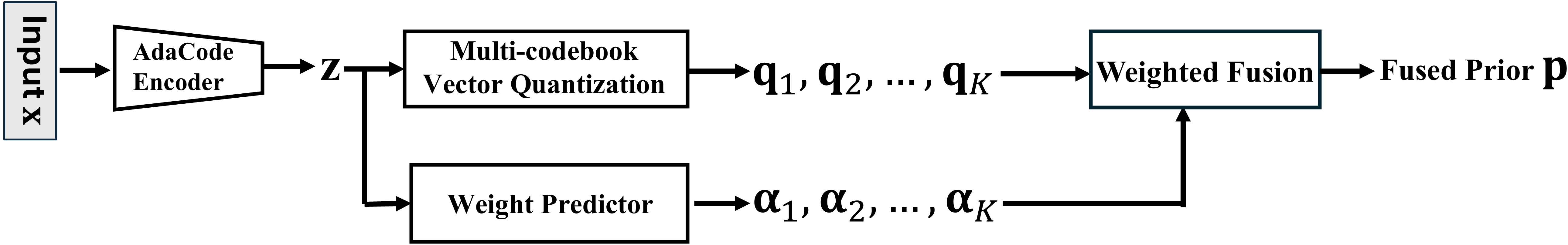}
\caption{Frozen AdaCode prior extractor. The input image is encoded and quantized by multiple codebooks. The resulting codebook representations are then adaptively fused through predicted weights to produce the ground-truth fused prior $\mathbf{p}$.}
\label{fig:adacode_prior_extractor}
\end{figure*}

\subsection{Controllable Image Compression}
Controllable image compression aims to use a single trained model to cover multiple operating preferences, instead of training a separate codec for each bitrate or reconstruction style. In distortion-oriented learned compression, this idea is often realized as variable-rate coding, where the model changes its target bitrate through a conditioning variable or gain mechanism \cite{choi2019conditionalautoencoder,cui2021agdic,song2021sftvr}. For perceptual and generative compression, the control target is broader because the decoder must also balance pixel fidelity and visual realism. Fidelity-controllable extreme compression introduced explicit control over this fidelity-realism behavior in a GAN-based codec \cite{iwai2021fidelitycontrollable}.

Multi-Realism further showed that one compressed representation can be decoded into reconstructions with different realism levels through a decoder-side control variable \cite{agustsson2023multirealism}. The transmitted representation is fixed, and the control signal changes the reconstruction behavior at the decoder. CRDR extends the controllable formulation by jointly considering rate, distortion, and realism, with the goal of using one model across a wider set of compression conditions \cite{iwai2024crdr}. Control-GIC also follows this once-for-all direction for controllable generative image compression with dynamic granularity adaptation \cite{li2025controlgic}. These methods establish controllability as an important practical requirement, but the choice of generative prior and the way it is made available to the decoder remain central design issues.

DC-VIC combines controllable compression with a pretrained generative model and introduces dual-conditioned training to control both the total rate and the allocation of information between token-based generation and feature-level modification \cite{iwai2024dcvic}. This design further shows that bitrate control and reconstruction-behavior control can be coupled in a single low-bitrate generative codec.

\subsection{Pretrained Generative Priors and Codebook-Based Models}
At low bitrates, the transmitted representation may not fully specify local structures and textures. Pretrained generative priors have therefore been studied as decoder-side reconstruction guidance, and codebook-based generative models provide one discrete form of such guidance. VQ-VAE introduced neural discrete representations \cite{oord2017vqvae}, VQ-VAE-2 improved hierarchical discrete latent modeling \cite{razavi2019vqvae2}, and VQGAN combined vector quantization with adversarial and Transformer-based generative modeling for high-resolution synthesis \cite{esser2021vqgan}. In image compression, pretrained VQGAN tokenizers have also been used directly for extreme compression \cite{mao2024vqganextreme}, and DC-VIC demonstrates that a pretrained codebook-based reconstruction prior can be integrated with controllable compression \cite{iwai2024dcvic}.

However, single-codebook generative priors represent each spatial location through one selected codebook entry, which can restrict the available reconstruction-prior representation when image content is diverse. AdaCode learns image-adaptive codebooks for image restoration by combining multiple basis codebooks with spatially varying fusion weights \cite{liu2023adacode}. This mechanism provides a continuous fused-prior representation rather than a single discrete codebook choice. These properties make adaptive codebook fusion relevant to low-bitrate generative compression, although its use under decoder-side information constraints remains less explored.

\section{Methodology}
\label{sec:methodology}

Fig.~\ref{fig:overview} summarizes AFP-GIC. The method combines an entropy-constrained compression backbone with adaptive fused-prior transfer from a frozen AdaCode model. Throughout this paper, unless otherwise specified, ``prior'' denotes the transferred codebook-based generative prior used for reconstruction guidance, while entropy priors denote the probability models used for arithmetic coding. The key asymmetry is that encoder-side prior features are derived from the original image, whereas decoder-side reconstruction must rely on the transmitted compressed representations and compact control/header information, without access to the original image or encoder-side fused prior. AFP-GIC handles this asymmetry through encoder-side prior guidance and decoder-side prior prediction.

\subsection{Overall Framework}
\label{sec:overall_framework}

Given an input image $\mathbf{x}$, the encoder produces a latent representation $\mathbf{y}$, which is quantized to $\hat{\mathbf{y}}$ and entropy coded. AFP-GIC differs from token-prior codecs by using adaptive fused-prior features from a frozen AdaCode model: the encoder injects an adapted fused prior, while the Prior Estimator predicts a compatible fused prior and the SFT extractor separately generates modulation features for the frozen AdaCode decoder \cite{wang2018sft}. This enables prior-guided reconstruction without transmitting the fused prior itself. The proposed prior-transfer mechanism comprises a Prior Feature Adapter and a Prior Estimator; Tables~\ref{tab:prior_adapter_config} and \ref{tab:prior_estimator_config} summarize their configurations.

\begin{table}[t]
\caption{The proposed Prior Feature Adapter configuration.}
\label{tab:prior_adapter_config}
\centering
\scriptsize
\renewcommand{\arraystretch}{1.48}
\setlength{\tabcolsep}{3.0pt}
\resizebox{0.97\columnwidth}{!}{%
\begin{tabular}{|p{0.16\columnwidth}|p{0.23\columnwidth}|p{0.31\columnwidth}|p{0.22\columnwidth}|}
\hline
Block / Module & Layer Type / Stride & (Filter Shape) $\times$ Filters & Output Shape \\
\hline
 \shortstack[l]{\rule{0pt}{2.4ex}input fused\\ prior} & \shortstack[l]{\rule{0pt}{2.4ex}AdaCode fused prior\\ / --} & \shortstack[l]{\rule{0pt}{2.4ex}fused prior feature map at\\ $h_p\times w_p$, 256 channels} & \shortstack[l]{\rule{0pt}{2.4ex}$B\times256\times$\\ $h_p\times w_p$} \\
 \hline
\shortstack[l]{\rule{0pt}{2.4ex}grid\\ alignment} & \shortstack[l]{\rule{0pt}{2.4ex}bilinear resize\\ / --} & \shortstack[l]{\rule{0pt}{2.4ex}resize to encoder\\ prior grid} & \shortstack[l]{\rule{0pt}{2.4ex}$B\times256\times$\\ $\frac{H}{8}\times \frac{W}{8}$} \\
\hline
\shortstack[l]{\rule{0pt}{2.4ex}adapter} & \shortstack[l]{\rule{0pt}{2.4ex}Conv2D\\ / $s=1$} & \shortstack[l]{\rule{0pt}{2.4ex}$(1\times1\times256)$\\ $\times260$} & \shortstack[l]{\rule{0pt}{2.4ex}$B\times260\times$\\ $\frac{H}{8}\times \frac{W}{8}$} \\
\hline
\end{tabular}
}
\renewcommand{\arraystretch}{1.0}
\par\smallskip
\parbox{\columnwidth}{\footnotesize\raggedright \(B\) denotes batch size, \(H,W\) the input-image height and width, and \(h_p,w_p\) the input fused-prior height and width.}
\end{table}

\begin{table}[t]
\caption{The proposed Prior Estimator configuration.}
\label{tab:prior_estimator_config}
\centering
\scriptsize
\renewcommand{\arraystretch}{1.4}
\setlength{\tabcolsep}{3.0pt}
\resizebox{0.97\columnwidth}{!}{%
\begin{tabular}{|p{0.16\columnwidth}|p{0.22\columnwidth}|p{0.33\columnwidth}|p{0.21\columnwidth}|}
\hline
Block / Module & Layer Type / Stride & (Filter Shape) $\times$ Filters & Output Shape \\
\hline
\shortstack[l]{\rule{0pt}{2.4ex}input} & \shortstack[l]{\rule{0pt}{2.4ex}decoder feature\\ / --} & \shortstack[l]{\rule{0pt}{2.4ex}decoder feature map at\\ $\frac{H}{8}\times\frac{W}{8}$, 192 channels} & \shortstack[l]{\rule{0pt}{2.4ex}$B\times192\times$\\ $\frac{H}{8}\times \frac{W}{8}$} \\
 \hline
\shortstack[l]{\rule{0pt}{2.4ex}block 0} & \shortstack[l]{\rule{0pt}{2.4ex}Conv2D\\ / $s=1$} & \shortstack[l]{\rule{0pt}{2.4ex}$(3\times3\times192)$\\ $\times384$} & \shortstack[l]{\rule{0pt}{2.4ex}$B\times384\times$\\ $\frac{H}{8}\times \frac{W}{8}$} \\
\hline
\shortstack[l]{\rule{0pt}{2.4ex}blocks 1--8} & \shortstack[l]{\rule{0pt}{2.4ex}ResBlock $\times 8$\\ / $s=1$} & \shortstack[l]{\rule{0pt}{2.4ex}per block:\\ GN(32)$+$SiLU$+$Conv2D\\ $(3\times3\times384)\times384$\\ GN(32)$+$SiLU$+$Conv2D\\ $(3\times3\times384)\times384$\\ residual addition} & \shortstack[l]{\rule{0pt}{2.4ex}$B\times384\times$\\ $\frac{H}{8}\times \frac{W}{8}$} \\
\hline
\shortstack[l]{\rule{0pt}{2.4ex}output head} & \shortstack[l]{\rule{0pt}{2.4ex}Conv2D\\ / $s=1$} & \shortstack[l]{\rule{0pt}{2.4ex}$(3\times3\times384)$\\ $\times256$} & \shortstack[l]{\rule{0pt}{2.4ex}$B\times256\times$\\ $\frac{H}{8}\times \frac{W}{8}$} \\
\hline
\shortstack[l]{\rule{0pt}{2.4ex}skip head} & \shortstack[l]{\rule{0pt}{2.4ex}Conv2D\\ / $s=1$} & \shortstack[l]{\rule{0pt}{2.4ex}$(1\times1\times384)$\\ $\times256$} & \shortstack[l]{\rule{0pt}{2.4ex}$B\times256\times$\\ $\frac{H}{8}\times \frac{W}{8}$} \\
\hline
\shortstack[l]{\rule{0pt}{2.4ex}output fused\\ prior} & \shortstack[l]{\rule{0pt}{2.4ex}element-wise add\\ / --} & \shortstack[l]{\rule{0pt}{2.4ex}output head\\ $+$ skip head} & \shortstack[l]{\rule{0pt}{2.4ex}$B\times256\times$\\ $\frac{H}{8}\times \frac{W}{8}$} \\
\hline
\end{tabular}
}
\renewcommand{\arraystretch}{1.0}

\vspace{0.2em}
\footnotesize{\textit{GN(32)} denotes Group Normalization with 32 groups \cite{wu2018group}, and \textit{SiLU} denotes the Sigmoid Linear Unit activation~\cite{elfwing2018silu}.}
\end{table}

\subsection{Encoder-Side Adaptive Fused-Prior Transfer}
\label{sec:encoder_transfer}

Encoder-side adaptive fused-prior transfer begins with the frozen AdaCode prior extractor shown in Fig.~\ref{fig:adacode_prior_extractor}. Given an input image $\mathbf{x}$, the frozen AdaCode encoder extracts a latent representation \(\mathbf{z} = E_{\mathrm{Ada}}(\mathbf{x})\), where \(E_{\mathrm{Ada}}(\cdot)\) denotes the pretrained AdaCode encoder. AdaCode is trained to accommodate diverse image content through multiple basis codebooks. In the original AdaCode model \cite{liu2023adacode}, fine-grained semantic labels are further merged into five coarse super-classes for basis-codebook diversification during pretraining. Table~\ref{tab:adacode_groups} lists these pretrained super-classes and their representative content. These groups are not imposed as hard constraints at inference time; they expose the codebooks to different structural and textural patterns. AFP-GIC inherits this pretrained prior space and transfers it to the compression model.

\begin{table}[t]
\caption{Five coarse semantic super-classes adopted in AdaCode \cite{liu2023adacode} to diversify the basis codebooks during pretraining.}
\label{tab:adacode_groups}
\centering
\small
\begin{tabular}{ll}
\hline
Super-class & Representative content \\
\hline
Architectures & buildings, facades, structural layouts \\
Indoor objects & furniture, appliances, indoor object regions \\
Natural scenes & vegetation, mountains, water, sky \\
Street views & roads, vehicles, urban outdoor scenes \\
Portraits & faces, hair, skin, person-centered regions \\
\hline
\end{tabular}
\end{table}

Although these super-classes are introduced in AdaCode for basis-codebook diversification, our codec does not rely on explicit category labels at coding time. Instead, AFP-GIC transfers the fused prior extracted by the frozen AdaCode model and adapts its influence through the encoder and decoder pathways described below. Instead of quantizing \(\mathbf{z}\) with a single codebook, AdaCode uses \(K\) codebooks. For the \(i\)-th codebook \(\mathcal{C}_{i}=\{\mathbf{c}_{i,m}\}_{m=1}^{M_i}\), where \(M_i\) is the number of codewords, quantization is defined as
\begin{equation}
\mathbf{q}_{i}(u,v)=\arg\min_{\mathbf{c}\in\mathcal{C}_{i}}
\left\|\mathbf{z}(u,v)-\mathbf{c}\right\|_{2},
\end{equation}
where \((u,v)\) indexes the spatial location in the latent map. Equivalently, the quantization operator \(Q_i(\cdot)\) produces the \(i\)-th quantized representation
\begin{equation}
\mathbf{q}_{i} = Q_{i}(\mathbf{z}), \qquad i=1,\ldots,K,
\end{equation}
where \(Q_{i}(\cdot)\) denotes quantization with the \(i\)-th codebook. In parallel, a weight predictor estimates spatially varying fusion-weight maps \(\{\alphabf_{1},\alphabf_{2},\ldots,\alphabf_{K}\}\), where \(\alphabf_{i}(u,v)\) denotes the fusion weight assigned to the \(i\)-th codebook at spatial location \((u,v)\), and \(\sum_{i=1}^{K}\alphabf_{i}(u,v)=1\) for each \((u,v)\). The adaptive fused prior is then written as
\begin{equation}
\mathbf{p}(u,v) = \sum_{i=1}^{K}\alphabf_{i}(u,v)\mathbf{q}_{i}(u,v),
\label{eq:fused_prior}
\end{equation}
where \(\mathbf{p}\) denotes the ground-truth fused prior extracted from the input image by the frozen AdaCode prior extractor.

At each spatial location, a single codebook restricts the prior to one selected entry. By predicting fusion weights and combining multiple codebook branches, the weight predictor forms a continuous fused representation that can cover a broader set of local structural and textural patterns.

For the compression task studied here, this construction replaces a fixed single-codebook prior with an image-adaptive prior \(\mathbf{p}\) assembled from multiple codebooks. In AFP-GIC, the ground-truth fused prior \(\mathbf{p}\) is not transmitted to the decoder. Instead, it is used as encoder-side guidance and passed through the prior feature adapter before entering the compression pathway. The encoder is thus encouraged to produce latent variables that remain compatible with the adaptive fused-prior space defined by the frozen AdaCode model.

The prior feature adapter maps the ground-truth fused prior into a feature space suitable for the compression encoder, i.e., \(\mathbf{f}_{p}=A_{\phi}(\mathbf{p})\), where \(A_{\phi}(\cdot)\) denotes the prior feature adapter and \(\mathbf{f}_{p}\) is the adapted prior feature. Its function is to align the fused prior with the feature space used by the compression backbone through spatial resizing and channel projection. The encoder-side latent representation can then be written as
\begin{equation}
\mathbf{y}=E_{\theta}(\mathbf{x},\mathbf{f}_{p},\beta_{\mathrm{rate}},\beta_{\mathrm{prior}}),
\label{eq:encoder_with_prior}
\end{equation}
where \(E_{\theta}(\cdot)\) denotes the compression encoder, and \(\beta_{\mathrm{rate}}\) and \(\beta_{\mathrm{prior}}\) are the two control variables. After the third encoder downsampling layer, \(\mathbf{f}_p\) is concatenated channel-wise with the encoder features at \(H/8\times W/8\), followed by a \(3\times3\) projection and residual addition. The adapter retains the backbone's 260-channel prior-input width, originally comprising 256 one-hot channels and four codeword channels.

\subsection{Decoder-Side Prior Prediction and Guided Reconstruction}
\label{sec:decoder_prediction}

The decoder operates under stricter informational constraints. During reconstruction, it lacks access to the original image and therefore cannot directly use the ground-truth fused prior \(\mathbf{p}\) in \EqLink{eq:fused_prior}. Instead, it must infer a compatible prior representation from the quantized latent \(\hat{\mathbf{y}}\). AFP-GIC therefore introduces a prior estimator that predicts a decoder-side adaptive fused prior
\begin{equation}
\hat{\mathbf{p}} = P_{\psi}(\hat{\mathbf{y}},\beta_{\mathrm{rate}},\beta_{\mathrm{prior}}),
\label{eq:prior_estimator}
\end{equation}
where \(P_{\psi}(\cdot)\) denotes the prior estimator and \(\hat{\mathbf{p}}\) is the predicted fused prior. The estimator produces a prior representation in the AdaCode-compatible feature space, which is then used as the decoder-side prior for reconstruction.

In the residual blocks of the prior estimator (Table~\ref{tab:prior_estimator_config}, blocks 1--8), each block adopts Group Normalization~\cite{wu2018group} and SiLU~\cite{elfwing2018silu} activation. For a given sample, with the batch index omitted, Group Normalization is written as
\begin{equation}
\bigl[\mathrm{GN}(\mathbf{h})\bigr]_{c,u,v}
=
\gamma_c
\frac{h_{c,u,v}-\mu_{g(c)}}{\sqrt{\sigma_{g(c)}^2+\epsilon}}
+\kappa_c,
\label{eq:group_norm}
\end{equation}
where \(g(c)\) denotes the group containing channel \(c\), and \(\mu_{g(c)}\) and \(\sigma_{g(c)}^2\) are the mean and variance computed over all channels in that group and all spatial locations. The learnable channel-wise affine parameters \(\gamma_c\) and \(\kappa_c\) are shared across spatial locations. Here, \(\mathrm{GN}(32)\) indicates that the channels are divided into 32 groups. The SiLU activation used after normalization is defined as \(\mathrm{SiLU}(z)=z\,\sigma(z)\), where \(\sigma(\cdot)\) denotes the sigmoid function.

Once \(\hat{\mathbf{p}}\) is obtained, the decoder still requires controllable feature modulation so that reconstruction can respond to the selected control pair. To this end, an SFT extractor generates decoder conditioning features \(\mathbf{f}_{\mathrm{sft}} = S_{\nu}(\hat{\mathbf{y}},\beta_{\mathrm{rate}},\beta_{\mathrm{prior}})\), where \(S_{\nu}(\cdot)\) denotes the Spatial Feature Transform (SFT) extractor \cite{wang2018sft} and \(\mathbf{f}_{\mathrm{sft}}\) denotes the resulting modulation features. In SFT-based conditioning, an intermediate decoder feature \(\mathbf{h}\) is modulated by spatially varying affine parameters,
\begin{equation}
\mathrm{SFT}(\mathbf{h}\,|\,\boldsymbol{\gamma},\boldsymbol{\delta})
= \boldsymbol{\gamma}\odot\mathbf{h}+\boldsymbol{\delta},
\label{eq:sft}
\end{equation}
where \(\boldsymbol{\gamma}\) and \(\boldsymbol{\delta}\) are learned scale and shift terms, and \(\odot\) denotes element-wise multiplication. In AFP-GIC, these modulation parameters are generated from \(\mathbf{f}_{\mathrm{sft}}\), while the predicted fused prior \(\hat{\mathbf{p}}\) is fed to the frozen AdaCode decoder as the decoder input prior.

The final reconstruction is then written as
\begin{equation}
\hat{\mathbf{x}} = D_{\mathrm{Ada}}(\hat{\mathbf{p}},\mathbf{f}_{\mathrm{sft}}),
\label{eq:decoder_recon}
\end{equation}
where \(D_{\mathrm{Ada}}(\cdot)\) denotes the frozen AdaCode decoder. \EqRangeLink{eq:prior_estimator}{eq:decoder_recon} summarize the decoder-side reconstruction process: the Prior Estimator predicts \(\hat{\mathbf{p}}\) from \(\hat{\mathbf{y}}\), while a separate SFT extractor produces \(\mathbf{f}_{\mathrm{sft}}\); both are supplied to the frozen AdaCode decoder.

Because the fused prior itself is not transmitted, this decoder-side construction keeps the entropy-coded bitstream based on the compressed latents and compact header information while still allowing reconstruction to use a pretrained adaptive fused-prior space.

\subsection{Design Considerations of Prior Transfer}
\label{sec:design_considerations}

The encoder-side and decoder-side prior modules are deliberately asymmetric. At the encoder, the fused prior is extracted from the original image itself, so the main requirement is feature alignment rather than prior generation. Here, feature alignment refers to resizing the fused prior to the encoder feature grid and projecting it into the encoder channel space. For this reason, the Prior Feature Adapter is kept lightweight: it resizes the fused prior to the encoder feature grid and applies a channel projection before injection into the compression backbone. The lightweight adapter is therefore chosen to target the observed scale and channel mismatch while keeping the encoder-side prior-transfer module small.

The decoder faces a different problem. Because the original image and the encoder-side ground-truth fused prior \(\mathbf{p}\) are unavailable at the decoder, the prior estimator must infer a compatible prior representation from the reconstructed compressed representation and the selected control variables. AFP-GIC therefore uses a dedicated Prior Estimator. The residual blocks of the prior estimator (Table~\ref{tab:prior_estimator_config}, blocks 1--8) refine the decoder feature into a prior-compatible representation, while the output skip branch preserves coarse information from the projected input feature.

This asymmetry motivates the prior-consistency term, while the lightweight decoder is chosen to keep the prior-guided reconstruction path efficient. Because the ground-truth fused prior \(\mathbf{p}\) and the decoder-side predicted fused prior \(\hat{\mathbf{p}}\) stem from disparate information conditions, they must be explicitly aligned to remain compatible with the shared, frozen generative decoder. Enforcing this consistency reduces the mismatch between encoder-side guidance and decoder-side prediction. The continuous fused-prior pathway also separates prior alignment from hard single-codebook selection, allowing the prior estimator to regress a fused feature representation rather than select a single discrete branch. This pathway is implemented with a fully convolutional decoder. Compared with the Transformer-based decoder used in DC-VIC, this lightweight design reduces parameters and decoding latency in the benchmark reported in Section~\ref{sec:complexity}.

\subsection{Dual-Control Formulation}
\label{sec:dual_control_formulation}

AFP-GIC is controlled by two variables, \(\beta_{\mathrm{rate}}\) and \(\beta_{\mathrm{prior}}\). The first governs bitrate-oriented behavior, whereas the second governs prior-oriented reconstruction behavior associated with adaptive fused-prior prediction and guided decoding.

The two variables are not used directly as raw scalars. Instead, each variable is first mapped to a Fourier-based embedding \cite{tancik2020fourier}, i.e., \(\mathbf{e}_{r}=\Phi_{r}(\beta_{\mathrm{rate}})\) and \(\mathbf{e}_{p}=\Phi_{p}(\beta_{\mathrm{prior}})\), and the two embeddings are then combined through a lightweight multi-layer perceptron (MLP):
\begin{equation}
\mathbf{e}_{\beta}=M_{\eta}\bigl([\mathbf{e}_{r},\mathbf{e}_{p}]\bigr),
\label{eq:beta_embedding}
\end{equation}
where \(\Phi_{r}(\cdot)\) and \(\Phi_{p}(\cdot)\) denote the two Fourier feature mappings, \([\cdot,\cdot]\) denotes concatenation, and \(M_{\eta}(\cdot)\) denotes the embedding MLP. The resulting control feature \(\mathbf{e}_{\beta}\) modulates both encoder-side and decoder-side processing. Conditioning is therefore represented through learned feature embeddings rather than fixed scalar gating.

The two control variables influence both encoding and decoding. At the encoder, they act together with the adapted prior feature in \EqLink{eq:encoder_with_prior}, ensuring that latent formation reflects both bitrate preference and prior-oriented guidance. At the decoder, the same pair conditions prior prediction in \EqLink{eq:prior_estimator} and the SFT extractor that produces the modulation features used in \EqLink{eq:sft}. Bitrate control and prior control therefore act jointly on the codec, rather than being separated into distinct encoder-side and decoder-side roles.

\subsection{Quantization and Entropy Coding}
\label{sec:quant_entropy}

Except for the proposed prior-transfer modules, AFP-GIC follows a standard hyperprior-based quantization and entropy-coding pipeline \cite{balle2018hyperprior,minnen2018joint,he2022elic}. Given an input image \(\mathbf{x}\), the compression encoder produces a latent representation \(\mathbf{y}=E_{\theta}(\mathbf{x},\cdots)\), and a hyper-encoder extracts the corresponding hyper-latent
\begin{equation}
\mathbf{z}_{h}=H_{\mathrm{enc}}(\mathbf{y}),
\label{eq:hyperlatent}
\end{equation}
where \(\mathbf{z}_{h}\) denotes the hyper-latent and \(H_{\mathrm{enc}}(\cdot)\) denotes the hyper-encoder. During training, hard rounding is approximated by a straight-through estimator \cite{bengio2013estimating}, while at test time both \(\mathbf{z}_{h}\) and \(\mathbf{y}\) are discretized by rounding and encoded with arithmetic coding. Here, \(\mathbf{z}_{h}\) is modeled by an entropy bottleneck \cite{balle2018hyperprior,minnen2018joint}, whereas \(\mathbf{y}\) is modeled slice-wise by a conditional Gaussian distribution whose parameters are predicted from the decoded hyperprior and previously reconstructed slices \cite{minnen2018joint,he2022elic}.

For the hyper-latent, the quantized value is written as
\begin{equation}
\hat{\mathbf{z}}_{h}
=
\operatorname{round}\!\left(\mathbf{z}_{h}-\mathbf{m}_{z}\right)
+ \mathbf{m}_{z},
\label{eq:z_quant}
\end{equation}
where \(\mathbf{m}_{z}\) denotes the learned median used by the entropy bottleneck. The main latent \(\mathbf{y}\) is divided channel-wise into \(S\) slices, \(\mathbf{y}=\{\mathbf{y}^{(s)}\}_{s=1}^{S}\). For the \(s\)-th slice, the hyper-decoder and slice-wise context transforms predict a mean \(\boldsymbol{\mu}^{(s)}\) and a scale \(\boldsymbol{\sigma}^{(s)}\), and quantization is expressed as
\begin{equation}
\hat{\mathbf{y}}^{(s)}
=
\operatorname{round}\!\left(\mathbf{y}^{(s)}-\boldsymbol{\mu}^{(s)}\right)
+ \boldsymbol{\mu}^{(s)}.
\label{eq:y_quant}
\end{equation}

For entropy coding, \(\hat{\mathbf{z}}_{h}\) is modeled by the entropy bottleneck, while \(\hat{\mathbf{y}}\) is modeled conditionally. After decoding \(\hat{\mathbf{z}}_{h}\), the hyper-decoder \(H_{\mathrm{dec}}(\cdot)\) produces hyper features
\begin{equation}
\mathbf{h}_{z}=H_{\mathrm{dec}}(\hat{\mathbf{z}}_{h}),
\label{eq:hyperdecode}
\end{equation}
from which the slice-wise mean and scale parameters are estimated recursively,
\begin{equation}
\bigl(\boldsymbol{\mu}^{(s)},\boldsymbol{\sigma}^{(s)}\bigr)
=
G_{s}\!\left(\mathbf{h}_{z},\hat{\mathbf{y}}^{(<s)}\right),
\label{eq:gaussian_params}
\end{equation}
where \(G_s(\cdot)\) denotes the context-dependent parameter predictor and \(\hat{\mathbf{y}}^{(<s)}\) collects the previously reconstructed slices. As in prior LIC models, the discrete symbol probabilities are obtained by integrating the corresponding continuous densities over unit-width quantization bins \cite{balle2018hyperprior,minnen2018joint,he2022elic,mentzer2018conditional}. The bitrate term used in optimization is then computed from the negative log-likelihoods of the quantized hyper-latent and main latent,
\begin{equation}
R=
\frac{
\displaystyle -\sum_{i}\log_{2} p(\hat{z}_{h,i})
\displaystyle -\sum_{s=1}^{S}\sum_{i}
\log_{2} p\!\left(\hat{y}^{(s)}_{i}\mid\hat{\mathbf{z}}_{h},\hat{\mathbf{y}}^{(<s)}\right)
}{HW},
\label{eq:total_rate}
\end{equation}
where \(H\) and \(W\) are the image height and width. The selected control pair is conveyed through a compact header index, which is excluded from the training-time rate term \(R\). The adaptive fused-prior mechanism proposed in this paper is built on top of this quantization and entropy-coding pipeline rather than replacing it.

\subsection{Objective Functions and Discriminator}
\label{sec:loss_and_discriminator}

During non-adversarial dual-conditioned training, the same two variables also determine the relative emphasis of the rate term and the prior-consistency term. We parameterize the corresponding weights exponentially,
\begin{equation}
w_{r}=\exp(\beta_{\mathrm{rate}}), \qquad
w_{p}=\exp(\beta_{\mathrm{prior}}).
\label{eq:beta_weights}
\end{equation}
The exponential parameterization has two practical consequences. First, \(w_r\) and \(w_p\) remain strictly positive for all control values, so the corresponding loss terms retain their intended roles. Second, the emphasis assigned to the two terms changes smoothly and monotonically as \(\beta_{\mathrm{rate}}\) or \(\beta_{\mathrm{prior}}\) varies. These weights rescale the rate-related and prior-related terms during training, so that the objective changes consistently with the selected control variables.

The non-adversarial generator objective of AFP-GIC contains four components: a rate term, an image distortion term, a perceptual term, and a prior-consistency term. Let \(R\) denote the bitrate estimated from the entropy model, let \(D(\mathbf{x},\hat{\mathbf{x}})\) denote the distortion between the input image and the reconstruction, let \(P(\mathbf{x},\hat{\mathbf{x}})\) denote the perceptual loss, and let \(\mathcal{L}_{\mathrm{prior}}=\left\|\hat{\mathbf{p}}-\mathbf{p}\right\|_2^2\) denote the prior-consistency term that keeps the predicted fused prior close to the ground-truth fused prior. With the dual-control weights in \EqLink{eq:beta_weights}, the generator objective used during non-adversarial training is written as
\begin{equation}
\mathcal{L}_{G}^{\mathrm{base}}
=
w_r \lambda_R R
 + \lambda_D D(\mathbf{x},\hat{\mathbf{x}})
 + \lambda_P P(\mathbf{x},\hat{\mathbf{x}})
 + w_p \lambda_{\mathrm{prior}} \mathcal{L}_{\mathrm{prior}},
\label{eq:base_generator_loss}
\end{equation}
where \(\lambda_R\), \(\lambda_D\), \(\lambda_P\), and \(\lambda_{\mathrm{prior}}\) are scalar coefficients. Here, \(D(\cdot,\cdot)\) denotes the pixel-domain mean squared error (MSE) distortion term, and \(P(\cdot,\cdot)\) is the Learned Perceptual Image Patch Similarity (LPIPS)-based perceptual term computed with the AlexNet backbone. For consistency, the reported LPIPS metric in the experiments is also computed with the AlexNet backbone. In \EquationLink{eq:base_generator_loss}, \(\beta_{\mathrm{rate}}\) adjusts the emphasis on coding cost through \(w_r\), whereas \(\beta_{\mathrm{prior}}\) adjusts the emphasis on prior alignment through \(w_p\).

After the initial non-adversarial period, adversarial supervision is introduced to improve perceptual realism at low bitrates. During adversarial refinement, the compression encoder and entropy models are fixed, and the rate term is omitted. The overall generator objective during adversarial refinement is
\begin{equation}
\mathcal{L}_{G}^{\mathrm{adv}}
=
\lambda_D D(\mathbf{x},\hat{\mathbf{x}})
+\lambda_P P(\mathbf{x},\hat{\mathbf{x}})
+\lambda_{\mathrm{adv}}\mathcal{L}_{\mathrm{adv}}
+\lambda_{\mathrm{prior}}\mathcal{L}_{\mathrm{prior}},
\label{eq:adversarial_generator_loss}
\end{equation}
where the generator-side adversarial loss is
\begin{equation}
\mathcal{L}_{\mathrm{adv}}
=
-\left\langle \log \sigma\!\left(D_{\xi}(\hat{\mathbf{x}};\beta_{\mathrm{rate}},\beta_{\mathrm{prior}})\right)\right\rangle,
\label{eq:adv_loss}
\end{equation}
where \(\sigma(\cdot)\) denotes the sigmoid function and \(\langle\cdot\rangle\) denotes averaging over the PatchGAN outputs. For adversarial training, we use a conditional PatchGAN discriminator \cite{isola2017pix2pix}. Instead of assigning one global real/fake score to the whole image, a patch discriminator evaluates local image regions, matching the local nature of many low-bitrate reconstruction artifacts, texture inconsistencies, and unnatural high-frequency patterns. The control pair is Fourier-encoded and mapped by an MLP to eight channels, which are spatially broadcast and concatenated with the RGB input before the first PatchGAN convolution. We denote the conditioned discriminator by \(D_{\xi}(\cdot;\beta_{\mathrm{rate}},\beta_{\mathrm{prior}})\). Its objective is
\begin{equation}
\begin{aligned}
\mathcal{L}_{D}
&=
-\frac{1}{2}\left\langle \log \sigma\!\left(D_{\xi}(\mathbf{x};\beta_{\mathrm{rate}},\beta_{\mathrm{prior}})\right)\right\rangle
\\
&\quad
-\frac{1}{2}\left\langle \log \left(1-\sigma\!\left(D_{\xi}(\hat{\mathbf{x}};\beta_{\mathrm{rate}},\beta_{\mathrm{prior}})\right)\right)\right\rangle.
\end{aligned}
\label{eq:disc_loss}
\end{equation}
The discriminator is introduced only after the initial non-adversarial period, once the codec has already learned stable controllable compression and prior prediction. This ordering allows the model to first establish a compression-reconstruction mapping before adversarial refinement is applied, reducing the risk of destabilizing the preceding training.

\subsection{Training Strategy}
\label{sec:training_strategy}

Directly optimizing AFP-GIC for a single, fixed operating point would undermine the flexibility required for controllable compression. Following the general controllable-training strategy used in DC-VIC \cite{iwai2024dcvic}, the model is first trained over sampled control pairs, then a small set of bitrate-specific operating points is selected on a validation set, and finally the model is fine-tuned on the selected operating points. In AFP-GIC, this schedule is coupled with adaptive fused-prior transfer and decoder-side prior alignment.

\vspace{0.5em}
\hypertarget{stage-training-i}{}%
\noindent\textbf{Stage I:} DUAL-CONDITIONED TRAINING WITH SAMPLED CONTROL PAIRS

In Stage~I, the model is trained under sampled pairs \((\beta_{\mathrm{rate}},\beta_{\mathrm{prior}})\), with the two variables uniformly drawn from discrete grids over \([0,\beta_{\mathrm{rate,max}}]\) and \([0,\beta_{\mathrm{prior,max}}]\). This exposes the model to a wide range of bitrate-prior trade-offs instead of concentrating training on one pair. Stage~I begins with a \(1.0\)M-iteration non-adversarial period optimized with the base objective in \EqLink{eq:base_generator_loss}, followed by a \(500\)K adversarial-refinement segment optimized with the adversarial objective in \EqLink{eq:adversarial_generator_loss}.

\vspace{0.5em}
\hypertarget{stage-training-ii}{}%
\noindent\textbf{Stage II:} VALIDATION-BASED BETA SELECTION

After Stage~I, the model can respond to many control pairs, but evaluation still requires a small number of concrete bitrate-specific operating points. In the following, an operating point refers to a selected control pair \( (\beta_{\mathrm{rate}},\beta_{\mathrm{prior}}) \) and its realized average bpp on the corresponding evaluation set. Stage~II selects one beta pair for each target bitrate. Let \(r_t\) denote a target bitrate. The selection procedure is as follows:

\begin{enumerate}
\item We first define a candidate set of prior-oriented control values \(\{\beta_{\mathrm{prior}}^{(1)},\beta_{\mathrm{prior}}^{(2)},\ldots,\beta_{\mathrm{prior}}^{(M)}\}\).

\item In Stage~II, \(\beta_{\mathrm{prior}}\) is evaluated on a uniform grid over \([0.25, 3.5]\) with interval \(0.25\). For each grid point and target bitrate, \(\beta_{\mathrm{rate}}\) is selected by binary search to match the target validation bitrate as closely as possible. More precisely, for each candidate \(\beta_{\mathrm{prior}}^{(k)}\), we search for a corresponding \(\beta_{\mathrm{rate}}\) whose average validation bitrate is closest to the target bitrate \(r_t\):
\begin{equation}
\beta_{\mathrm{rate}}^{\star}(\beta_{\mathrm{prior}},r_t)
=
\arg\min_{\beta_{\mathrm{rate}}}
\left|
\bar{R}(\beta_{\mathrm{rate}},\beta_{\mathrm{prior}})-r_t
\right|,
\label{eq:beta_rate_select}
\end{equation}
where \(\bar{R}(\cdot)\) denotes the average validation bitrate. This step yields a set of candidate pairs whose achieved bitrate is close to \(r_t\).

\item For each candidate pair, we generate reconstructions on the validation set and rank them using a composite score that balances reference fidelity and distribution-level perceptual realism measured by the Fr\'{e}chet Inception Distance (FID). The ranking score is
\begin{equation}
J_{\mathrm{sel}} = \alpha \cdot \mathrm{PSNR} - \mathrm{FID},
\label{eq:beta_score}
\end{equation}
where a larger \(J_{\mathrm{sel}}\) indicates a better validation-time trade-off between reference fidelity and perceptual realism. In practice, this score is used only to rank a small number of candidate beta pairs on the validation set: PSNR keeps the selected pair near the desired reference-fidelity level, while FID serves as a coarse realism-oriented screening signal. FID is therefore used only for validation-time beta-pair selection and is not part of the training objective itself. Here, this screening FID is computed on the same 2000-image Open Images validation subset using the local patch-based pipeline: following HiFiC \cite{mentzer2020hific}, we extract \(256\times256\) patches from both real and reconstructed images, add a half-patch spatial shift, and compute FID on the resulting patch sets.

\item The selected beta pair for \(r_t\) is the candidate with the highest validation score under \EqLink{eq:beta_score}.
\end{enumerate}

Repeating this procedure for all target bitrates yields a compact set of bitrate-specific operating points used in Stage~III and in the main experiments. We do not treat FID as a primary reporting metric because it is a set-level statistic, is sensitive to sample size and evaluation protocol, and is less suitable for judging image-wise fidelity in controllable low-bitrate reconstruction.

\vspace{0.5em}
\hypertarget{stage-training-iii}{}%
\noindent\textbf{Stage III:} SELECTED-PAIR FINE-TUNING

Stage~III no longer samples from the full beta space. Instead, training is restricted to the compact selected-pair set obtained from Stage~II, where \(\mathcal{B}_{\mathrm{sel}}\) denotes the set of selected beta pairs. Each iteration draws a pair uniformly from \(\mathcal{B}_{\mathrm{sel}}\). In Stage~III, training continues with the adversarial objective in \EqLink{eq:adversarial_generator_loss}, retaining adversarial and prior-consistency supervision while the rate term remains omitted. Consequently, the Stage~III model fine-tuned on \(\mathcal{B}_{\mathrm{sel}}\) is evaluated as a five-point variable-rate codec.

\subsection{Motivating Analysis}
\label{sec:theoretical_justification}

Compared with controllable codecs that rely on a single-codebook prior, such as DC-VIC, AFP-GIC adopts adaptive fused-prior transfer. The following analysis supports two structural points relevant to our design: first, the reconstruction-error upper bound decreases when the decoder-side prior is better aligned with the ideal prior for the current sample; second, the adaptive fused-prior family contains the single-branch family and is strictly richer when the branch features are non-degenerate.

\begin{proposition}[Prior alignment bound]
\label{prop:prior_alignment_bound}
Let the reconstruction be written as \(\hat{\mathbf{x}} = f(\hat{\mathbf{y}},\hat{\mathbf{p}})\), where \(\hat{\mathbf{y}}\) is the quantized compression latent and \(\hat{\mathbf{p}}\) is the predicted fused prior supplied to the decoder. Let \(\mathbf{p}^{\star}\) denote an ideal prior for reconstructing the current image from \(\hat{\mathbf{y}}\). For fixed \(\hat{\mathbf{y}}\), assume that the decoder is Lipschitz continuous with respect to its prior input. That is, for any two decoder-side prior inputs \(\mathbf{q}_1\) and \(\mathbf{q}_2\),
\begin{equation}
\|f(\hat{\mathbf{y}},\mathbf{q}_{1})-f(\hat{\mathbf{y}},\mathbf{q}_{2})\|
\le
L\|\mathbf{q}_{1}-\mathbf{q}_{2}\|,
\label{eq:lipschitz_prior}
\end{equation}
for some constant \(L>0\). Then
\begin{equation}
\|f(\hat{\mathbf{y}},\hat{\mathbf{p}})-\mathbf{x}\|^{2}
\le
2L^{2}\|\hat{\mathbf{p}}-\mathbf{p}^{\star}\|^{2}
+
2\|f(\hat{\mathbf{y}},\mathbf{p}^{\star})-\mathbf{x}\|^{2}.
\label{eq:prior_bound}
\end{equation}
\end{proposition}

\begin{proof}
Starting from the reconstruction error, add and subtract \(f(\hat{\mathbf{y}},\mathbf{p}^{\star})\):
\[
f(\hat{\mathbf{y}},\hat{\mathbf{p}})-\mathbf{x}
=
\bigl(f(\hat{\mathbf{y}},\hat{\mathbf{p}})-f(\hat{\mathbf{y}},\mathbf{p}^{\star})\bigr)
+
\bigl(f(\hat{\mathbf{y}},\mathbf{p}^{\star})-\mathbf{x}\bigr).
\]
Applying the triangle inequality gives
\[
\|f(\hat{\mathbf{y}},\hat{\mathbf{p}})-\mathbf{x}\|
\le
\|f(\hat{\mathbf{y}},\hat{\mathbf{p}})-f(\hat{\mathbf{y}},\mathbf{p}^{\star})\|
+
\|f(\hat{\mathbf{y}},\mathbf{p}^{\star})-\mathbf{x}\|.
\]
Applying \EqLink{eq:lipschitz_prior} with \(\mathbf{q}_1=\hat{\mathbf{p}}\) and \(\mathbf{q}_2=\mathbf{p}^{\star}\) yields
\[
\|f(\hat{\mathbf{y}},\hat{\mathbf{p}})-\mathbf{x}\|
\le
L\|\hat{\mathbf{p}}-\mathbf{p}^{\star}\|
+
\|f(\hat{\mathbf{y}},\mathbf{p}^{\star})-\mathbf{x}\|.
\]
Squaring both sides and using \((a+b)^2\le 2a^2+2b^2\) yields \EqLink{eq:prior_bound}.
\end{proof}

\noindent\textit{Interpretation:} \EquationLink{eq:prior_bound} separates the reconstruction error into a prior-mismatch term, \(2L^{2}\|\hat{\mathbf{p}}-\mathbf{p}^{\star}\|^{2}\), and a residual term under the ideal prior, \(2\|f(\hat{\mathbf{y}},\mathbf{p}^{\star})-\mathbf{x}\|^{2}\). Thus, moving the predicted fused prior closer to the ideal prior tightens the bound. Since \(\mathbf{p}^{\star}\) is unobservable, AFP-GIC uses the ground-truth fused prior \(\mathbf{p}\) extracted from the input image as the supervision target for \(\hat{\mathbf{p}}\). The proposition motivates prior alignment but does not assume that \(\mathbf{p}=\mathbf{p}^{\star}\).

\begin{proposition}[Expressive advantage of adaptive fused priors]
\label{prop:fused_prior_expressive_advantage}
At any spatial location \((u,v)\), let \(\mathbf{p}_{k}(u,v)\) denote the prior feature vector produced by the \(k\)-th codebook branch, and let \(\alphabf(u,v)=[\alphabf_{1}(u,v),\ldots,\alphabf_{K}(u,v)]^{\top}\in\Delta^{K-1}\) denote the local fusion-weight vector, where \(\alphabf_{k}(u,v)\) is its \(k\)-th component. The adaptive fused prior at that location is
\begin{equation}
\mathbf{p}_{A}(u,v)=\sum_{k=1}^{K}\alphabf_{k}(u,v)\mathbf{p}_{k}(u,v).
\label{eq:adaptive_fused_prior_family}
\end{equation}
Let \(\mathbf{p}^{\star}(u,v)\) denote an ideal prior feature vector at \((u,v)\), and let \(\Delta^{K-1}=\{\alphabf\in\mathbb{R}^{K}:\alphabf_k\ge 0,\sum_{k=1}^{K}\alphabf_k=1\}\) denote the probability simplex. Then
\begin{equation}
\begin{aligned}
\phantom{\le}\,\min_{\alphabf(u,v)\in\Delta^{K-1}}
\Bigl\|
\mathbf{p}^{\star}(u,v)-\sum_{k=1}^{K}\alphabf_{k}(u,v)\mathbf{p}_{k}(u,v)
\Bigr\|
\\
\le \min_{k}\,\|\mathbf{p}^{\star}(u,v)-\mathbf{p}_{k}(u,v)\|.
\end{aligned}
\raisetag{4pt}
\label{eq:fused_family_dominance}
\end{equation}
\end{proposition}

\begin{proof}
At any fixed spatial location \((u,v)\), the single-codebook case is recovered when the fusion weights degenerate to a one-hot vector. Specifically, if for some branch index \(j\) we set \(\alphabf_j(u,v)=1\) and \(\alphabf_k(u,v)=0\) for all \(k\neq j\), then \EquationLink{eq:adaptive_fused_prior_family} reduces to \(\mathbf{p}_A(u,v)=\mathbf{p}_j(u,v)\). Hence every single-codebook prior feature at location \((u,v)\) is a feasible member of the adaptive fused-prior family indexed by \(\Delta^{K-1}\). Taking the minimum approximation error over a larger feasible set cannot yield a worse value than taking the minimum over the subset of one-hot choices, which establishes \EqLink{eq:fused_family_dominance}.
\end{proof}

\noindent\textit{Interpretation:} \EquationLink{eq:fused_family_dominance} is a set-inclusion statement. It does not claim that optimization will always find the globally best fused prior, but it does show that the fused formulation contains all single-branch choices and becomes a richer candidate family when the branch features are non-degenerate. Combined with \EqLink{eq:prior_bound}, this supports the design of AFP-GIC: the fused-prior family enlarges the feasible set beyond one-hot single-branch choices, so a closer approximation to an ideal prior may be representable; when the decoder-side prior is better aligned, the reconstruction-error bound becomes tighter.

\section{Experiments}
\label{sec:experiments}

\begin{figure*}[!t]
\centering
\includegraphics[width=0.95\textwidth]{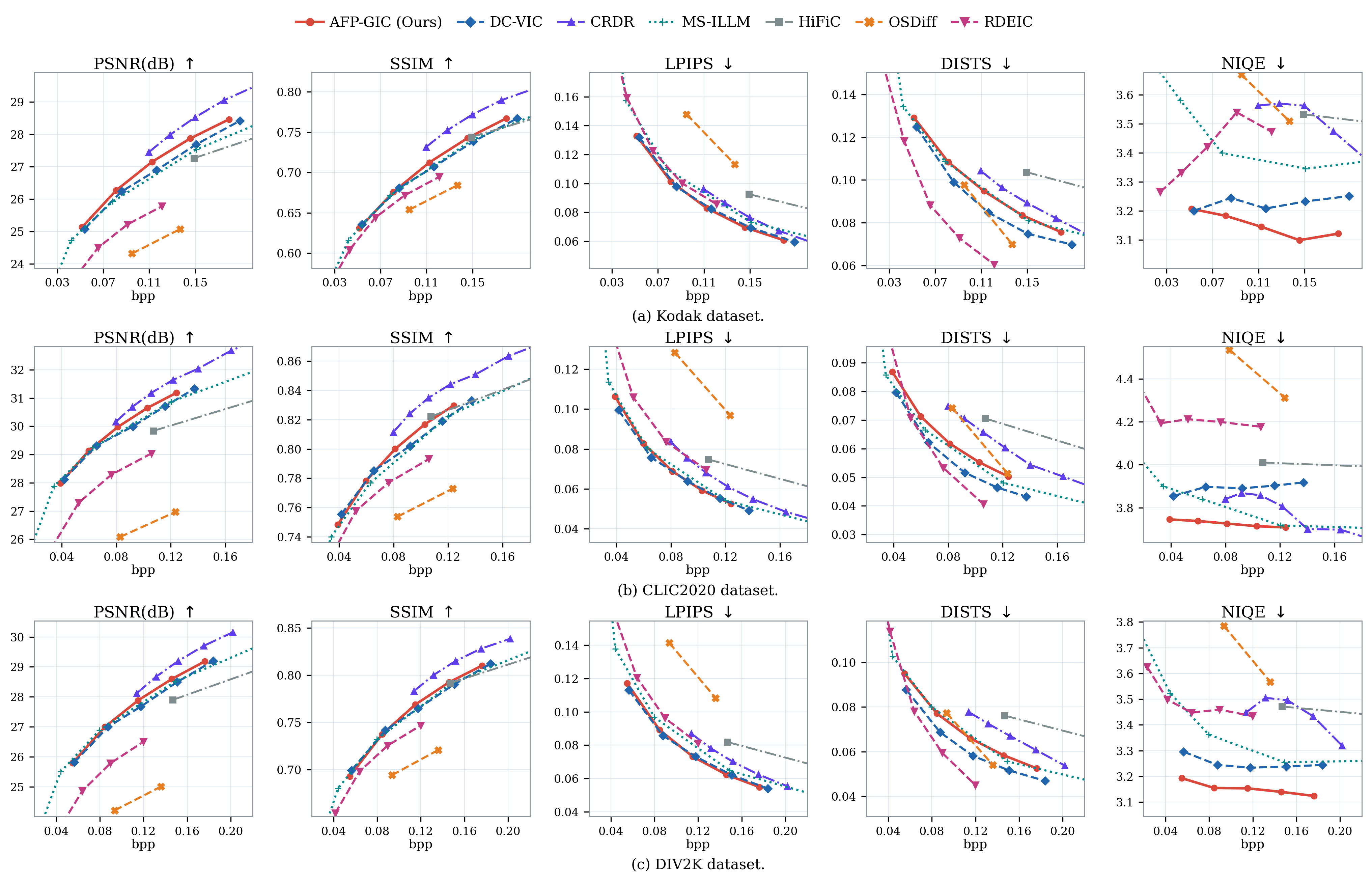}
\caption{Quantitative comparison with learned generative codecs on (a) Kodak, (b) CLIC2020, and (c) DIV2K. The reported metrics are PSNR, SSIM, LPIPS, DISTS, and NIQE. ``\(\uparrow\)'' indicates higher is better; ``\(\downarrow\)'' indicates lower is better.}
\label{fig:quantitative_main}
\end{figure*}

\subsection{Experimental Setup}
\label{sec:experimental_setup}

\noindent\textbf{1) Datasets:}
AFP-GIC is trained on random \(256\times256\) crops sampled from about 1.13 million images in Open Images \cite{kuznetsova2020openimages}. For the validation-time beta-selection in \StageIILink{}, we use a 2000-image subset sampled from the Open Images validation split under the same \(256\times256\) crop setting. Evaluation is conducted on 24 Kodak images \cite{kodakdataset}, 428 CLIC2020 test images \cite{clic2020}, and 100 DIV2K validation images \cite{agustsson2017div2k}.

\vspace{0.5em}
\noindent\textbf{2) Baselines and Comparison Protocol:}
We compare AFP-GIC with two conventional codecs, VVC Intra \cite{bross2021vvc} and BPG \cite{bellard2014bpg}, and six learned generative codecs: DC-VIC \cite{iwai2024dcvic}, CRDR \cite{iwai2024crdr}, MS-ILLM \cite{muckley2023illm}, HiFiC \cite{mentzer2020hific}, OSDiff \cite{jia2025osdiff}, and RDEIC \cite{li2025rdeic}. We use their officially released models, configurations, and code. AFP-GIC is evaluated at the five operating points selected in Section~\ref{sec:training_strategy}. Each point is paired with the closest released baseline bitrate on the same dataset; no baseline is retrained or manually reselected.

For the VVC Intra anchor, we use Fraunhofer's VVC encoder-decoder implementation, with \texttt{vvencapp} (v1.15.0-dev) \cite{VVenC} for encoding and \texttt{vvdecapp} (v3.2.0-dev) \cite{VVdeC} for decoding. In the VVenC version used in our evaluation, YUV 4:2:0 is the only supported input format. For the BPG reference codec, we use the encoder-decoder from the official release \cite{bellard2014bpg} (version 0.9.8) with YUV 4:4:4 chroma sampling, following the default setting used in our evaluation. All decoded images are converted to RGB for metric evaluation.

\vspace{0.5em}
\noindent\textbf{3) Evaluation Metrics:}
We report bits per pixel (bpp), PSNR \cite{wang2009mse}, single-scale structural similarity index measure (SSIM) \cite{wang2004ssim}, LPIPS \cite{zhang2018lpips}, Deep Image Structure and Texture Similarity (DISTS) \cite{ding2022dists}, and Natural Image Quality Evaluator (NIQE) \cite{mittal2013niqe}. PSNR and SSIM measure reference fidelity; LPIPS and DISTS measure full-reference perceptual similarity; NIQE measures no-reference naturalness. Because reference similarity may not fully reflect realism under generative completion \cite{blau2018perceptiondistortion,blau2019rdp}, we interpret these metrics jointly. FID \cite{heusel2017fid} is used for beta-pair selection and reported in Appendix~\ref{app:fid_results}, rather than as a primary metric, because it compares dataset-level distributions and depends on protocol and sample size. The supplementary material provides a taxonomy of the evaluated image-quality metrics, including their strengths and limitations (Section~II and Table~S5), together with SNR and MS-SSIM \cite{wang2003msssim} rate curves, point-wise values, and Bj{\o}ntegaard delta (BD) summaries \cite{bjontegaard2001bdpsnr} at the operating points used in Fig.~\ref{fig:quantitative_main}. BD metrics summarize average curve differences over each method pair's common low-bitrate interval; the same interpolation and integration are used for the five metrics reported in Table~\ref{tab:bd_vs_msillm}.

\vspace{0.5em}
\noindent\textbf{4) Implementation Details:}
Unless otherwise stated, the AdaCode prior extractor and decoder remain frozen throughout training. The trainable components are the compression backbone, prior feature adapter, prior estimator, SFT extractor, and the conditional PatchGAN discriminator. The training batch size is 6. During the initial non-adversarial period of \StageILink{}, which lasts \(1.0\) million iterations, the generator uses Adam \cite{kingma2015adam} with learning rate \(1\times 10^{-4}\), while the entropy-model auxiliary parameters use Adam with learning rate \(1\times 10^{-3}\); gradient clipping is fixed at 1.0. Adversarial refinement is subsequently introduced in two additional phases of 500,000 iterations each: one before \StageIILink{} and the other during \StageIIILink{} selected-pair fine-tuning. In both phases, the generator and discriminator use Adam with learning rate \(1\times 10^{-4}\). For the main AFP-GIC model, the non-adversarial period uses \(\lambda_R = 0.5\), \(\lambda_D = 50\), \(\lambda_P = 1.0\), and \(\lambda_{\mathrm{prior}} = 0.006\). During adversarial refinement, the compression encoder and entropy models are fixed and the rate term is omitted; \(\lambda_D\), \(\lambda_P\), \(\lambda_{\mathrm{adv}}\), and \(\lambda_{\mathrm{prior}}\) are set to \(50\), \(1.0\), \(0.01\), and \(1.0\), respectively.

The dual-control variables are sampled over \(\beta_{\mathrm{rate}}\in[0,3.0]\) and \(\beta_{\mathrm{prior}}\in[0,3.5]\) under the exponential weighting policy in Section~\ref{sec:training_strategy}. In \StageIILink{}, we use the patch-based Inception-v3 FID in \(J_{\mathrm{sel}}=\alpha\cdot\mathrm{PSNR}-\mathrm{FID}\), with \(\alpha=2\). The final reported \StageIIILink{} model adopts the selected pair set \((1.921,2.25)\), \((1.312,3.5)\), \((0.844,3.25)\), \((0.422,0.5)\), and \((0.091,2.25)\), written here in the order \((\beta_{\mathrm{rate}},\beta_{\mathrm{prior}})\) and corresponding approximately to operating points at \(0.05\), \(0.075\), \(0.10\), \(0.125\), and \(0.15\) bpp. Final AFP-GIC results are reported at these selected operating points.

We assess \StageIIILink{} convergence at the approximately 0.05-bpp operating point over the final 100K in-loop validation iterations (11 checkpoints, 10K apart). PSNR is \(25.046\pm0.009\) dB (population standard deviation; range \(25.032\)--\(25.062\) dB), and the bitrate standard deviation is \(4.9\times10^{-7}\) bpp. This narrow variation indicates a stable final plateau and limited sensitivity to checkpoint choice. These within-run values come from the training-time validation pipeline; they do not estimate independent-seed variability or replace the offline results in Fig.~\ref{fig:quantitative_main}.

\subsection{Quantitative Comparison with Representative Codecs}
\label{sec:quantitative_comparison}

The evaluation prioritizes perceptual naturalness across multiple very-low-bitrate operating points using one controllable model, followed by decoder-side efficiency and model footprint; PSNR and structural similarity are retained as fidelity safeguards rather than treated as the sole objectives.

\begin{figure}[!t]
\centering
\includegraphics[width=0.94\columnwidth]{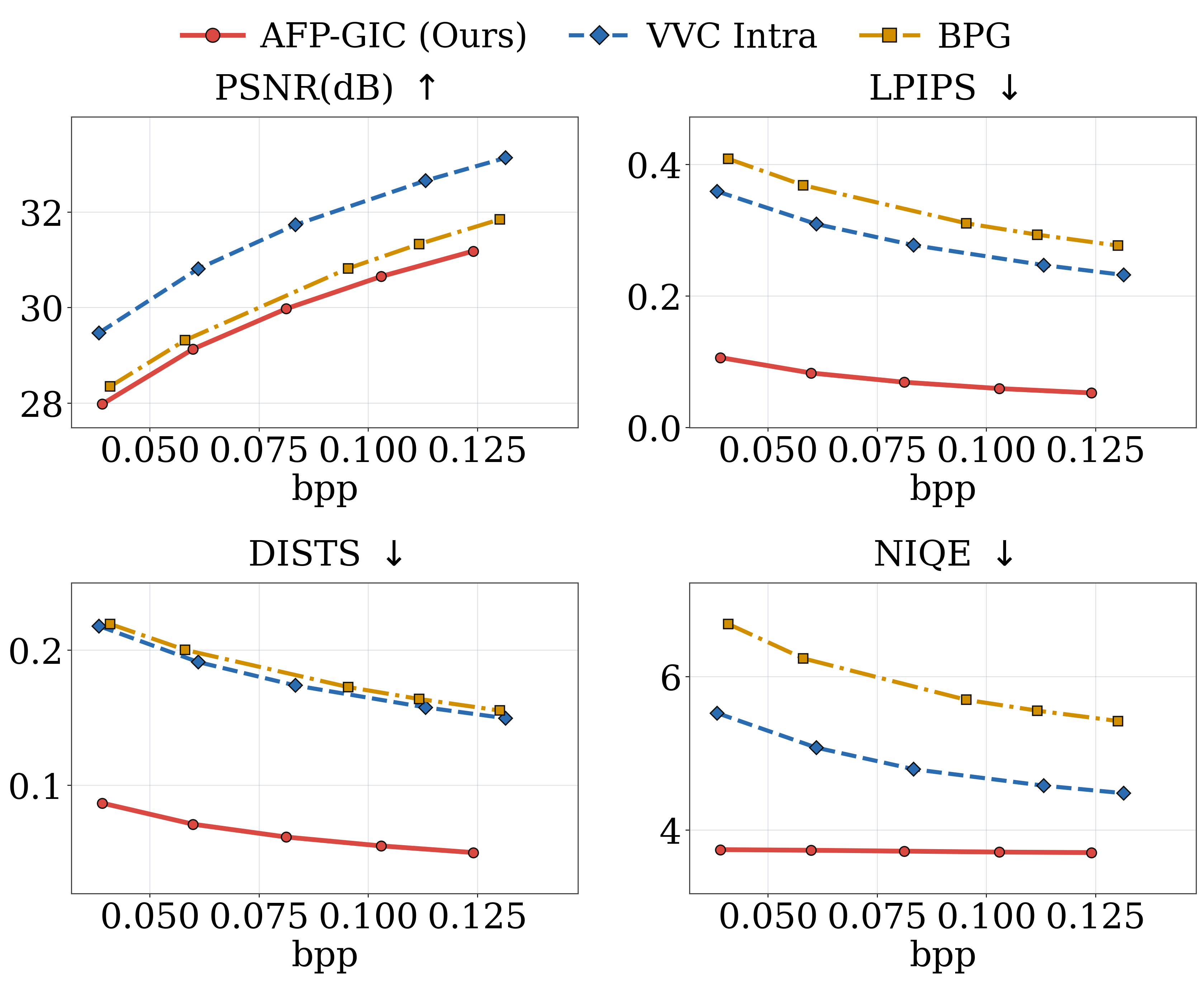}
\caption{Quantitative comparison of AFP-GIC, VVC Intra, and BPG on CLIC2020. ``\(\uparrow\)'' indicates higher is better; ``\(\downarrow\)'' indicates lower is better.}
\label{fig:traditional_clic2020}
\end{figure}

We first compare AFP-GIC with representative learned generative codecs, including DC-VIC \cite{iwai2024dcvic}, CRDR \cite{iwai2024crdr}, MS-ILLM \cite{muckley2023illm}, HiFiC \cite{mentzer2020hific}, OSDiff \cite{jia2025osdiff}, and RDEIC \cite{li2025rdeic}. Fig.~\ref{fig:quantitative_main} summarizes the results on Kodak, CLIC2020, and DIV2K using PSNR, SSIM, LPIPS, DISTS, and NIQE at the closest available released operating points, while Table~\ref{tab:bd_vs_msillm} reports the corresponding Bj{\o}ntegaard-delta summaries relative to MS-ILLM. Appendices~\ref{app:numerical_results} and~\ref{app:control_gic} provide the underlying values and a partial Control-GIC \cite{li2025controlgic} comparison. Table~\ref{tab:bd_vs_msillm} shows negative BD-NIQE relative to MS-ILLM on all three datasets. Relative to DC-VIC, AFP-GIC has higher single-scale BD-SSIM but slightly worse BD-DISTS, whereas BD-MS-SSIM and FID favor DC-VIC. It trails CRDR in BD-PSNR but remains competitive in BD-LPIPS; its clearest perceptual advantage lies in NIQE and very-low-rate visual results. Fig.~\ref{fig:traditional_clic2020} shows a similar contrast against VVC Intra and BPG on CLIC2020: the conventional codecs remain strong under distortion-oriented evaluation, whereas AFP-GIC is often more favorable in the reported low-bitrate perceptual trends, especially NIQE, and in the corresponding visual comparisons.

On Kodak, paired bootstrap analysis with 10,000 resamples supports a positive BD-PSNR gain of +0.233 dB over DC-VIC (95\% CI: [+0.176, +0.304] dB). The corresponding BD-NIQE difference is $-0.069$, although its 95\% CI [$-0.154$, $+0.013$] includes zero. Per-operating-point confidence intervals are reported in Table~S6 of the Supplementary Material.

\begin{table}[!t]
\centering
\caption{Bj{\o}ntegaard delta metrics \cite{bjontegaard2001bdpsnr} relative to MS-ILLM. ``\(\uparrow\)'' indicates higher is better; ``\(\downarrow\)'' indicates lower is better.}
\label{tab:bd_vs_msillm}
\footnotesize
\setlength{\tabcolsep}{1.5pt}
\renewcommand{\arraystretch}{0.92}
\begin{tabular*}{\columnwidth}{@{\extracolsep{\fill}}lccccc@{}}
\hline
\rule{0pt}{1.8ex}Method & PSNR $\uparrow$ & SSIM $\uparrow$ & LPIPS $\downarrow$ & DISTS $\downarrow$ & NIQE $\downarrow$ \\
\hline
\multicolumn{6}{l}{\rule{0pt}{1.8ex}\textit{Kodak}} \\
\rule{0pt}{1.8ex}MS-ILLM & 0.0000 & 0.00000 & 0.00000 & 0.00000 & 0.00000 \\
\rule{0pt}{1.8ex}DC-VIC & 0.0659 & -0.00485 & -0.00504 & -0.00449 & -0.20078 \\
\rule{0pt}{1.8ex}CRDR & 0.9689 & 0.02719 & 0.00477 & 0.00634 & 0.09218 \\
\rule{0pt}{1.8ex}HiFiC & -0.3895 & -0.00499 & 0.02199 & 0.02144 & 0.09426 \\
\rule{0pt}{1.8ex}OSDiff & -2.1296 & -0.04516 & 0.04355 & -0.00859 & 0.15246 \\
\rule{0pt}{1.8ex}RDEIC & -1.2087 & -0.01495 & -0.00008 & -0.02065 & -0.04009 \\
\rule{0pt}{1.8ex}AFP-GIC (Ours) & 0.2850 & -0.00081 & -0.00614 & 0.00213 & -0.26891 \\
\hline
\multicolumn{6}{l}{\rule{0pt}{1.8ex}\textit{CLIC2020}} \\
\rule{0pt}{1.8ex}MS-ILLM & 0.0000 & 0.00000 & 0.00000 & 0.00000 & 0.00000 \\
\rule{0pt}{1.8ex}DC-VIC & -0.0634 & 0.00143 & -0.00047 & -0.00201 & 0.10165 \\
\rule{0pt}{1.8ex}CRDR & 0.7225 & 0.01947 & 0.00891 & 0.01193 & -0.01288 \\
\rule{0pt}{1.8ex}HiFiC & -0.9147 & 0.00241 & 0.01848 & 0.01911 & 0.21617 \\
\rule{0pt}{1.8ex}OSDiff & -3.8811 & -0.04727 & 0.05370 & 0.01046 & 0.63256 \\
\rule{0pt}{1.8ex}RDEIC & -1.8024 & -0.01370 & 0.02292 & 0.00807 & 0.41444 \\
\rule{0pt}{1.8ex}AFP-GIC (Ours) & 0.0632 & 0.00307 & 0.00091 & 0.00393 & -0.06072 \\
\hline
\multicolumn{6}{l}{\rule{0pt}{1.8ex}\textit{DIV2K}} \\
\rule{0pt}{1.8ex}MS-ILLM & 0.0000 & 0.00000 & 0.00000 & 0.00000 & 0.00000 \\
\rule{0pt}{1.8ex}DC-VIC & -0.1275 & -0.00291 & -0.00366 & -0.00569 & -0.10367 \\
\rule{0pt}{1.8ex}CRDR & 0.5713 & 0.01909 & 0.00796 & 0.01033 & 0.13684 \\
\rule{0pt}{1.8ex}HiFiC & -0.7207 & -0.00461 & 0.01935 & 0.01959 & 0.14926 \\
\rule{0pt}{1.8ex}OSDiff & -3.1581 & -0.05951 & 0.04801 & -0.00016 & 0.33162 \\
\rule{0pt}{1.8ex}RDEIC & -1.6605 & -0.02067 & 0.01185 & -0.00173 & 0.07644 \\
\rule{0pt}{1.8ex}AFP-GIC (Ours) & 0.0113 & -0.00149 & -0.00347 & 0.00084 & -0.19965 \\
\hline
\end{tabular*}
\end{table}

\subsection{Qualitative Comparison}
\label{sec:visual_comparison}

\begin{figure*}[!t]
{\centering
\includegraphics[width=0.76\textwidth]{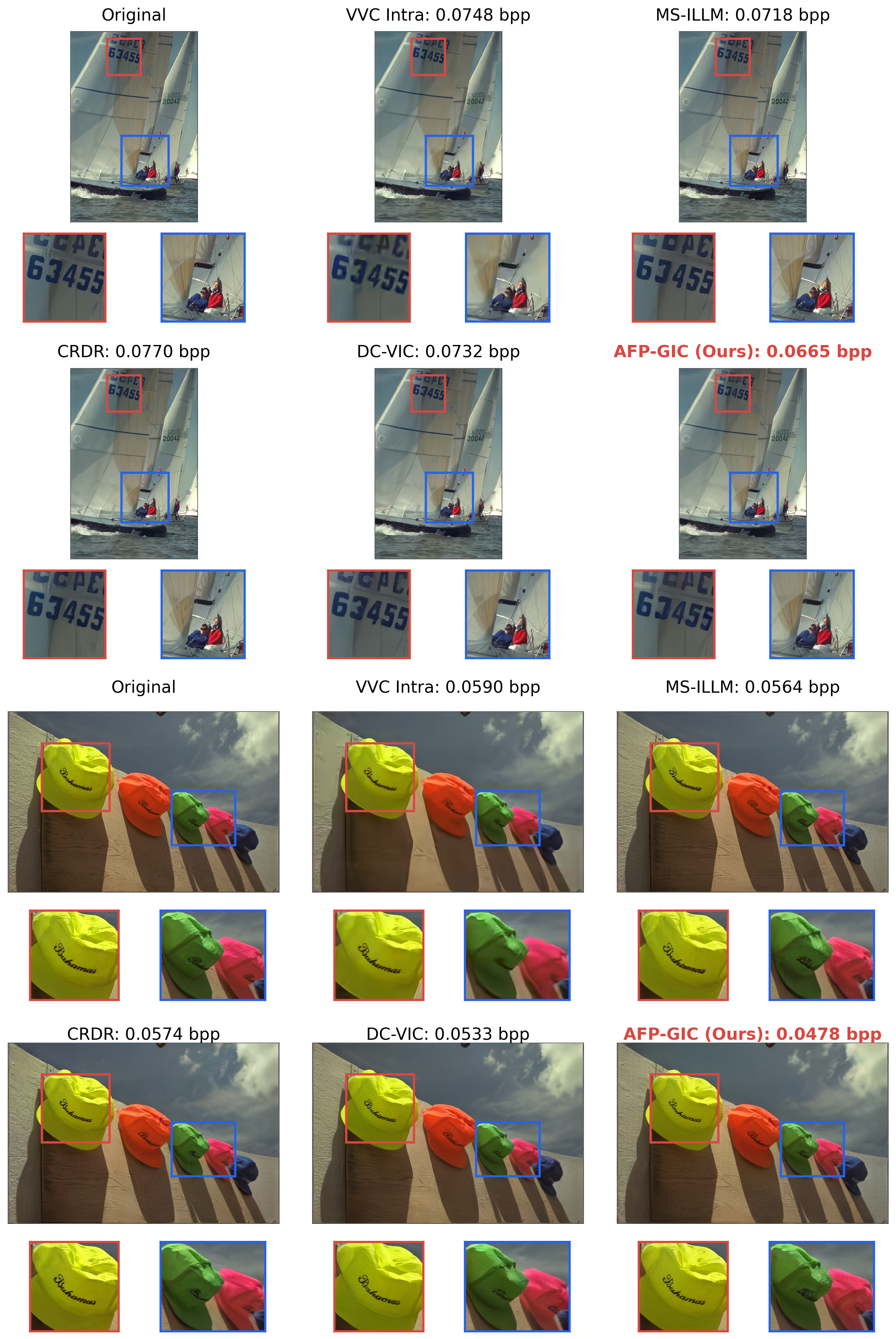}\par}
\vspace{2pt}
\refstepcounter{figure}\label{fig:kodim10_visual_trial}
\noindent\makebox[\textwidth][l]{\parbox[t]{\textwidth}{\raggedright{\color{accessblue}\figcapheadfont FIGURE \thefigure. \ }\figcapfont Qualitative comparison among AFP-GIC and representative learned generative codecs on Kodak in the low-bitrate regime. Baselines are shown at their closest available released bitrates.\par}}
\end{figure*}

\begin{figure*}[!t]
{\centering
\includegraphics[width=0.98\textwidth]{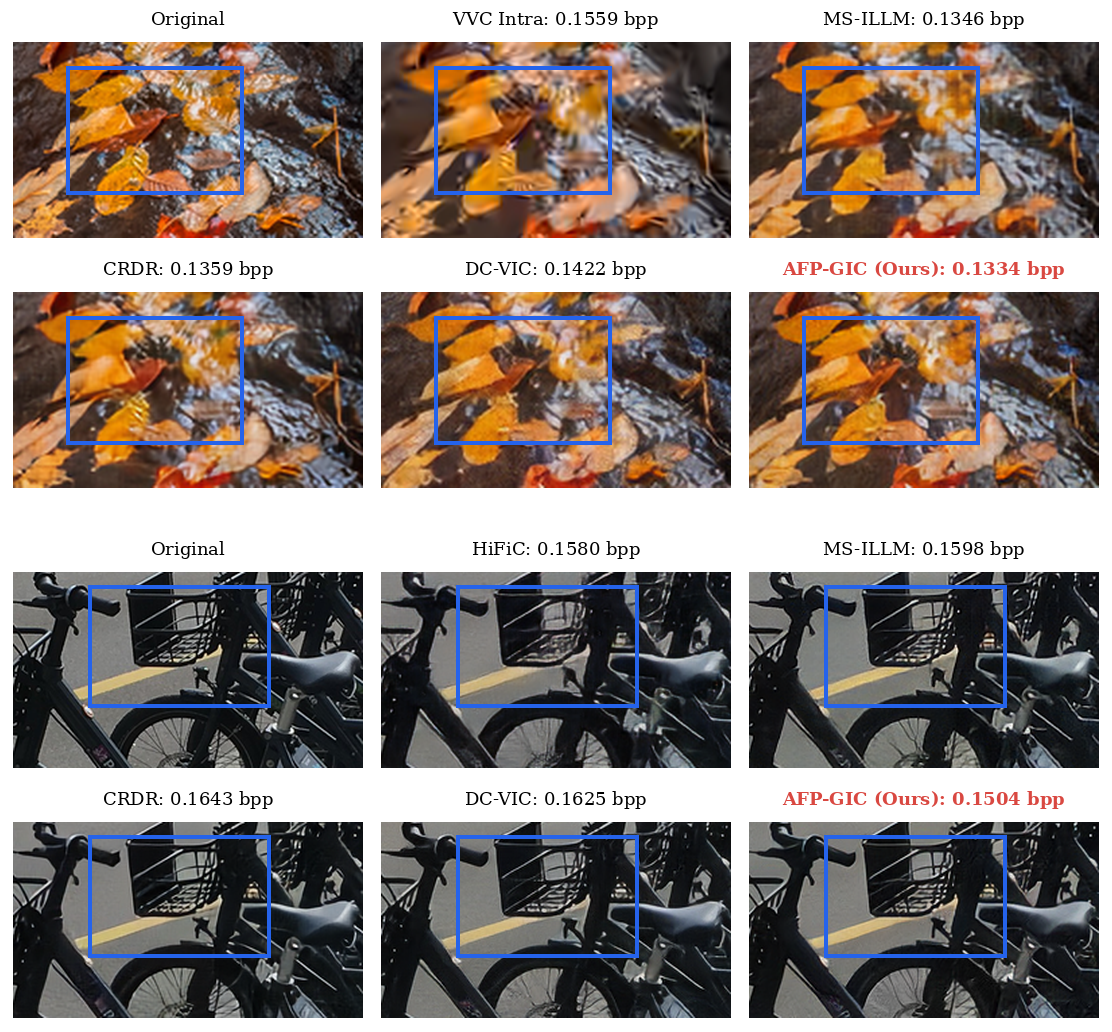}\par}
\vspace{2pt}
\refstepcounter{figure}\label{fig:clic2020_visual_trial}
\noindent\makebox[\textwidth][l]{\parbox[t]{\textwidth}{\raggedright{\color{accessblue}\figcapheadfont FIGURE \thefigure. \ }\figcapfont Qualitative comparison among AFP-GIC and representative codecs on CLIC2020 at close low bitrates. Baselines are shown at their closest available released bitrates.\par}}
\end{figure*}

\begin{figure*}[!t]
{\centering
\includegraphics[width=\textwidth]{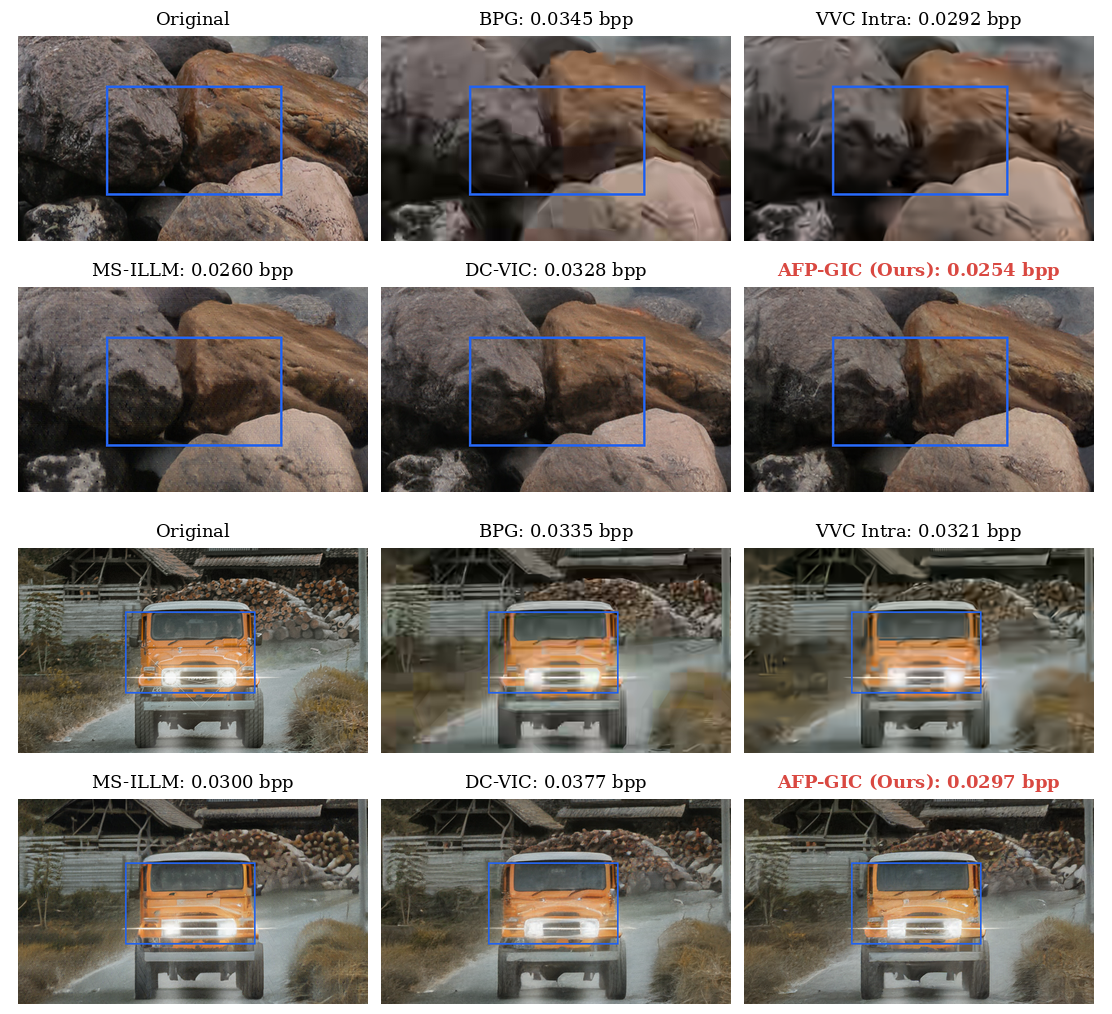}\par}
\vspace{2pt}
\refstepcounter{figure}\label{fig:clic2020_visual_ultralow}
\noindent\makebox[\textwidth][l]{\parbox[t]{\textwidth}{\raggedright{\color{accessblue}\figcapheadfont FIGURE \thefigure. \ }\figcapfont Qualitative comparison among AFP-GIC and representative codecs on CLIC2020 at extreme low bitrates. Baselines are shown at their closest available released bitrates.\par}}
\end{figure*}

Figures~\ref{fig:kodim10_visual_trial}--\ref{fig:diffusion_visual_kodim21} show representative Kodak and CLIC2020 comparisons. Baselines use their closest released bitrates, so some operate slightly above AFP-GIC. Cropped views are used for legibility and should be read with the global trends in Fig.~\ref{fig:quantitative_main}.

\begin{figure}[!t]
\centering
\includegraphics[width=\columnwidth]{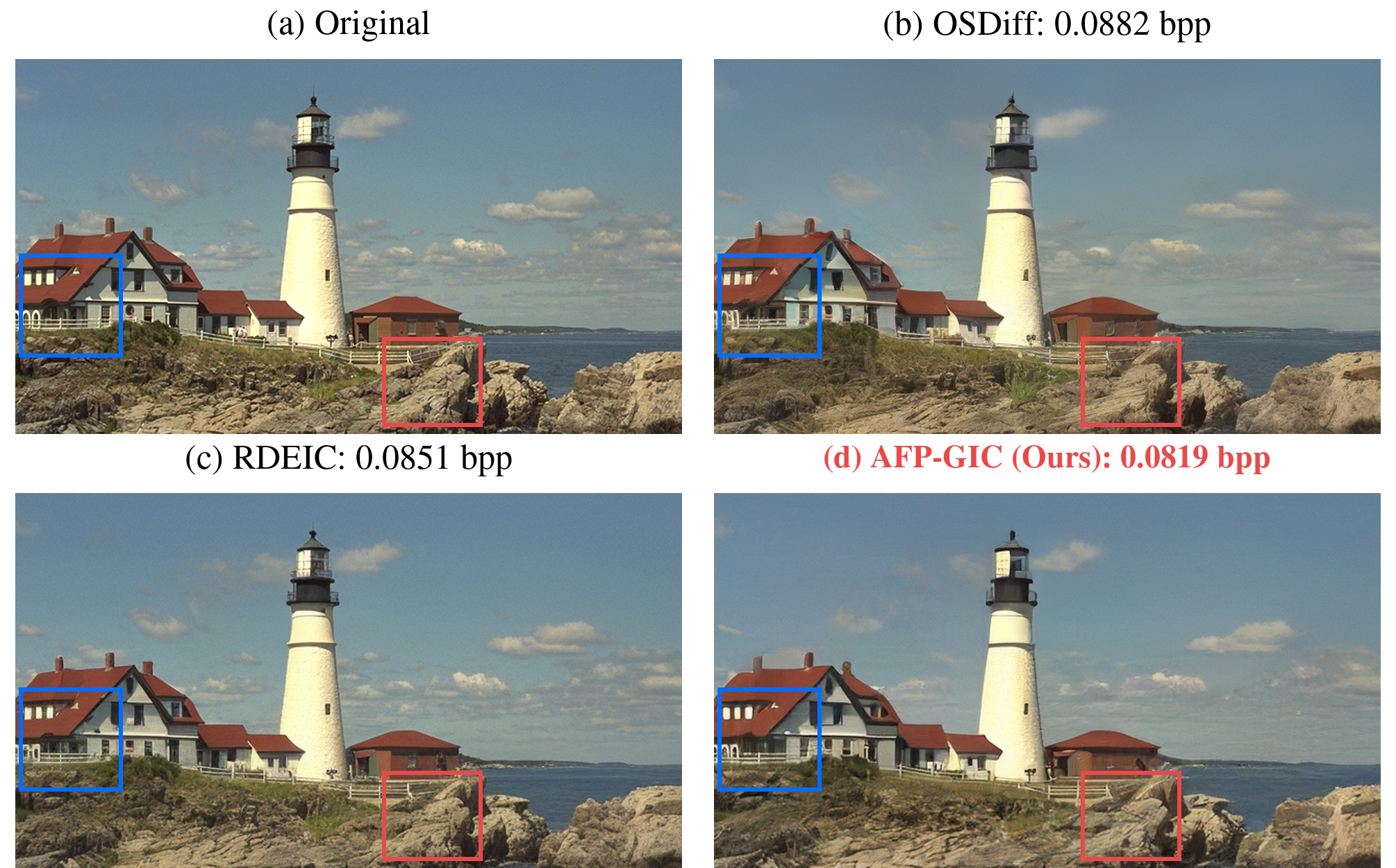}
\caption{Qualitative comparison of a Kodak image crop from full-image reconstructions by AFP-GIC and recent diffusion-based codecs. Baselines use their closest released bitrates.}
\label{fig:diffusion_visual_kodim21}
\end{figure}

AFP-GIC uses the lowest bitrate in both Kodak examples, yet keeps the ``63455'' digits and nearby boat structures clearer in the first and the ``Bahamas'' embroidery and separation among adjacent hats clearer in the second. These examples illustrate how adaptive fused-prior transfer can preserve semantically important local structure under severe rate constraints.

Figures~\ref{fig:clic2020_visual_trial} and \ref{fig:clic2020_visual_ultralow} extend the comparison to CLIC2020, including rates below 0.035 bpp. In Fig.~\ref{fig:clic2020_visual_trial}, AFP-GIC retains sharper leaf edges and a more coherent bicycle basket than the softer generative baselines and blocking-prone VVC Intra. Fig.~\ref{fig:clic2020_visual_ultralow} shows clearer rock boundaries and more legible truck headlights, while the conventional codecs exhibit blocking and MS-ILLM shows noise-like artifacts.

In Fig.~\ref{fig:diffusion_visual_kodim21}, despite operating at the lowest bpp, AFP-GIC preserves finer rock and railing textures and reproduces the house colors with greater fidelity to the original than the recent diffusion-based codecs OSDiff and RDEIC.

\begin{table}[!t]
\centering
\caption{Inference parameter counts and GPU encoder/decoder runtimes on 100 DIV2K patches cropped to \(256\times256\) under a unified RTX 4090 benchmark.}
\label{tab:runtime_learned_gpu}
\footnotesize
\setlength{\tabcolsep}{4pt}
\renewcommand{\arraystretch}{0.94}
\resizebox{\columnwidth}{!}{
\begin{tabular}{lccc}
\hline
\shortstack[l]{\rule{0pt}{2.4ex}Method} & \shortstack[c]{Inference parameter count} & \shortstack[c]{\strut Enc. time\\[-0.6ex]{}[ms]} & \shortstack[c]{\strut Dec. time\\[-0.6ex]{}[ms]}\\
\hline
\shortstack[l]{\rule{0pt}{2.4ex}MS-ILLM} & 181.5M & 16.37 & 20.38\\
\shortstack[l]{\rule{0pt}{2.4ex}CRDR} & 127.7M & 168.77 & 199.62\\
\shortstack[l]{\rule{0pt}{2.4ex}HiFiC} & 148.5M & 110.02 & 220.07\\
\shortstack[l]{\rule{0pt}{2.4ex}OSDiff} & 1362.6M & 67.99 & 147.70\\
\shortstack[l]{\rule{0pt}{2.4ex}RDEIC} & 1380.3M & 76.20 & 204.25\\
\shortstack[l]{\rule{0pt}{2.4ex}DC-VIC} & 151.7M & 61.61 & 98.27\\
\shortstack[l]{\rule{0pt}{2.4ex}\textbf{AFP-GIC (Ours)}} & 120.6M & 81.34 & 80.47\\
\hline
\end{tabular}
}
\end{table}

\subsection{Complexity and Runtime Analysis}
\label{sec:complexity}
\label{sec:runtime_analysis}
We compare inference parameter count and encoder/decoder runtime under a unified RTX 4090 setup (Table~\ref{tab:runtime_learned_gpu}). Counts include all frozen and trainable modules required at inference but exclude training-only networks. AFP-GIC decoder latency covers the complete decompression path: header parsing, entropy decoding and hyper-decoding, the SFT Extractor, Prior Estimator, frozen AdaCode decoder with SFT modulation, and output postprocessing. Decoder cost is especially relevant to write-once, read-many deployment, where an image is encoded once and decoded repeatedly.

Within this common benchmark, AFP-GIC has the smallest inference model (\textbf{120.6M}). Relative to DC-VIC \cite{iwai2024dcvic} (151.7M), AFP-GIC reduces the inference parameter count by \textbf{20.5\%} (31.1M) and decoder time by \textbf{18.1\%} (80.47 vs.\ 98.27 ms). It also has over ten times fewer parameters and decodes faster than OSDiff and RDEIC, although their encoders are slightly faster. MS-ILLM decodes fastest but uses separate bitrate-specific checkpoints, whereas one AFP-GIC model covers all five rates; Table~\ref{tab:runtime_learned_gpu} reports per-checkpoint parameters. Appendix~\ref{app:quality_complexity} combines resource use with BD quality.

Under the same RTX 4090 training benchmark (six $256\times256$ crops), AFP-GIC reduces median \StageIIILink{} step time by \textbf{33.3\%} (239.40 vs.\ 359.08 ms) and peak memory by \textbf{39.0\%} (8.75 vs.\ 14.34 GiB) relative to DC-VIC. Because the baselines use official checkpoints, the inference results are system-level comparisons across different training data, budgets, and pretraining procedures. AFP-GIC's training measurements exclude AdaCode pretraining.

\subsection{Ablation Study}
\label{sec:ablation_study}

\begin{table}[!t]
\centering
\caption{Bj{\o}ntegaard delta metrics \cite{bjontegaard2001bdpsnr} relative to AFP-GIC for the module-level ablation. ``\(\uparrow\)'' indicates higher is better; ``\(\downarrow\)'' indicates lower is better.}
\label{tab:module_ablation_bd}
\footnotesize
\setlength{\tabcolsep}{1.5pt}
\renewcommand{\arraystretch}{0.92}
\begin{tabular*}{\columnwidth}{@{\extracolsep{\fill}}lccccc@{}}
\hline
\rule{0pt}{1.8ex}Variant & PSNR $\uparrow$ & SSIM $\uparrow$ & LPIPS $\downarrow$ & DISTS $\downarrow$ & NIQE $\downarrow$ \\
\hline
\multicolumn{6}{l}{\rule{0pt}{1.8ex}\textit{Kodak}} \\
\rule{0pt}{1.8ex}AFP-GIC & 0.0000 & 0.00000 & 0.00000 & 0.00000 & 0.00000 \\
\rule{0pt}{1.8ex}w/o encoder prior & -0.0441 & -0.00480 & -0.00002 & 0.00131 & -0.09819 \\
\rule{0pt}{1.8ex}w/o SFT modulation & -4.5970 & -0.15781 & 0.20332 & 0.12380 & 0.32139 \\
\rule{0pt}{1.8ex}w/o prior consistency & -0.0111 & -0.00067 & -0.00278 & -0.00241 & -0.10173 \\
\hline
\multicolumn{6}{l}{\rule{0pt}{1.8ex}\textit{CLIC2020}} \\
\rule{0pt}{1.8ex}AFP-GIC & 0.0000 & 0.00000 & 0.00000 & 0.00000 & 0.00000 \\
\rule{0pt}{1.8ex}w/o encoder prior & -0.1387 & -0.00537 & 0.00093 & 0.00092 & -0.03155 \\
\rule{0pt}{1.8ex}w/o SFT modulation & -6.3202 & -0.17752 & 0.19990 & 0.12100 & -0.48406 \\
\rule{0pt}{1.8ex}w/o prior consistency & 0.0818 & 0.00090 & -0.00520 & -0.00119 & -0.05604 \\
\hline
\multicolumn{6}{l}{\rule{0pt}{1.8ex}\textit{DIV2K}} \\
\rule{0pt}{1.8ex}AFP-GIC & 0.0000 & 0.00000 & 0.00000 & 0.00000 & 0.00000 \\
\rule{0pt}{1.8ex}w/o encoder prior & -0.1227 & -0.00562 & 0.00217 & 0.00121 & -0.03358 \\
\rule{0pt}{1.8ex}w/o SFT modulation & -5.0790 & -0.16649 & 0.17416 & 0.10836 & -0.07105 \\
\rule{0pt}{1.8ex}w/o prior consistency & 0.0339 & 0.00002 & -0.00364 & -0.00156 & -0.05025 \\
\hline
\end{tabular*}
\end{table}

\noindent\textbf{Module-Level Contributions:} Table~\ref{tab:module_ablation_bd} compares fully retrained variants under the same three-stage procedure and training budget, with AFP-GIC as the BD reference. The encoder-side prior path and prior-consistency loss have small measured effects: removing the former reduces BD-PSNR by 0.0441--0.1387 dB, while removing the latter leaves BD-PSNR/BD-SSIM largely unchanged and improves the perceptual BD metrics. Removing only the inherited SFT modulation, while retaining the predicted fused prior, causes the largest full-reference quality losses. Tables~\ref{tab:c2_cross_dataset} and~\ref{tab:prior_estimator_substitution_bd} provide evidence for the contribution of the fused prior itself.

\begin{table}[!t]
\caption{\StageIILink{} analysis of the PSNR response to \(\beta_{\mathrm{prior}}\) at fixed \(\beta_{\mathrm{rate}}\), using \StageILink{} checkpoints on the 2000-image validation set.}
\label{tab:prior_consistency_control_response}
\centering
\footnotesize
\setlength{\tabcolsep}{1.8pt}
\renewcommand{\arraystretch}{1.00}
\begin{tabular*}{\columnwidth}{@{\extracolsep{\fill}}llcccccc@{}}
\hline
Model & Statistic & \multicolumn{5}{c}{Target bitrate (bpp)} & Mean \\
\cline{3-7}
\rule{0pt}{1.8ex} & & 0.050 & 0.075 & 0.100 & 0.125 & 0.150 & \\
\hline
\rule{0pt}{1.8ex}Fixed control & \(\beta_{rate}\) & 1.921 & 1.312 & 0.797 & 0.422 & 0.141 & -- \\
\hline
\rule{0pt}{1.8ex}AFP-GIC & \(\beta_{prior}^{*}\) & 3.5 & 3.5 & 3.0 & 3.0 & 3.0 & -- \\
 & \(\Delta_{\max}\) PSNR (dB) & 0.044 & 0.061 & 0.098 & 0.058 & 0.056 & 0.063 \\
\shortstack[l]{w/o prior\\consistency} & \(\beta_{prior}^{*}\) & 3.5 & 2.0 & 2.0 & 2.0 & 2.0 & -- \\
 & \(\Delta_{\max}\) PSNR (dB) & 0.007 & 0.015 & 0.014 & 0.015 & 0.014 & 0.013 \\
\hline
\end{tabular*}
\par\smallskip
\parbox{\columnwidth}{\footnotesize\raggedright The \(\beta_{\mathrm{rate}}\) values are used for Stage II analysis of Stage I checkpoints, before Stage III selected-pair fine-tuning.}
\end{table}

The contribution of the prior-consistency loss \(\mathcal{L}_{\mathrm{prior}}\) is instead reflected in preserving a measurable response to the second control variable, \(\beta_{\mathrm{prior}}\), as analyzed below. Table~\ref{tab:prior_consistency_control_response} evaluates \StageILink{} checkpoints during \StageIILink{} validation, before selected-pair fine-tuning, so its PSNR-maximizing controls differ from the final \StageIIILink{} operating pairs. At each fixed \(\beta_{\mathrm{rate}}\), \(\beta_{\mathrm{prior}}^{*}\) maximizes PSNR over the eight values \(\{0,0.5,\ldots,3.5\}\), and \(\Delta_{\max}\mathrm{PSNR}=\max_{\beta_{\mathrm{prior}}}[\mathrm{PSNR}(\beta_{\mathrm{prior}})-\mathrm{PSNR}(0)]\). Omitting \(\mathcal{L}_{\mathrm{prior}}\) reduces the mean \(\Delta_{\max}\mathrm{PSNR}\) from 0.0634 to 0.0132~dB, a 79.1\% reduction, and changes the mean Pearson correlation between \(\beta_{\mathrm{prior}}\) and PSNR across the five rates from \(+0.731\) to \(-0.221\). The prior-consistency loss preserves a measurable response to \(\beta_{\mathrm{prior}}\), but its PSNR effect is practically negligible at the evaluated operating points (\(\leq 0.1\) dB). Because \(\beta_{\mathrm{prior}}\) also changes bitrate slightly, this diagnostic is not a rate-matched RD gain. Appendix~\ref{app:prior_estimator_substitution} further evaluates the Prior Estimator through controlled substitutions of its predicted fused prior.

\par\vspace{0.5em}
\begin{figure*}[!t]
\centering
\includegraphics[width=0.82\textwidth]{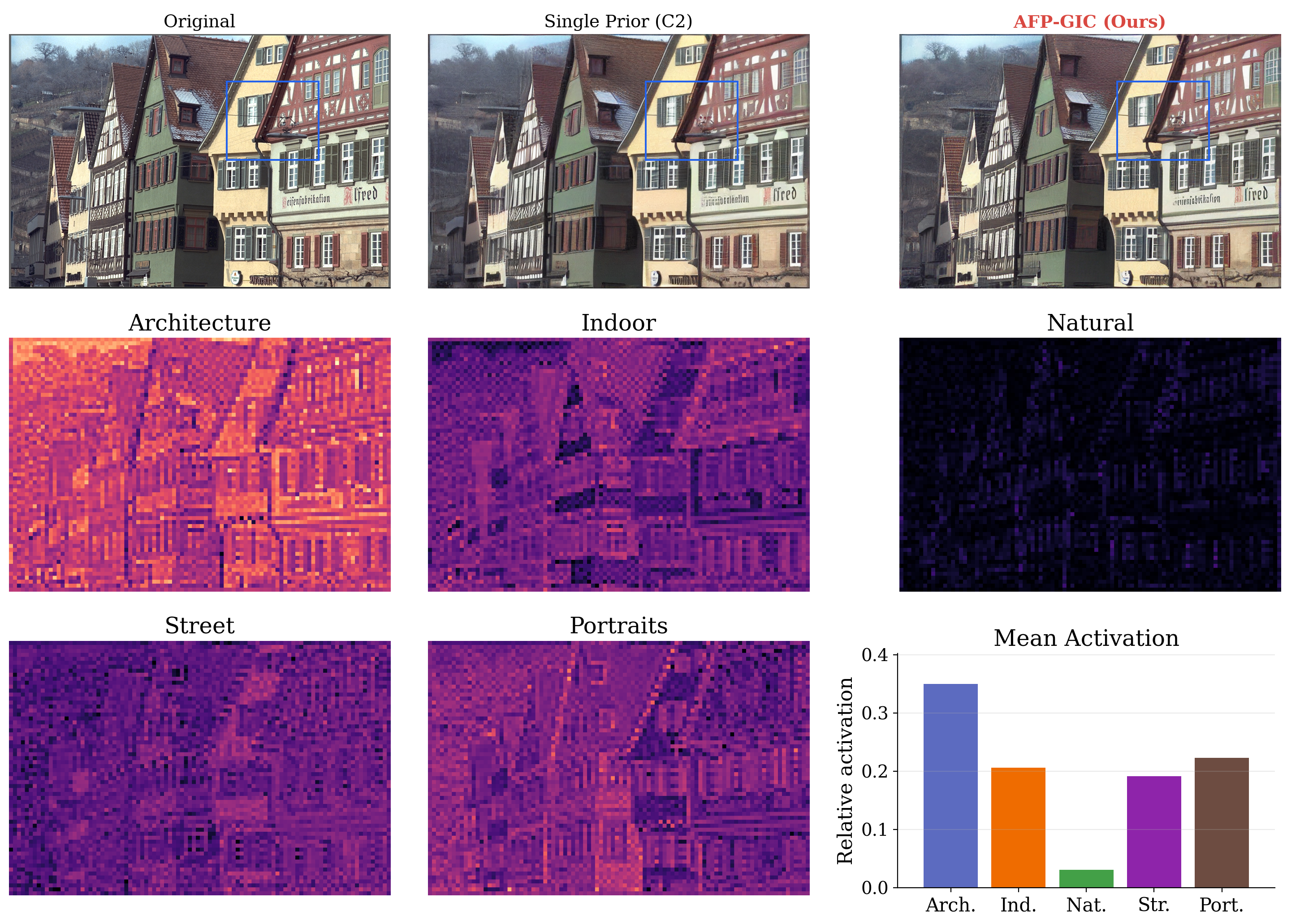}
\caption{Comparison of the representative single-prior baseline C2 and AFP-GIC at close bitrates on Kodak. Lower rows show AFP-GIC's spatial prior responses and the mean normalized activation of each prior component.}
\label{fig:adacode_usage_panel}
\end{figure*}

\noindent\textbf{Spatial Adaptation of the Transferred Adaptive Fused Prior:} Among the single-prior branches listed in Table~\ref{tab:adacode_groups}, C2 is used as the representative single-prior baseline because, in our preliminary evaluation, it gives the highest reconstruction PSNR among the pretrained AdaCode single branches on Kodak. The first row of Fig.~\ref{fig:adacode_usage_panel} shows that AFP-GIC reconstructs clearer local detail than C2 at comparable bitrates around 0.10 bpp. In the blue-boxed building region, facade texture and fine structural patterns are rendered more clearly by AFP-GIC, suggesting that the adaptive fused prior is better matched to this image than the fixed C2 branch.

The second and third rows of Fig.~\ref{fig:adacode_usage_panel} provide qualitative context for this behavior. The activation maps show that the transferred adaptive fused prior is used non-uniformly across the image plane: regions with different structural and textural content activate different prior responses. The mean-activation summary further shows that the architecture component has the largest average normalized activation, whereas the natural component has the smallest. This indicates that, for this building-heavy example, AFP-GIC assigns the largest average activation to architecture-related prior cues while still combining multiple prior families in a spatially adaptive manner rather than collapsing to a fixed single-prior response.

\par\vspace{0.5em}
\noindent\textbf{Adaptive Fused Prior vs.\ Single Prior:} To isolate the effect of adaptive fused-prior transfer, we compare AFP-GIC with the representative single-prior baseline C2 following the same three-stage procedure. Table~\ref{tab:c2_cross_dataset} reports five-point BD comparisons relative to AFP-GIC over each dataset's common bitrate interval. AFP-GIC improves PSNR, SSIM, LPIPS, and DISTS on all three datasets, while C2 yields slightly lower NIQE.

\begin{table}[!t]
\centering
\caption{Bj{\o}ntegaard delta metrics \cite{bjontegaard2001bdpsnr} relative to AFP-GIC for the five-point single-prior comparison. ``\(\uparrow\)'' indicates higher is better; ``\(\downarrow\)'' indicates lower is better.}
\label{tab:c2_cross_dataset}
\footnotesize
\setlength{\tabcolsep}{1.5pt}
\renewcommand{\arraystretch}{0.92}
\begin{tabular*}{\columnwidth}{@{\extracolsep{\fill}}lccccc@{}}
\hline
\rule{0pt}{1.8ex}Variant & PSNR $\uparrow$ & SSIM $\uparrow$ & LPIPS $\downarrow$ & DISTS $\downarrow$ & NIQE $\downarrow$ \\
\hline
\multicolumn{6}{l}{\rule{0pt}{1.8ex}\textit{Kodak}} \\
\rule{0pt}{1.8ex}AFP-GIC & 0.0000 & 0.00000 & 0.00000 & 0.00000 & 0.00000 \\
\rule{0pt}{1.8ex}C2 (single prior) & -0.6318 & -0.01566 & 0.01285 & 0.00635 & -0.05806 \\
\hline
\multicolumn{6}{l}{\rule{0pt}{1.8ex}\textit{CLIC2020}} \\
\rule{0pt}{1.8ex}AFP-GIC & 0.0000 & 0.00000 & 0.00000 & 0.00000 & 0.00000 \\
\rule{0pt}{1.8ex}C2 (single prior) & -0.9202 & -0.01573 & 0.01365 & 0.00671 & -0.01481 \\
\hline
\multicolumn{6}{l}{\rule{0pt}{1.8ex}\textit{DIV2K}} \\
\rule{0pt}{1.8ex}AFP-GIC & 0.0000 & 0.00000 & 0.00000 & 0.00000 & 0.00000 \\
\rule{0pt}{1.8ex}C2 (single prior) & -0.8908 & -0.02232 & 0.01520 & 0.00731 & -0.01488 \\
\hline
\end{tabular*}
\end{table}

\par\vspace{0.5em}
\begin{figure*}[!t]
\centering
\includegraphics[width=0.75\textwidth]{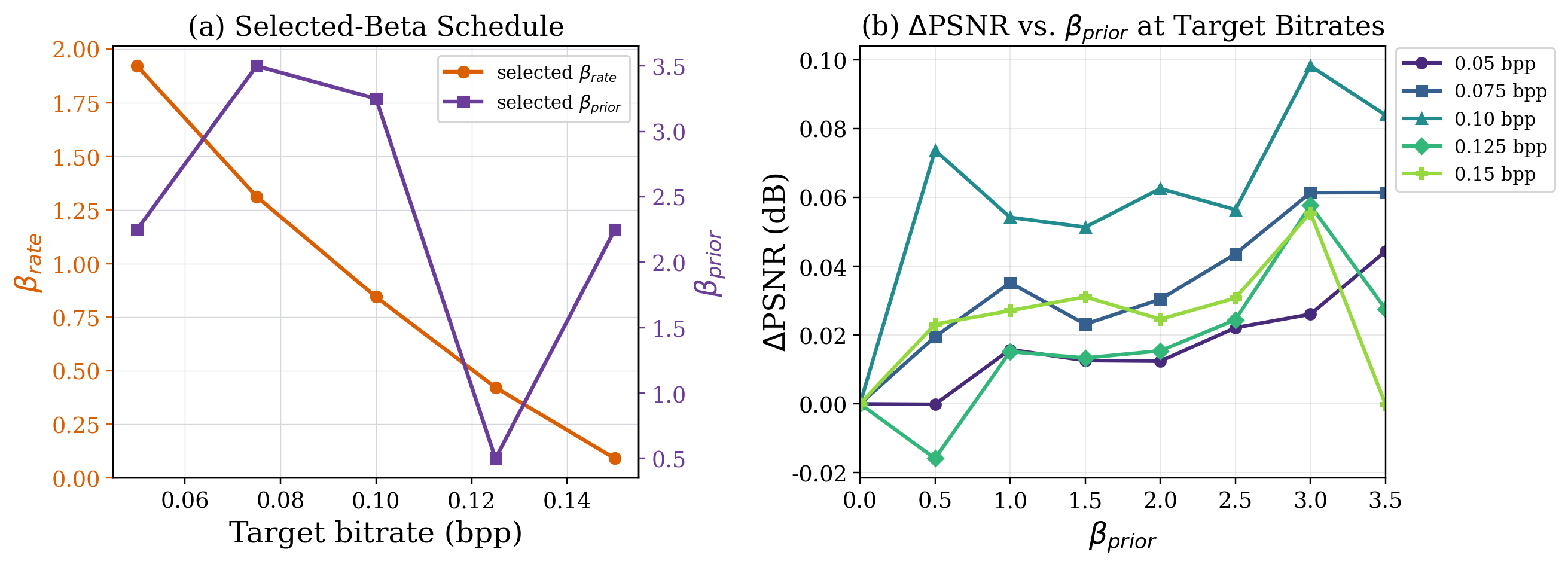}
\caption{Effects of \(\beta_{rate}\) and \(\beta_{prior}\) in AFP-GIC. (a) The selected \(\beta_{rate}\) and \(\beta_{prior}\) values across the five target bitrates. (b) With \(\beta_{rate}\) fixed for each target bitrate, PSNR change relative to \(\beta_{prior}=0\) is plotted at the eight uniformly spaced values \(\{0,0.5,\ldots,3.5\}\).}
\label{fig:dual_control_compact}
\end{figure*}

\noindent\textbf{Effects of \(\beta_{rate}\) and \(\beta_{prior}\):} To understand the effects of \(\beta_{rate}\) and \(\beta_{prior}\), we analyze the beta-pair selection procedure in Section~\ref{sec:training_strategy}. Fig.~\ref{fig:dual_control_compact}(a) plots the selected values of \(\beta_{rate}\) and \(\beta_{prior}\) for the five target bitrates. The selected \(\beta_{rate}\) shows a clear decreasing trend as the target bitrate increases, which matches its role in scaling the rate-related term through \(w_r\) in \EqLink{eq:beta_weights} and \EqLink{eq:base_generator_loss}. In contrast, the selected \(\beta_{prior}\) shows a more flexible selection pattern across the five target bitrates. Fig.~\ref{fig:dual_control_compact}(b) then fixes \(\beta_{rate}\) at the selected value for each of the five target bitrates and evaluates \(\beta_{prior}\) at the eight uniformly spaced values \(\{0,0.5,\ldots,3.5\}\) on the Open Images validation set. The resulting curves show that increasing \(\beta_{prior}\) generally improves PSNR, although local fluctuations remain across the five target bitrate regimes. Fig.~\ref{fig:dual_control_compact} indicates that \(\beta_{\mathrm{rate}}\) primarily controls bitrate, while \(\beta_{\mathrm{prior}}\) produces small PSNR changes at fixed \(\beta_{\mathrm{rate}}\), rather than demonstrating independent fine-grained quality control.

\begin{table}[!t]
\caption{Average header overhead (\%) in the actual bitstream at the five selected target operating points (approximately 0.05, 0.075, 0.10, 0.125, and 0.15 bpp).}
\label{tab:header_overhead}
\centering
\footnotesize
\setlength{\tabcolsep}{3.5pt}
\renewcommand{\arraystretch}{1.00}
\begin{tabular*}{\columnwidth}{@{\extracolsep{\fill}}lccccc@{}}
\hline
\shortstack[l]{\rule{0pt}{2.4ex}Dataset} & \shortstack[c]{\rule{0pt}{2.4ex}Target\\0.05} & \shortstack[c]{\rule{0pt}{2.4ex}Target\\0.075} & \shortstack[c]{\rule{0pt}{2.4ex}Target\\0.10} & \shortstack[c]{\rule{0pt}{2.4ex}Target\\0.125} & \shortstack[c]{\rule{0pt}{2.4ex}Target\\0.15} \\
\hline
\shortstack[l]{\rule{0pt}{2.4ex}Kodak} & 0.268 & 0.173 & 0.127 & 0.099 & 0.080 \\
\shortstack[l]{\rule{0pt}{2.4ex}CLIC2020} & 0.059 & 0.041 & 0.031 & 0.025 & 0.021 \\
\shortstack[l]{\rule{0pt}{2.4ex}DIV2K} & 0.038 & 0.026 & 0.020 & 0.016 & 0.013 \\
\hline
\end{tabular*}
\end{table}

\par\vspace{0.5em}
\noindent\textbf{Effect of Header Overhead on Actual Bitrate:} The transmitted header has a fixed size of 6 bytes (48 bits), consisting of 4 bytes for the image size, 1 byte for the maximum absolute value of the quantized latent \(\hat{\mathbf{y}}\), and 1 byte for the operating-point index. Table~\ref{tab:header_overhead} reports the image-wise header fraction averaged over each dataset at the five selected target operating points; the calculation is detailed in Section~III of the Supplementary Material. The overhead remains very small across all settings: even on Kodak at the lowest target operating point it averages only 0.268\%, and it is substantially smaller on CLIC2020 and DIV2K.

\begin{table}[!t]
\caption{Sensitivity of beta-pair selection to \(\alpha\) in \EqLink{eq:beta_score}. The main model uses \(\alpha=2\).}
\label{tab:alpha_selection_regimes}
\centering
\small
\setlength{\tabcolsep}{4pt}
\renewcommand{\arraystretch}{1.22}
\begin{tabular}{p{0.26\columnwidth}|p{0.66\columnwidth}}
\hline
\centering\(\alpha\) values
&
\centering Changed operating points relative to \(\alpha=2\)
\tabularnewline
\hline
\raggedright \(\{0.01,0.1,1,2\}\)
&
\raggedright No change.
\tabularnewline
\raggedright \(\{3,\ldots,9\}\)
&
\raggedright Only the \(0.10\) bpp operating point changes:\\
\(0.10\ \text{bpp}\rightarrow(0.797,0.25)\).
\tabularnewline
\raggedright \(\{10,20\}\)
&
\raggedright Three bitrates change:\\
\(0.05\ \text{bpp}\rightarrow(1.921,2.5)\),\\
\(0.10\ \text{bpp}\rightarrow(0.797,0.25)\), and\\
\(0.15\ \text{bpp}\rightarrow(0.141,0.75)\).
\tabularnewline
\hline
\end{tabular}
\end{table}

\par\vspace{0.5em}
\noindent\textbf{Sensitivity to the Beta-Selection Coefficient \(\alpha\):} Table~\ref{tab:alpha_selection_regimes} summarizes how the validation-time coefficient \(\alpha\) in \EqLink{eq:beta_score} affects the selected operating points and their beta pairs. We include small \(\alpha\) values below \(2\), for which the scaled PSNR term contributes less and FID has greater relative influence, and extend the sweep above \(2\), where PSNR increasingly dominates and FID has less relative influence. The selection pattern is stable over much of the sweep. Because \EqLink{eq:beta_score} ranks candidate beta pairs using both PSNR-based reference fidelity and FID-based realism screening, we use \(\alpha=2\) as the main setting within the stable tested range \(\{0.01,0.1,1,2\}\). Consistent with this weighting shift, the intermediate range \(\{3,\ldots,9\}\) changes only one operating point, namely the \(0.10\) bpp operating point, whereas \(\alpha\in\{10,20\}\) changes two additional operating points, namely \(0.05\) and \(0.15\) bpp. The \(0.075\) and \(0.125\) bpp points remain unchanged throughout the sweep. Under the \StageIIILink{} in-loop Kodak evaluation used for this diagnostic, the \(\alpha=10\) follow-up changes the estimated bitrate by only \(+0.0002\) bpp and PSNR by \(-0.03\) dB relative to \(\alpha=2\), even though the validation-time selected pairs are not identical. Overall, the ablation shows that \(\alpha\) affects the validation-time selection rule, while the tested follow-up setting shows only small changes in the final Kodak bitrate and PSNR.

\begin{table}[!t]
\caption{Local decoder-sensitivity and prior-alignment checks on Kodak using random perturbations and interpolation toward the ground-truth fused prior.}
\label{tab:local_decoder_sanity}
\centering
\footnotesize
\setlength{\tabcolsep}{4pt}
\renewcommand{\arraystretch}{1.00}
\begin{tabular*}{\columnwidth}{@{\extracolsep{\fill}}lccc@{}}
\hline
\multicolumn{4}{l}{\rule{0pt}{2.4ex}\textit{Random perturbation sensitivity}} \\
\shortstack[c]{Prior perturb.\\(\%)} & Mean \(Q_{\rho}\) & Max. \(Q_{\rho}\) & \shortstack[c]{Mean relative\\output change (\%)} \\
\hline
0.25 & 0.335 & 0.636 & 0.075 \\
0.50 & 0.301 & 0.590 & 0.135 \\
1.00 & 0.291 & 0.568 & 0.261 \\
2.00 & 0.286 & 0.559 & 0.514 \\
\hline
\multicolumn{4}{l}{\rule{0pt}{2.4ex}\textit{Interpolation toward the ground-truth fused prior}} \\
\(t\) & Relative prior error & Mean PSNR (dB) & \(\Delta\)PSNR (dB) \\
\hline
0.00 & 1.00 & 26.994 & 0.000 \\
0.25 & 0.75 & 27.214 & +0.220 \\
0.50 & 0.50 & 27.318 & +0.324 \\
0.75 & 0.25 & 27.359 & +0.365 \\
1.00 & 0.00 & 27.363 & +0.369 \\
\hline
\end{tabular*}
\end{table}

\noindent\textbf{Local Decoder Sensitivity and Prior Alignment:} We conduct two controlled tests corresponding to the local-sensitivity premise and the prior-alignment implication of Proposition~\ref{prop:prior_alignment_bound}. First, we evaluate the decoder response to small perturbations of the predicted fused prior. For each of the 24 Kodak images and five operating points, we keep the quantized latent \(\hat{\mathbf{y}}\), the selected control pair, the SFT features, and all model parameters fixed, and perturb only \(\hat{\mathbf{p}}\). Specifically, we independently sample five Gaussian tensors \(\mathbf{g}\) with the same shape as \(\hat{\mathbf{p}}\), normalize each as \(\mathbf{u}=\mathbf{g}/\|\mathbf{g}\|_2\), and construct \(\boldsymbol{\delta}=\rho\|\hat{\mathbf{p}}\|_2\mathbf{u}\), where \(\rho\in\{0.25\%,0.50\%,1.00\%,2.00\%\}\). This produces \(24\times5\times5=600\) perturbation trials per magnitude and 2,400 perturbed decodes in total. Let \(f(\hat{\mathbf{y}},\hat{\mathbf{p}})\) denote the decoder output with the non-prior inputs held fixed. For each trial, we compute the finite-difference quotient
\begin{equation}
Q_{\rho}=
\frac{\|f(\hat{\mathbf{y}},\hat{\mathbf{p}}+\boldsymbol{\delta})-
f(\hat{\mathbf{y}},\hat{\mathbf{p}})\|_2}
{\|\boldsymbol{\delta}\|_2}
\label{eq:decoder_fdq}
\end{equation}
and the relative output change
\begin{equation}
R_{\rho}=
\frac{\|f(\hat{\mathbf{y}},\hat{\mathbf{p}}+\boldsymbol{\delta})-
f(\hat{\mathbf{y}},\hat{\mathbf{p}})\|_2}
{\|f(\hat{\mathbf{y}},\hat{\mathbf{p}})\|_2}\times100\%.
\label{eq:relative_output_change}
\end{equation}
Table~\ref{tab:local_decoder_sanity} reports the mean and maximum \(Q_{\rho}\) and the mean \(R_{\rho}\) over the 600 trials at each magnitude. As \(\rho\) increases from 0.25\% to 2.00\%, the mean output change increases from 0.075\% to 0.514\%, while the mean \(Q_{\rho}\) remains within 0.286--0.335 and its maximum does not exceed 0.636. The approximately proportional output response and the absence of a sharp increase in \(Q_{\rho}\) at smaller perturbations are consistent with stable local sensitivity around the predicted priors encountered by the model. This finite-sample test does not establish global Lipschitz continuity.

Second, we examine whether reducing the mismatch between the predicted fused prior \(\hat{\mathbf{p}}\) and the ground-truth image-adaptive fused prior \(\mathbf{p}\) improves reconstruction. We replace \(\hat{\mathbf{p}}\) with \(\mathbf{p}_{t}=(1-t)\hat{\mathbf{p}}+t\mathbf{p}\) for \(t\in\{0,0.25,0.50,0.75,1\}\), while again fixing \(\hat{\mathbf{y}}\), the control pair, the SFT features, and all model parameters. Because \(\|\mathbf{p}_{t}-\mathbf{p}\|_2/\|\hat{\mathbf{p}}-\mathbf{p}\|_2=1-t\), the interpolation progressively reduces the relative prior error from 1 to 0 without changing the bitstream or the other decoder inputs. Mean PSNR increases monotonically from 26.994 dB at \(t=0\) to 27.363 dB at \(t=1\), and the endpoint improvement is positive for all 120 image--operating-point combinations. Table~\ref{tab:local_decoder_sanity} evaluates floating-point reconstructions in the forward diagnostic, whereas Table~\ref{tab:exact_points_kodak} reports metrics from saved RGB reconstructions after compression and decompression. Here, \(\mathbf{p}\) is the ground-truth image-adaptive fused prior extracted from the input image by the frozen AdaCode prior extractor and used to supervise \(\hat{\mathbf{p}}\). It is unavailable during practical decoding and is not assumed to equal the unobservable ideal prior \(\mathbf{p}^{\star}\) in Proposition~\ref{prop:prior_alignment_bound}. The interpolation therefore tests the effect of reducing \(\|\hat{\mathbf{p}}-\mathbf{p}\|_2\), rather than directly verifying alignment with \(\mathbf{p}^{\star}\).

\par\vspace{0.4em}
\noindent\textbf{Impact of Adversarial Supervision:} Fig.~\ref{fig:ablation_wo_discriminator} illustrates the visual contribution of the PatchGAN discriminator at 0.0833 bpp for a representative Kodak sample. Without adversarial supervision, the reconstruction exhibits noticeable noise-like artifacts and less coherent texture patterns on structured regions such as the grass. In contrast, the full AFP-GIC model produces cleaner local structures and more coherent high-frequency textures, suggesting that adversarial supervision suppresses local artifacts in this example.

\begin{figure}[H]
\vspace{-6pt}
\centering
\includegraphics[width=\columnwidth]{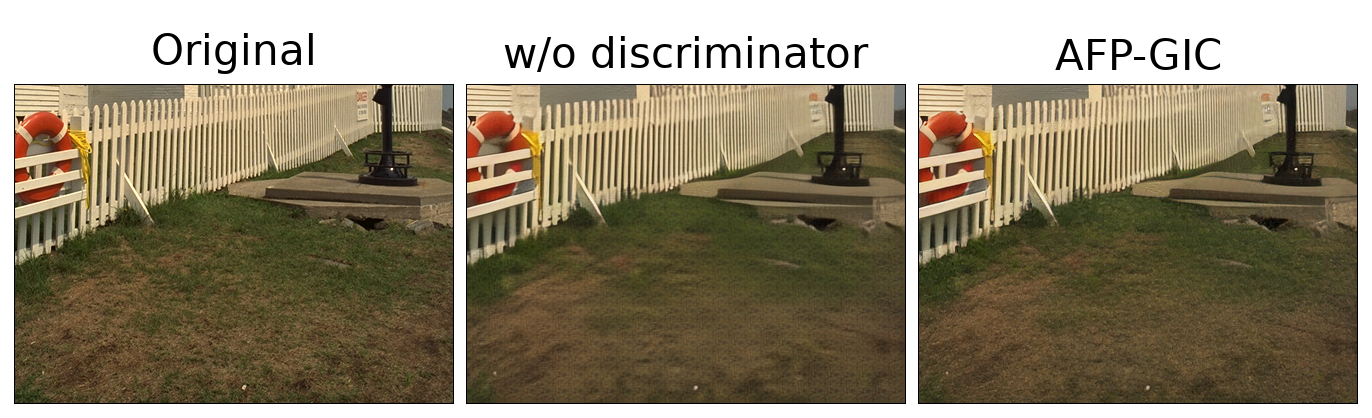}
\refstepcounter{figure}\label{fig:ablation_wo_discriminator}
\par\vspace{1pt}
{\raggedright{\color{accessblue}\figcapheadfont FIGURE \thefigure. }\figcapfont Visual effect of adversarial supervision at 0.0833 bpp.\par}
\vspace{-4pt}
\end{figure}

\section{Discussion and Future Work}
\label{sec:discussion_future}

At very low bitrates, the compressed latent cannot fully specify local structure, so plausible reconstruction depends strongly on the transferred prior. AFP-GIC adopts an image-adaptive fused prior from AdaCode rather than a fixed single-codebook prior to balance reference fidelity and perceptual naturalness. While this prior-driven synthesis yields superior no-reference naturalness (NIQE), DC-VIC retains better results in full-reference metrics (MS-SSIM, DISTS) and FID---reflecting the classic perception-distortion trade-off where perceptual gains rarely align uniformly across complementary criteria. Additionally, the NIQE differences between the released and locally trained DC-VIC models in Fig.~\ref{fig:dcvic_training_pipeline_comparison} highlight the metric's notable sensitivity to training configurations.

\begin{figure}[H]
\centering
\includegraphics[width=\columnwidth]{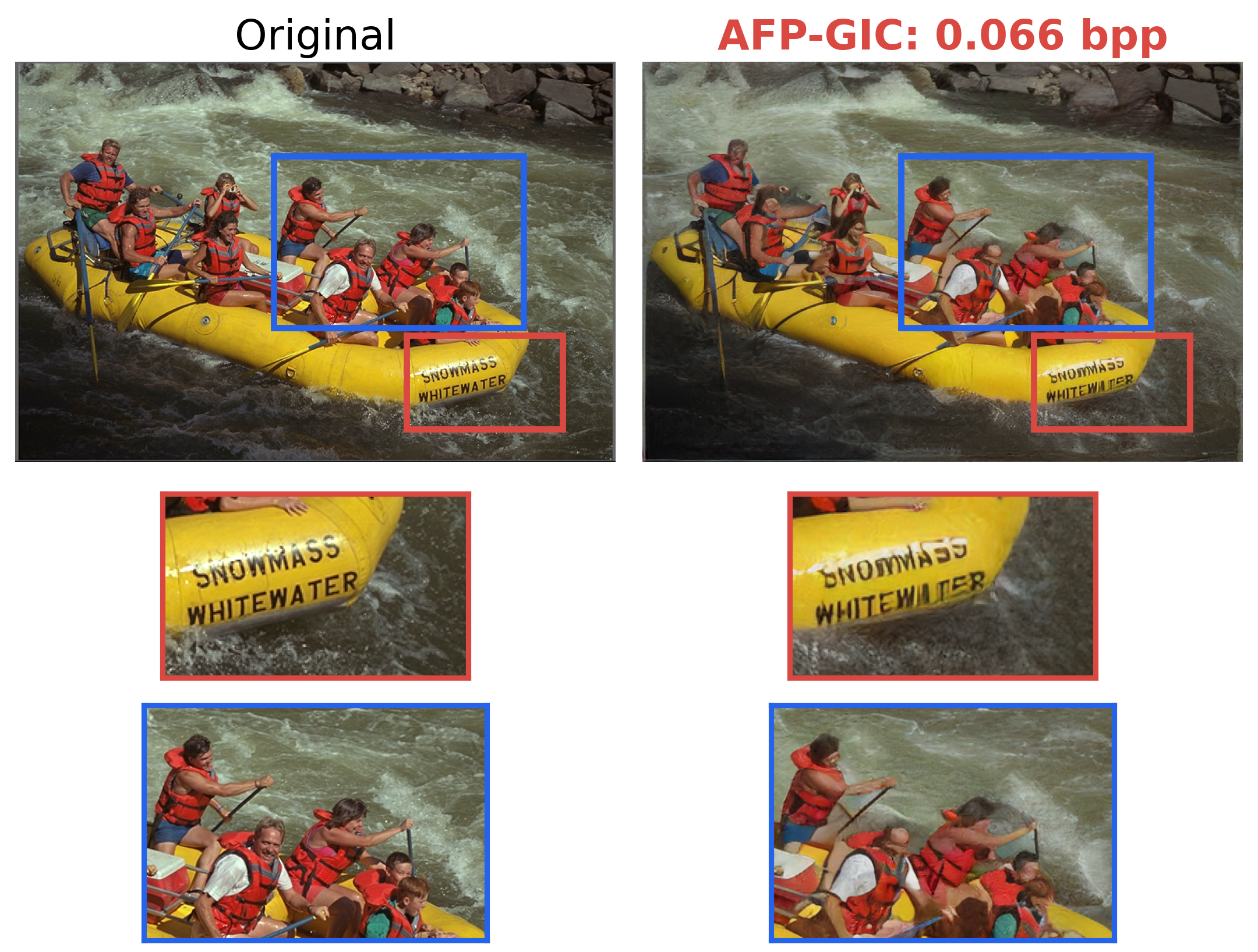}
\caption{Representative failure case of AFP-GIC at 0.066 bpp on Kodak. Although the global scene is preserved, the enlarged regions show malformed text and inaccurate fine structures.}
\label{fig:afpgic_failure_case}
\end{figure}

Prior guidance cannot recover every detail when the bit budget is extremely small. In the representative lowest-rate case in Fig.~\ref{fig:afpgic_failure_case}, AFP-GIC preserves the global scene and color layout but distorts the raft text, faces, hands, and paddles. These errors become less visible at higher operating points, although symbol-level and fine geometric details remain difficult when little direct information is available.

Controllability and decoder efficiency also matter in deployment. A single trained AFP-GIC model supports all reported operating points and adapts to changing bitrate requirements without maintaining separate bitrate-specific models. Its compact decoder is useful in write-once, read-many applications, where encoded images may be decoded repeatedly. The practical value of the framework therefore extends beyond low-rate naturalness to flexible rate control and efficient repeated decoding.

The present study also has several limitations. The experiments are conducted with a fixed training dataset and training budget, so the scaling behavior of AFP-GIC with larger datasets or longer training schedules remains to be examined. Since the transferred prior is obtained from a frozen AdaCode model, the final reconstruction behavior also depends on the representation capacity and inductive bias of that external prior. With the remaining model components fixed, replacing the predicted fused prior with its ground-truth counterpart yields a BD-PSNR gain of approximately 0.39 dB (Table~\ref{tab:prior_estimator_substitution_bd}), suggesting modest headroom from improving prior estimation alone. In addition, the loss terms remain coupled in the current framework; more fine-grained ablations could further separate their individual contributions. The evaluation follows the common objective-metric and visual-comparison protocol used in recent perceptual compression work; subjective preference studies would be complementary, but are outside the scope of the present study.

AFP-GIC is designed for severe information bottlenecks, and its behavior at higher bitrates remains an open question. At higher bitrates, fine details may be transmitted more directly with less reliance on prior synthesis; future work can explore this transition. Other directions include incorporating stronger foundation models as external priors, improving decoder-side prior estimation, and developing more data-efficient control-selection strategies.

\section{Conclusion}
\label{sec:conclusion}

This paper introduced AFP-GIC, a controllable generative compression framework that transfers an adaptive fused prior from a frozen AdaCode model to address the limitations imposed by a fixed single-codebook prior at very low bitrates. AFP-GIC combines encoder-side prior guidance with decoder-side prior prediction, thereby handling the encoder-decoder prior asymmetry without transmitting the fused prior itself. AFP-GIC offers lower decoder latency than DC-VIC, at the cost of higher encoder latency (81.3 vs.\ 61.6 ms) under the \(256\times256\) benchmark. The evidence supports adaptive fused-prior transfer as a practical direction for controllable low-bitrate perceptual compression.

\raggedbottom
\appendices
\section{Numerical Values Underlying the Quantitative Comparison}
\label{app:numerical_results}

Tables~\ref{tab:exact_points_kodak}--\ref{tab:exact_points_div2k} report the numerical values underlying Fig.~\ref{fig:quantitative_main}, using available evaluated operating points within or near AFP-GIC's bitrate range without introducing synthetic points.

\begin{table}[H]
\caption{Numerical values for all operating points used for Kodak in Fig.~\ref{fig:quantitative_main}.}
\label{tab:exact_points_kodak}
\centering
\scriptsize
\setlength{\tabcolsep}{1.6pt}
\renewcommand{\arraystretch}{0.88}
\resizebox{\columnwidth}{!}{
\begin{tabular}{lcccccc}
\hline
\rule{0pt}{1.8ex}Method & bpp & PSNR $\uparrow$ & SSIM $\uparrow$ & LPIPS $\downarrow$ & DISTS $\downarrow$ & NIQE $\downarrow$ \\
\hline
\multirow{5}{*}{AFP-GIC (Ours)} & 0.0516 & 25.133 & 0.63108 & 0.13285 & 0.12913 & 3.207 \\
 & 0.0812 & 26.268 & 0.67602 & 0.10142 & 0.10845 & 3.184 \\
 & 0.1124 & 27.149 & 0.71247 & 0.08298 & 0.09482 & 3.146 \\
 & 0.1456 & 27.870 & 0.74269 & 0.06964 & 0.08350 & 3.100 \\
 & 0.1794 & 28.463 & 0.76722 & 0.06077 & 0.07567 & 3.122 \\
\hline
\multirow{5}{*}{DC-VIC} & 0.0537 & 25.077 & 0.63559 & 0.13200 & 0.12475 & 3.201 \\
 & 0.0860 & 26.234 & 0.68069 & 0.09773 & 0.09903 & 3.245 \\
 & 0.1164 & 26.892 & 0.70713 & 0.08245 & 0.08474 & 3.209 \\
 & 0.1508 & 27.688 & 0.73850 & 0.06941 & 0.07479 & 3.233 \\
 & 0.1888 & 28.415 & 0.76714 & 0.05957 & 0.06985 & 3.251 \\
\hline
\multirow{5}{*}{CRDR} & 0.1095 & 27.444 & 0.73164 & 0.09630 & 0.10431 & 3.563 \\
 & 0.1280 & 27.983 & 0.75246 & 0.08663 & 0.09637 & 3.570 \\
 & 0.1497 & 28.520 & 0.77175 & 0.07654 & 0.08931 & 3.562 \\
 & 0.1750 & 29.052 & 0.78966 & 0.06728 & 0.08214 & 3.474 \\
 & 0.2043 & 29.522 & 0.80446 & 0.05914 & 0.07390 & 3.376 \\
\hline
\multirow{8}{*}{MS-ILLM} & 0.0039 & 20.541 & 0.49148 & 0.49866 & 0.32674 & 4.567 \\
 & 0.0071 & 21.680 & 0.50437 & 0.39791 & 0.26217 & 3.772 \\
 & 0.0420 & 24.730 & 0.61612 & 0.15759 & 0.13435 & 3.581 \\
 & 0.0783 & 25.922 & 0.67174 & 0.10990 & 0.10881 & 3.400 \\
 & 0.1508 & 27.532 & 0.74146 & 0.07322 & 0.08101 & 3.346 \\
 & 0.2935 & 29.634 & 0.82115 & 0.04508 & 0.06152 & 3.414 \\
 & 0.4300 & 31.034 & 0.86047 & 0.03203 & 0.04871 & 3.147 \\
 & 0.7068 & 33.212 & 0.90694 & 0.01804 & 0.03347 & 3.029 \\
\hline
\multirow{2}{*}{HiFiC} & 0.1490 & 27.264 & 0.74413 & 0.09263 & 0.10373 & 3.533 \\
 & 0.3033 & 29.098 & 0.80909 & 0.06301 & 0.08173 & 3.459 \\
\hline
\multirow{2}{*}{OSDiff} & 0.0950 & 24.317 & 0.65391 & 0.14785 & 0.09766 & 3.669 \\
 & 0.1369 & 25.073 & 0.68424 & 0.11322 & 0.06999 & 3.510 \\
\hline
\multirow{5}{*}{RDEIC} & 0.0245 & 22.372 & 0.55678 & 0.21948 & 0.15458 & 3.265 \\
 & 0.0429 & 23.455 & 0.60406 & 0.15958 & 0.11828 & 3.332 \\
 & 0.0655 & 24.500 & 0.64422 & 0.12272 & 0.08829 & 3.421 \\
 & 0.0910 & 25.215 & 0.67170 & 0.10017 & 0.07286 & 3.540 \\
 & 0.1211 & 25.779 & 0.69474 & 0.08564 & 0.06054 & 3.473 \\
\hline
\end{tabular}
}
\end{table}

\begin{table}[H]
\caption{Numerical values for all operating points used for CLIC2020 in Fig.~\ref{fig:quantitative_main}.}
\label{tab:exact_points_clic}
\centering
\scriptsize
\setlength{\tabcolsep}{1.6pt}
\renewcommand{\arraystretch}{0.88}
\resizebox{\columnwidth}{!}{
\begin{tabular}{lcccccc}
\hline
\rule{0pt}{1.8ex}Method & bpp & PSNR $\uparrow$ & SSIM $\uparrow$ & LPIPS $\downarrow$ & DISTS $\downarrow$ & NIQE $\downarrow$ \\
\hline
\multirow{5}{*}{AFP-GIC (Ours)} & 0.0391 & 27.986 & 0.74826 & 0.10629 & 0.08683 & 3.745 \\
 & 0.0599 & 29.136 & 0.77824 & 0.08284 & 0.07125 & 3.738 \\
 & 0.0812 & 29.978 & 0.80001 & 0.06885 & 0.06175 & 3.726 \\
 & 0.1029 & 30.652 & 0.81680 & 0.05923 & 0.05520 & 3.714 \\
 & 0.1240 & 31.184 & 0.82966 & 0.05264 & 0.05027 & 3.708 \\
\hline
\multirow{5}{*}{DC-VIC} & 0.0417 & 28.123 & 0.75538 & 0.09951 & 0.07970 & 3.854 \\
 & 0.0655 & 29.316 & 0.78508 & 0.07582 & 0.06217 & 3.897 \\
 & 0.0922 & 29.996 & 0.80208 & 0.06394 & 0.05162 & 3.890 \\
 & 0.1158 & 30.712 & 0.81892 & 0.05533 & 0.04644 & 3.903 \\
 & 0.1372 & 31.328 & 0.83302 & 0.04925 & 0.04326 & 3.918 \\
\hline
\multirow{7}{*}{CRDR} & 0.0797 & 30.164 & 0.81134 & 0.08370 & 0.07482 & 3.840 \\
 & 0.0918 & 30.686 & 0.82396 & 0.07561 & 0.07053 & 3.868 \\
 & 0.1056 & 31.181 & 0.83481 & 0.06821 & 0.06574 & 3.858 \\
 & 0.1215 & 31.640 & 0.84408 & 0.06125 & 0.06036 & 3.806 \\
 & 0.1401 & 32.042 & 0.85078 & 0.05492 & 0.05436 & 3.700 \\
 & 0.1642 & 32.679 & 0.86330 & 0.04860 & 0.05027 & 3.697 \\
 & 0.1923 & 33.259 & 0.87346 & 0.04307 & 0.04534 & 3.642 \\
\hline
\multirow{8}{*}{MS-ILLM} & 0.0036 & 22.715 & 0.63373 & 0.43856 & 0.25872 & 4.760 \\
 & 0.0062 & 24.170 & 0.65690 & 0.32810 & 0.20203 & 4.112 \\
 & 0.0343 & 27.880 & 0.74014 & 0.11368 & 0.08574 & 3.900 \\
 & 0.0629 & 29.239 & 0.77682 & 0.08016 & 0.06657 & 3.840 \\
 & 0.1203 & 30.871 & 0.82254 & 0.05415 & 0.04808 & 3.717 \\
 & 0.2310 & 32.836 & 0.86947 & 0.03494 & 0.03516 & 3.697 \\
 & 0.3309 & 34.112 & 0.89444 & 0.02657 & 0.02883 & 3.692 \\
 & 0.5121 & 35.884 & 0.92125 & 0.01737 & 0.02006 & 3.604 \\
\hline
\multirow{2}{*}{HiFiC} & 0.1074 & 29.841 & 0.82213 & 0.07479 & 0.07052 & 4.010 \\
 & 0.2258 & 31.576 & 0.86326 & 0.05290 & 0.05324 & 3.982 \\
\hline
\multirow{2}{*}{OSDiff} & 0.0829 & 26.085 & 0.75381 & 0.12837 & 0.07426 & 4.536 \\
 & 0.1235 & 26.970 & 0.77299 & 0.09684 & 0.05138 & 4.313 \\
\hline
\multirow{5}{*}{RDEIC} & 0.0174 & 24.280 & 0.68599 & 0.21636 & 0.15262 & 4.366 \\
 & 0.0326 & 25.686 & 0.72395 & 0.14822 & 0.10625 & 4.194 \\
 & 0.0524 & 27.291 & 0.75779 & 0.10596 & 0.07078 & 4.212 \\
 & 0.0764 & 28.288 & 0.77704 & 0.08365 & 0.05327 & 4.199 \\
 & 0.1059 & 29.043 & 0.79313 & 0.06960 & 0.04061 & 4.177 \\
\hline
\end{tabular}
}
\end{table}

\begin{table}[H]
\caption{Numerical values for all operating points used for DIV2K in Fig.~\ref{fig:quantitative_main}.}
\label{tab:exact_points_div2k}
\centering
\scriptsize
\setlength{\tabcolsep}{1.6pt}
\renewcommand{\arraystretch}{0.88}
\resizebox{\columnwidth}{!}{
\begin{tabular}{lcccccc}
\hline
\rule{0pt}{1.8ex}Method & bpp & PSNR $\uparrow$ & SSIM $\uparrow$ & LPIPS $\downarrow$ & DISTS $\downarrow$ & NIQE $\downarrow$ \\
\hline
\multirow{5}{*}{AFP-GIC (Ours)} & 0.0549 & 25.793 & 0.69276 & 0.11726 & 0.09511 & 3.193 \\
 & 0.0846 & 26.999 & 0.73746 & 0.08919 & 0.07713 & 3.155 \\
 & 0.1151 & 27.883 & 0.76893 & 0.07312 & 0.06599 & 3.154 \\
 & 0.1459 & 28.596 & 0.79253 & 0.06230 & 0.05834 & 3.140 \\
 & 0.1760 & 29.181 & 0.81021 & 0.05486 & 0.05250 & 3.124 \\
\hline
\multirow{5}{*}{DC-VIC} & 0.0566 & 25.815 & 0.69922 & 0.11306 & 0.08778 & 3.296 \\
 & 0.0876 & 26.999 & 0.74162 & 0.08584 & 0.06869 & 3.244 \\
 & 0.1175 & 27.676 & 0.76465 & 0.07317 & 0.05813 & 3.234 \\
 & 0.1507 & 28.493 & 0.79053 & 0.06210 & 0.05156 & 3.238 \\
 & 0.1839 & 29.201 & 0.81209 & 0.05392 & 0.04684 & 3.244 \\
\hline
\multirow{5}{*}{CRDR} & 0.1137 & 28.118 & 0.78313 & 0.08699 & 0.07781 & 3.448 \\
 & 0.1315 & 28.672 & 0.80006 & 0.07815 & 0.07256 & 3.505 \\
 & 0.1517 & 29.193 & 0.81502 & 0.07012 & 0.06703 & 3.496 \\
 & 0.1751 & 29.700 & 0.82814 & 0.06242 & 0.06080 & 3.434 \\
 & 0.2017 & 30.155 & 0.83863 & 0.05536 & 0.05376 & 3.320 \\
\hline
\multirow{8}{*}{MS-ILLM} & 0.0043 & 20.604 & 0.52408 & 0.47883 & 0.28467 & 4.394 \\
 & 0.0081 & 21.932 & 0.55455 & 0.36388 & 0.22503 & 3.837 \\
 & 0.0441 & 25.501 & 0.67972 & 0.13795 & 0.10275 & 3.524 \\
 & 0.0798 & 26.889 & 0.73187 & 0.09683 & 0.08004 & 3.362 \\
 & 0.1489 & 28.510 & 0.79226 & 0.06465 & 0.05563 & 3.255 \\
 & 0.2797 & 30.549 & 0.85300 & 0.04036 & 0.04048 & 3.264 \\
 & 0.4089 & 31.919 & 0.88417 & 0.02885 & 0.03171 & 3.199 \\
 & 0.6566 & 33.908 & 0.91813 & 0.01703 & 0.02037 & 3.084 \\
\hline
\multirow{2}{*}{HiFiC} & 0.1468 & 27.899 & 0.79160 & 0.08189 & 0.07616 & 3.471 \\
 & 0.2930 & 29.791 & 0.84598 & 0.05618 & 0.05762 & 3.415 \\
\hline
\multirow{2}{*}{OSDiff} & 0.0935 & 24.213 & 0.69396 & 0.14146 & 0.07734 & 3.785 \\
 & 0.1360 & 25.006 & 0.72070 & 0.10832 & 0.05402 & 3.567 \\
\hline
\multirow{5}{*}{RDEIC} & 0.0231 & 22.137 & 0.59805 & 0.22599 & 0.15814 & 3.626 \\
 & 0.0417 & 23.531 & 0.65366 & 0.16163 & 0.11395 & 3.499 \\
 & 0.0639 & 24.859 & 0.69793 & 0.12067 & 0.07808 & 3.448 \\
 & 0.0896 & 25.781 & 0.72533 & 0.09633 & 0.05942 & 3.459 \\
 & 0.1199 & 26.502 & 0.74673 & 0.08105 & 0.04499 & 3.434 \\
\hline
\end{tabular}
}
\end{table}

\section{Supplementary FID Evaluation}
\label{app:fid_results}

FID \cite{heusel2017fid} is computed with \texttt{pytorch-fid} v0.3.0 from 2048-dimensional Inception-v3 features. Following HiFiC \cite{mentzer2020hific}, each image is partitioned into non-overlapping \(256\times256\) patches and a second patch grid shifted by 128 pixels in each spatial dimension. Because FID estimates differences between feature distributions and depends on sample count and the patch-extraction protocol, we report it on CLIC2020 (428 images) and DIV2K (100 images), but not Kodak (24 images). Fig.~\ref{fig:appendix_fid} reports directly evaluated FID points in the low-bitrate region relevant to AFP-GIC; Tables~\ref{tab:fid_values_clic} and~\ref{tab:fid_values_div2k} also include two reproduced Control-GIC points.

\begin{figure}[H]
\centering
\includegraphics[width=\columnwidth]{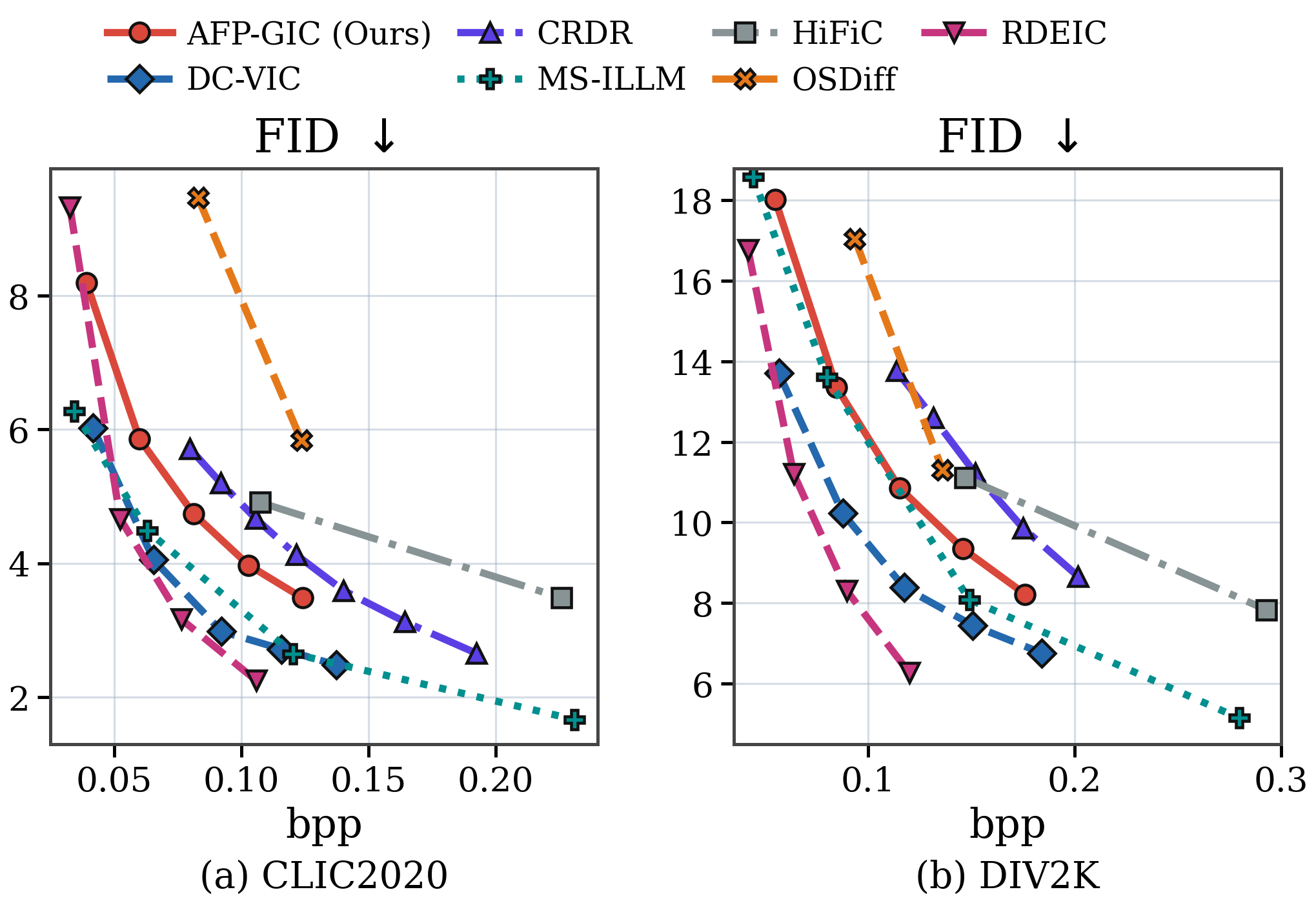}
\caption{FID--rate comparison on CLIC2020 and DIV2K for the learned methods in Fig.~\ref{fig:quantitative_main}. Lower is better.}
\label{fig:appendix_fid}
\end{figure}

\begin{table}[H]
\caption{FID values on CLIC2020. Within each row, the bpp and FID entries correspond from left to right.}
\label{tab:fid_values_clic}
\centering
\scriptsize
\setlength{\tabcolsep}{2pt}
\renewcommand{\arraystretch}{0.96}
\begin{tabular}{p{0.22\columnwidth}p{0.39\columnwidth}p{0.30\columnwidth}}
\hline
\rule{0pt}{1.8ex}Method & Evaluated bpp & FID \(\downarrow\) \\
\hline
AFP-GIC (Ours) & 0.0391, 0.0599, 0.0812, 0.1029, 0.1240 & 8.190, 5.858, 4.738, 3.970, 3.482 \\
DC-VIC & 0.0417, 0.0655, 0.0922, 0.1158, 0.1372 & 6.020, 4.051, 2.981, 2.717, 2.481 \\
CRDR & 0.0797, 0.0918, 0.1056, 0.1215, 0.1401, 0.1642, 0.1923 & 5.715, 5.204, 4.667, 4.127, 3.593, 3.126, 2.657 \\
MS-ILLM & 0.0343, 0.0629, 0.1203, 0.2310 & 6.275, 4.490, 2.649, 1.665 \\
HiFiC & 0.1074, 0.2258 & 4.907, 3.483 \\
OSDiff & 0.0829, 0.1235 & 9.463, 5.841 \\
RDEIC & 0.0326, 0.0524, 0.0764, 0.1059 & 9.315, 4.662, 3.162, 2.246 \\
Control-GIC & 0.0841, 0.1892 & 16.433, 3.789 \\
\hline
\end{tabular}
\end{table}

\begin{table}[H]
\caption{FID values on DIV2K. Within each row, the bpp and FID entries correspond from left to right.}
\label{tab:fid_values_div2k}
\centering
\scriptsize
\setlength{\tabcolsep}{2pt}
\renewcommand{\arraystretch}{0.96}
\begin{tabular}{p{0.22\columnwidth}p{0.39\columnwidth}p{0.30\columnwidth}}
\hline
\rule{0pt}{1.8ex}Method & Evaluated bpp & FID \(\downarrow\) \\
\hline
AFP-GIC (Ours) & 0.0549, 0.0846, 0.1151, 0.1459, 0.1760 & 18.015, 13.362, 10.861, 9.347, 8.220 \\
DC-VIC & 0.0566, 0.0876, 0.1175, 0.1507, 0.1839 & 13.712, 10.231, 8.382, 7.443, 6.761 \\
CRDR & 0.1137, 0.1315, 0.1517, 0.1751, 0.2017 & 13.770, 12.603, 11.233, 9.857, 8.661 \\
MS-ILLM & 0.0441, 0.0798, 0.1489, 0.2797 & 18.591, 13.614, 8.088, 5.143 \\
HiFiC & 0.1468, 0.2930 & 11.112, 7.831 \\
OSDiff & 0.0935, 0.1360 & 17.041, 11.308 \\
RDEIC & 0.0417, 0.0639, 0.0896, 0.1199 & 16.780, 11.202, 8.306, 6.265 \\
Control-GIC & 0.0868, 0.1961 & 41.658, 9.713 \\
\hline
\end{tabular}
\end{table}

\section{Partial Comparison with Control-GIC}
\label{app:control_gic}

Using the released checkpoint, we evaluate two granularity settings reported in \cite[Appendix A.3]{li2025controlgic}. The fine-, medium-, and coarse-scale weights are 0, 0.23, and 0.77 in the first setting, and 0.10, 0.67, and 0.23 in the second. Table~\ref{tab:control_gic_partial} reports actual bpp under the same metric pipeline as the other baselines. No BD metrics are computed from these two points.

\begin{table}[H]
\caption{Partial comparison with Control-GIC using two reproduced operating points.}
\label{tab:control_gic_partial}
\centering
\scriptsize
\setlength{\tabcolsep}{1.6pt}
\renewcommand{\arraystretch}{0.92}
\resizebox{\columnwidth}{!}{
\begin{tabular}{lcccccc}
\hline
\rule{0pt}{1.8ex}Dataset & bpp & PSNR $\uparrow$ & SSIM $\uparrow$ & LPIPS $\downarrow$ & DISTS $\downarrow$ & NIQE $\downarrow$ \\
\hline
Kodak & 0.0866 & 22.304 & 0.55771 & 0.19583 & 0.15858 & 3.572 \\
 & 0.1957 & 25.717 & 0.67888 & 0.07061 & 0.08397 & 2.980 \\
CLIC2020 & 0.0841 & 25.034 & 0.68966 & 0.13909 & 0.11038 & 3.696 \\
 & 0.1892 & 28.270 & 0.77376 & 0.05882 & 0.05450 & 3.562 \\
DIV2K & 0.0868 & 22.266 & 0.59302 & 0.18334 & 0.13748 & 3.316 \\
 & 0.1961 & 25.926 & 0.71425 & 0.07056 & 0.06194 & 3.042 \\
\hline
\end{tabular}
}
\end{table}

\section{Bitrate-Normalized Quality--Complexity Trade-Off}
\label{app:quality_complexity}

Fig.~\ref{fig:quality_complexity_tradeoff} plots mean BD-PSNR against inference parameter count and mean BD-NIQE against decoder latency, using the quality results in Table~\ref{tab:bd_vs_msillm} and the resource measurements in Table~\ref{tab:runtime_learned_gpu}. BD-PSNR and BD-NIQE are computed relative to MS-ILLM over each pairwise common bitrate interval and then averaged equally across Kodak, CLIC2020, and DIV2K. AFP-GIC has the smallest model (120.6M parameters), with mean BD-PSNR of +0.120 dB and mean BD-NIQE of -0.176. CRDR achieves higher mean BD-PSNR but has greater decoder latency and positive mean BD-NIQE; MS-ILLM decodes fastest but uses more parameters. The figure shows complementary quality and complexity dimensions rather than a scalar ranking.

\begin{figure}[H]
\centering
\includegraphics[width=\columnwidth]{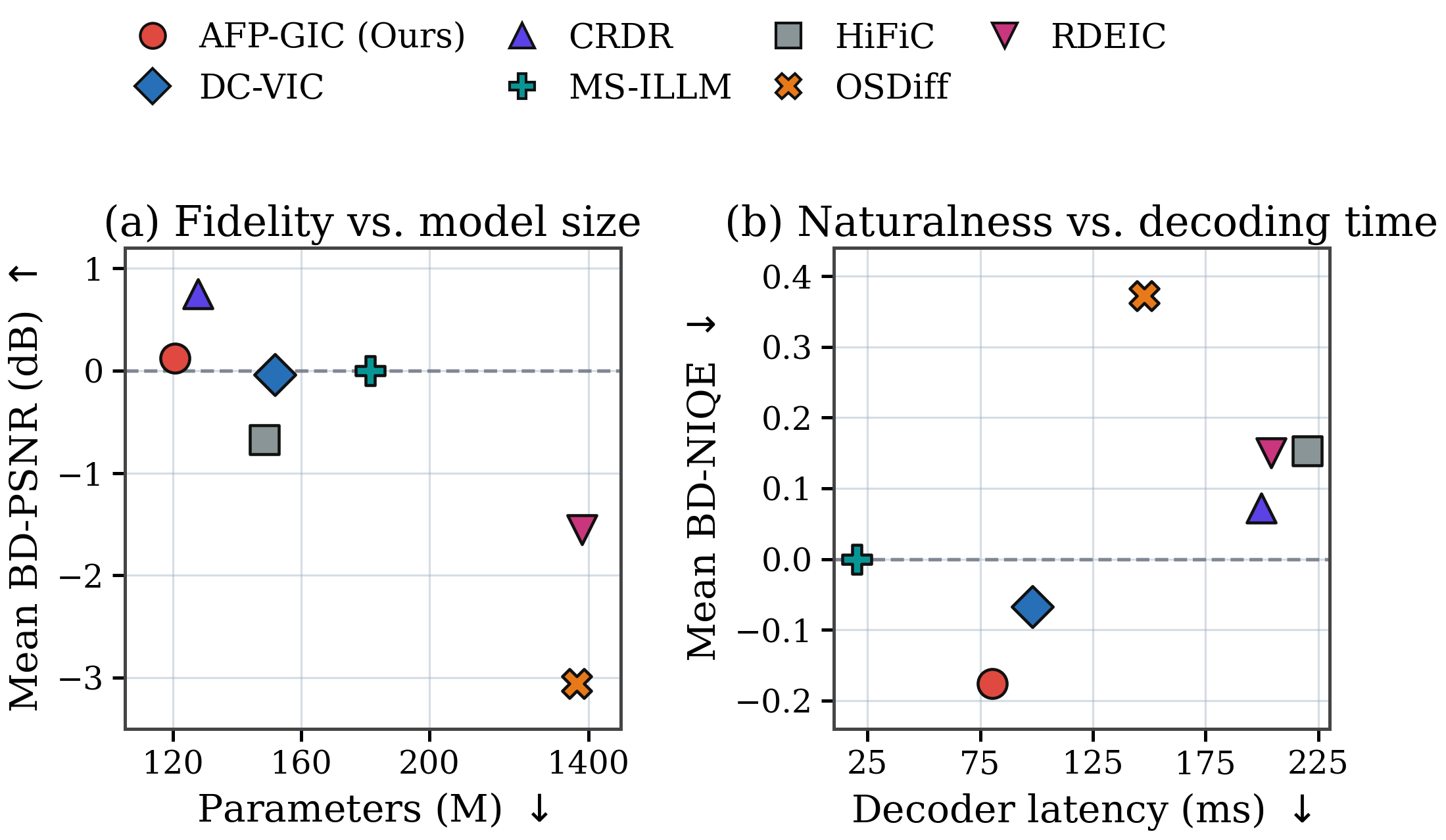}
\caption{Bitrate-normalized quality--complexity trade-off averaged over Kodak, CLIC2020, and DIV2K: (a) mean BD-PSNR versus inference parameter count; (b) mean BD-NIQE versus decoder latency. BD values use method-specific common bitrate intervals relative to MS-ILLM; the nonlinear parameter axis retains the actual tick values.}
\label{fig:quality_complexity_tradeoff}
\end{figure}

\section{Prior Estimator Substitution Diagnostic}
\label{app:prior_estimator_substitution}

To assess the Prior Estimator within the AFP-GIC architecture, we fix the bitstream, \(\hat{\mathbf{y}}\), controls, SFT features, decoder, and model parameters for all 24 Kodak images at each of the five operating points. We replace \(\hat{\mathbf{p}}\) with (i) zeros, (ii) its channel-wise spatial mean, \(\bar{\mathbf{p}}_{c,u,v}=(H_{\hat p}W_{\hat p})^{-1}\sum_{u'=1}^{H_{\hat p}}\sum_{v'=1}^{W_{\hat p}}\hat{\mathbf{p}}_{c,u',v'}\), or (iii) the ground-truth fused prior \(\mathbf{p}\) as a non-deployable oracle. Here, \(H_{\hat p}\) and \(W_{\hat p}\) denote the height and width of the predicted prior feature map. Table~\ref{tab:prior_estimator_substitution_bd} reports BD differences relative to AFP-GIC.

\begin{table}[H]
\centering
\caption{Bj{\o}ntegaard delta metrics \cite{bjontegaard2001bdpsnr} relative to AFP-GIC for prior substitution on Kodak. ``\(\uparrow\)'' indicates higher is better; ``\(\downarrow\)'' indicates lower is better.}
\label{tab:prior_estimator_substitution_bd}
\footnotesize
\setlength{\tabcolsep}{1.5pt}
\renewcommand{\arraystretch}{0.96}
\begin{tabular*}{\columnwidth}{@{\extracolsep{\fill}}lccccc@{}}
\hline
\rule{0pt}{1.8ex}Variant & PSNR $\uparrow$ & SSIM $\uparrow$ & LPIPS $\downarrow$ & DISTS $\downarrow$ & NIQE $\downarrow$ \\
\hline
AFP-GIC & 0.0000 & 0.00000 & 0.00000 & 0.00000 & 0.00000 \\
Zero prior & -1.6389 & -0.00364 & 0.02489 & 0.02288 & 0.20278 \\
Spatial-mean prior & -1.4684 & -0.00898 & 0.01907 & 0.01979 & 0.26289 \\
GT fused prior (oracle) & 0.3861 & 0.03558 & -0.00116 & -0.00084 & 0.03869 \\
\hline
\end{tabular*}
\end{table}

Zero- and spatial-mean-prior substitutions degrade all five metrics, indicating that the learned spatial prediction is useful. The GT fused-prior oracle improves the four full-reference metrics but slightly raises NIQE, suggesting that more accurate prior prediction can improve fidelity without uniformly improving every criterion.

\section{Supplementary Comparison Under the DC-VIC Training Pipeline}
\label{app:dcvic_training_pipeline}

Fig.~\ref{fig:dcvic_training_pipeline_comparison} compares the released DC-VIC checkpoint, our DC-VIC reproduction, and AFP-GIC on CLIC2020. Our reproduction and AFP-GIC follow the same three-stage procedure based on the official DC-VIC training guidelines and use the same training-data protocol.

\begin{figure}[!t]
\centering
\includegraphics[width=\columnwidth]{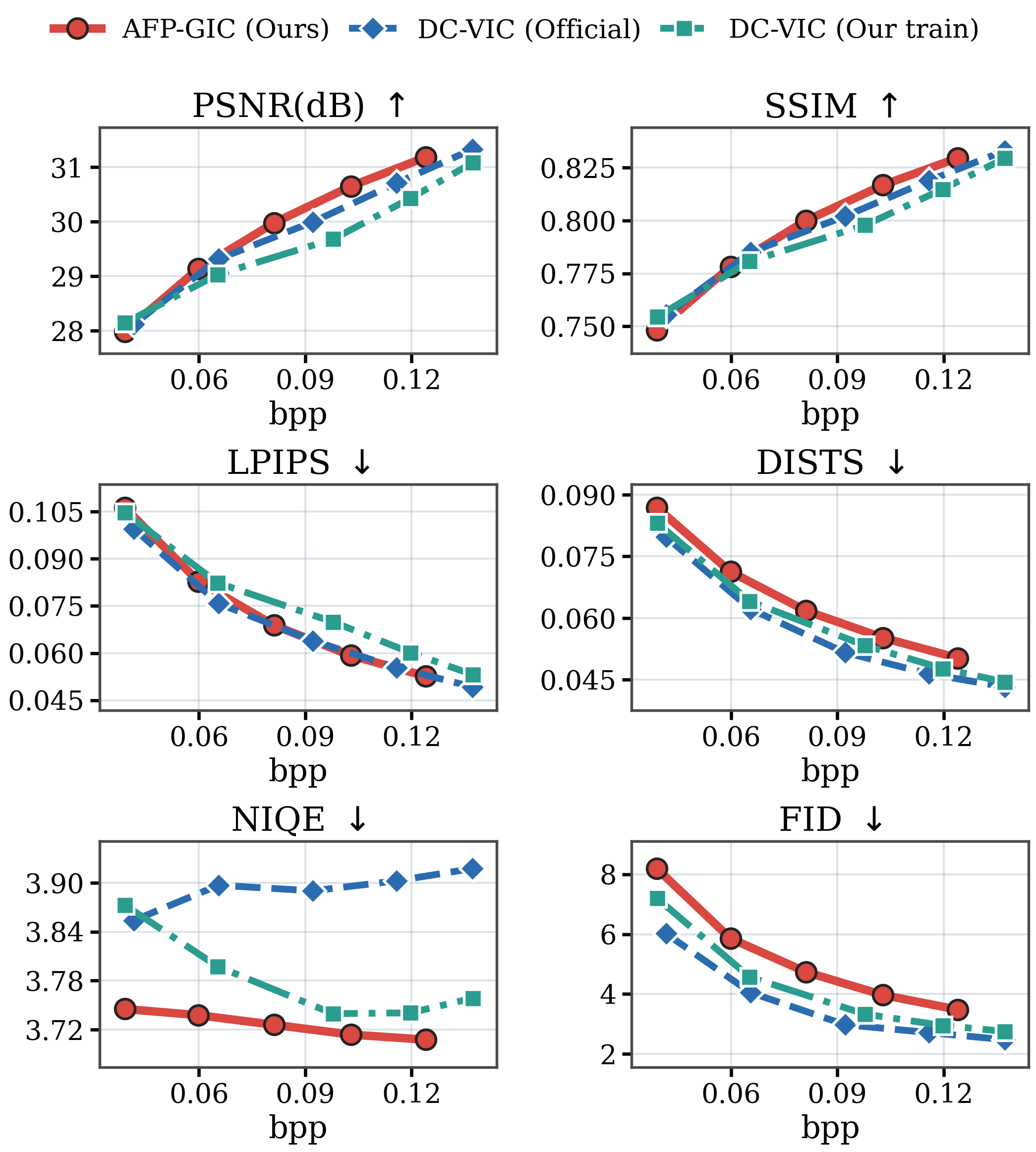}
\caption{Supplementary CLIC2020 comparison among the released DC-VIC checkpoint, our locally trained DC-VIC model, and AFP-GIC, including distribution-level FID.}
\label{fig:dcvic_training_pipeline_comparison}
\end{figure}

The released DC-VIC checkpoint outperforms our reproduction on all reported metrics except NIQE. AFP-GIC, trained with the same procedure and training data as our reproduction, nevertheless achieves higher PSNR over most of the common bitrate range, higher SSIM on average, and consistently lower NIQE than the released checkpoint, while remaining competitive in LPIPS. Overall, AFP-GIC shows favorable PSNR, SSIM, and NIQE trends against the stronger released model, while the released DC-VIC checkpoint remains stronger in DISTS and FID. Crucially, under matched local training protocols, AFP-GIC consistently outperforms our DC-VIC reproduction across most operating bitrates in PSNR, SSIM, LPIPS, and NIQE, though AFP-GIC additionally benefits from frozen AdaCode pretraining.

\bibliographystyle{IEEEtran}
\bibliography{references}

{\biographyfont\par\vspace{2.8\baselineskip}}
\begin{IEEEbiography}[{\includegraphics[width=1in,height=1.25in,clip,keepaspectratio]{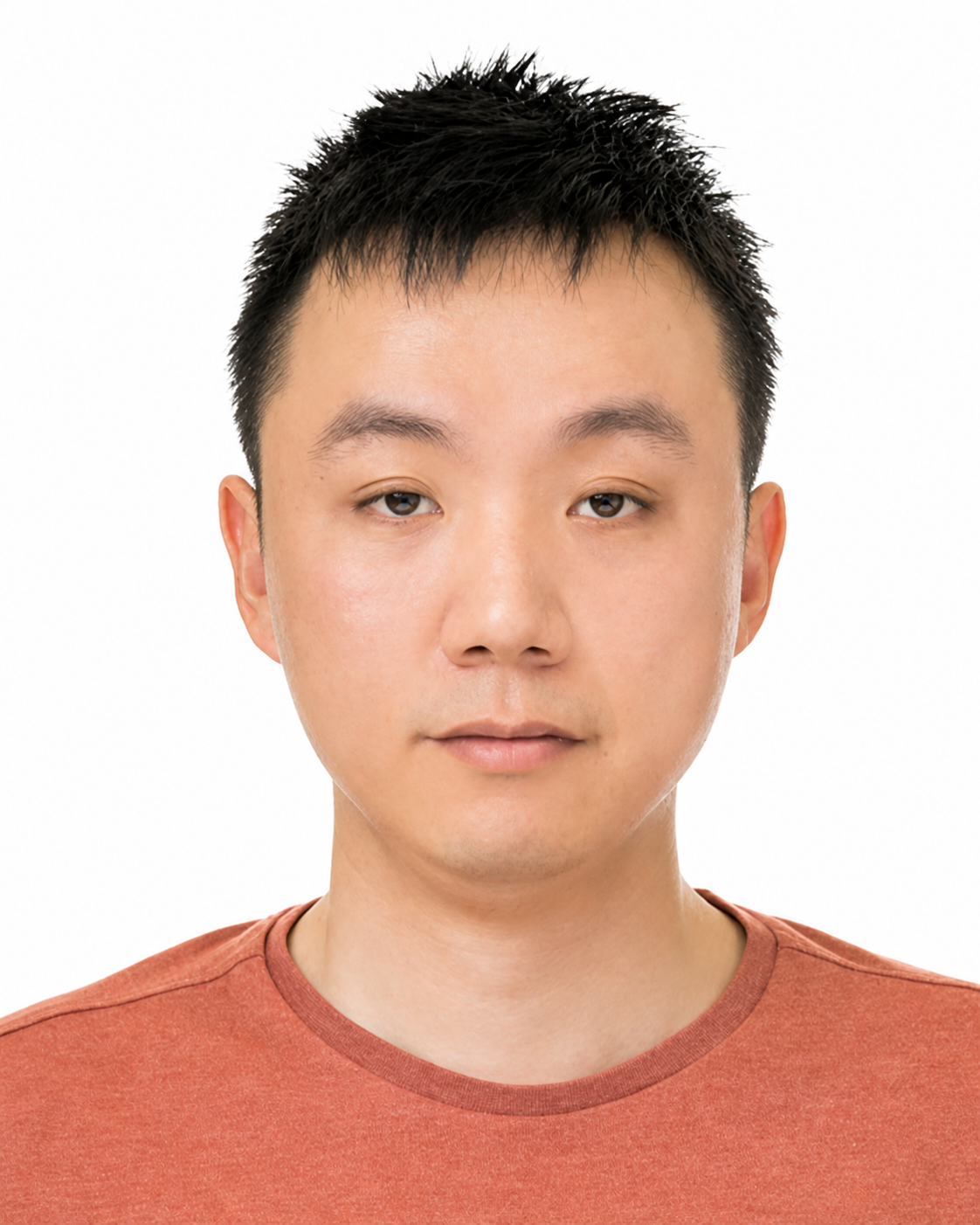}}]{Yifei Pei}
(Member, IEEE) received the M.S. degree in Computer Science and Engineering from Santa Clara University, Santa Clara, CA, USA, in 2020, where he is currently pursuing the Ph.D. degree in the same field. His research interests include image and video coding, learned compression, generative models, and compressive sensing. He is the first author of four IEEE conference papers, including three in IEEE ISCAS and one in IEEE VCIP, and the first named inventor on two granted U.S. patents related to low-bitrate image coding and video compressed sensing.
\end{IEEEbiography}

{\biographyfont\par\vspace{3.4\baselineskip}}
\begin{IEEEbiography}[{\includegraphics[width=1in,height=1.25in,clip,keepaspectratio]{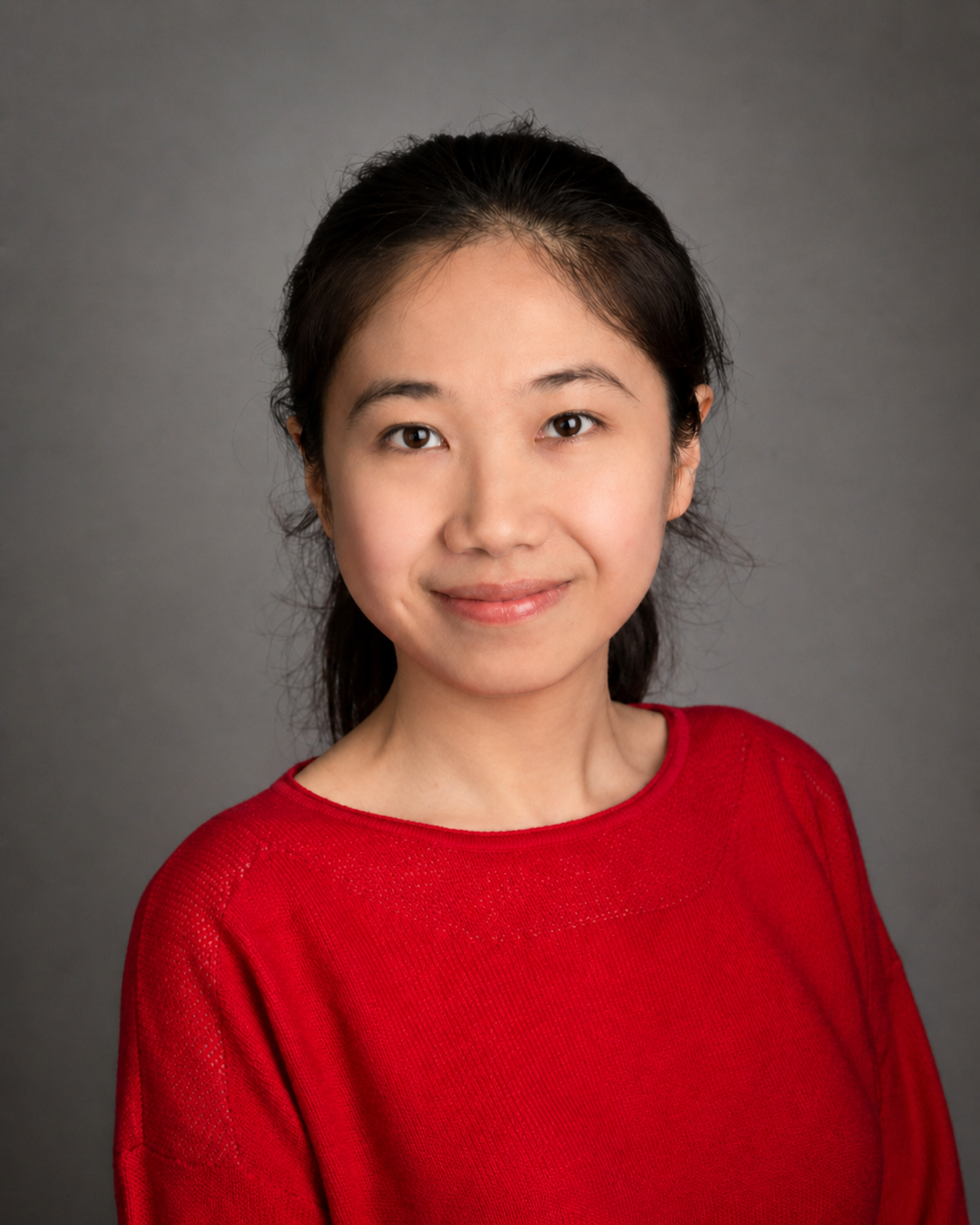}}]{Ying Liu}
(S'11--M'13) received the B.S. degree in communications engineering from Beijing University of Posts and Telecommunications, Beijing, China, in 2006, and the M.S. and Ph.D. degrees in electrical engineering from The State University of New York at Buffalo, Buffalo, NY, USA, in 2008 and 2012, respectively. She is currently an Associate Professor with the Department of Computer Science and Engineering, Santa Clara University, Santa Clara, CA, USA. Her main research interests include deep learning, image and video coding, coding for machines, point cloud coding, and multimodal large language models. Dr. Liu served as an Associate Editor for IEEE Transactions on Circuits and Systems for Video Technology from 2022 to 2025 and has served as an Associate Editor for Journal on Image and Video Processing since 2025.
\end{IEEEbiography}

\newpage
\begin{IEEEbiography}[{\includegraphics[width=1in,height=1.25in,clip,keepaspectratio]{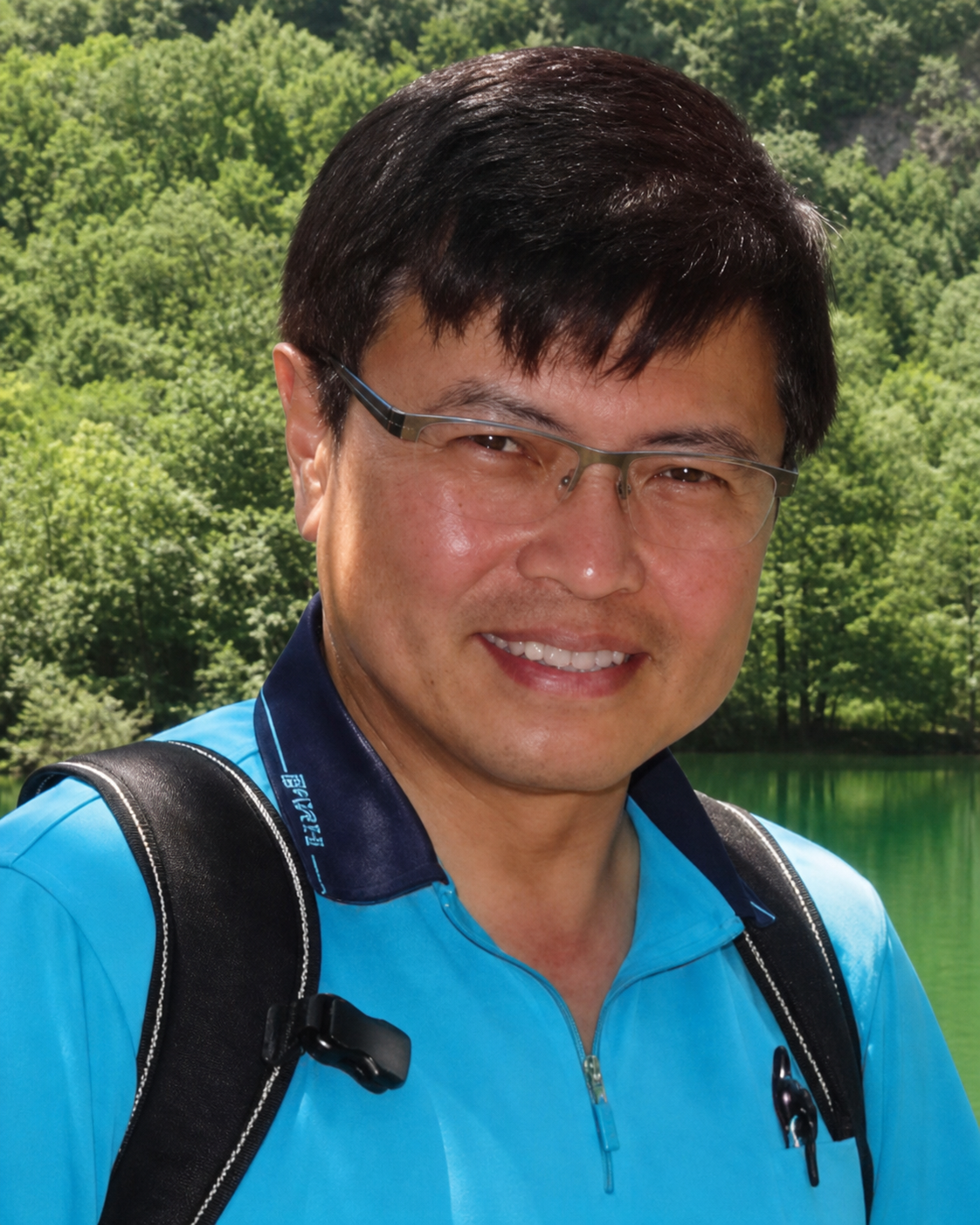}}]{Nam Ling}
(S'88--M'90--SM'99--F'08--LF'22) received the B.Eng. degree from the National University of Singapore, Singapore, in 1981, and the M.S. and Ph.D. degrees from the University of Louisiana at Lafayette, Lafayette, LA, USA, in 1985 and 1989, respectively. He is currently the Wilmot J. Nicholson Family Chair Professor with the Department of Computer Science and Engineering and the Associate Dean for Research for the School of Engineering at Santa Clara University, USA. He served as Department Chair from 2010 to 2023, Associate Dean for Graduate Studies/Research/Faculty Development from 2002 to 2010, and was the Sanfilippo Family Chair Professor from 2010 to 2020. He is/was also a Chair/Distinguished/Guest/Consulting Professor with several universities internationally. He has authored or coauthored over 350 publications and seven adopted standard contributions. He has been granted more than 20 U.S./European/PCT patents. He has delivered more than 120 invited colloquia worldwide. He is an IEEE Fellow due to his contributions to video coding algorithms and architectures. He is also an IET Fellow and an AAIA Fellow. He was named as an IEEE Distinguished Lecturer twice and was also an APSIPA Distinguished Lecturer. He was a recipient of the IEEE ICCE Best Paper Award (First Place) and the Umedia Best/Excellent Paper Award thrice. He received six awards from Santa Clara University, four at the university level and two at the school/college level. He was a Keynote Speaker for IEEE APCCAS, VCVP twice, JCPC, IEEE ICAST, IEEE ICIEA, IET FC and Umedia, IEEE Umedia, IEEE ICCIT, ICNLP/SSPS/CVPS, and Workshop at XUPT twice. He has served as the General Chair/Co-Chair for IEEE Hot Chips, VCVP twice, IEEE ICME, IEEE VCIP, IEEE SiPS, SocialSec, and Umedia six times. He has also served as the Technical Program (TPC) Co-Chair for IEEE ISCAS twice, APSIPA ASC, IEEE APCCAS, IEEE SiPS twice, DCV, and IEEE VCIP. He is currently a TPC Chair for IEEE ICME 2026. He was the Technical Committee Chair for IEEE CASCOM TC and IEEE TCMM, and the Chair for the APSIPA U.S. Chapter. He served as a Guest Editor or an Associate Editor for the IEEE Transactions on Circuits and Systems--I: Regular Papers, the IEEE Journal of Selected Topics in Signal Processing, the Journal of Signal Processing Systems (Springer), Multidimensional Systems and Signal Processing (Springer), and other journals.
\end{IEEEbiography}

\EOD
\end{document}


\twocolumn[
\begin{center}
{\LARGE Supplementary Material for \textit{Adaptive Fused Prior Transfer for Controllable Generative Image Compression}\par}
\vspace{3pt}
{\normalsize Yifei Pei, Ying Liu, and Nam Ling\par}
{\small Department of Computer Science and Engineering, Santa Clara University, Santa Clara, CA 95053 USA\par}
\end{center}
\vspace{2pt}
]
\thispagestyle{plain}
\pagestyle{plain}
\raggedbottom

\section{SNR and MS-SSIM Results}
The operating points, source images, and saved RGB reconstructions are identical to those used for Fig.~3 of the main manuscript. No model was retrained and no operating point was added. For each image, SNR is computed as
\begin{equation}
\mathrm{SNR}=10\log_{10}\!\left(\frac{\operatorname{mean}(x^2)}{\operatorname{mean}[(x-\hat{x})^2]}\right),
\end{equation}
Both means cover all RGB samples. MS-SSIM is computed on RGB images in $[0,1]$ using \texttt{pytorch-msssim} v0.2.1 with its default five-scale weights $(0.0448, 0.2856, 0.3001, 0.2363, 0.1333)$. Dataset-level values are arithmetic means of the per-image measurements.

Fig.~\ref{fig:supp_snr_msssim_rd} shows the corresponding rate--quality curves, Table~\ref{tab:supp_bd_snr_msssim} reports BD-SNR and BD-MS-SSIM relative to MS-ILLM, and Tables~\ref{tab:supp_kodak}--\ref{tab:supp_div2k} provide the point-wise values. The BD calculation follows the same pairwise common-bitrate and log-bpp fitting protocol as Table~4 of the main manuscript. Because all methods are evaluated on the same source images within each dataset, dataset-average SNR differs from dataset-average PSNR by a dataset-specific constant; accordingly, BD-SNR equals BD-PSNR up to numerical precision.

\begin{strip}
\centering
\includegraphics[width=0.84\textwidth]{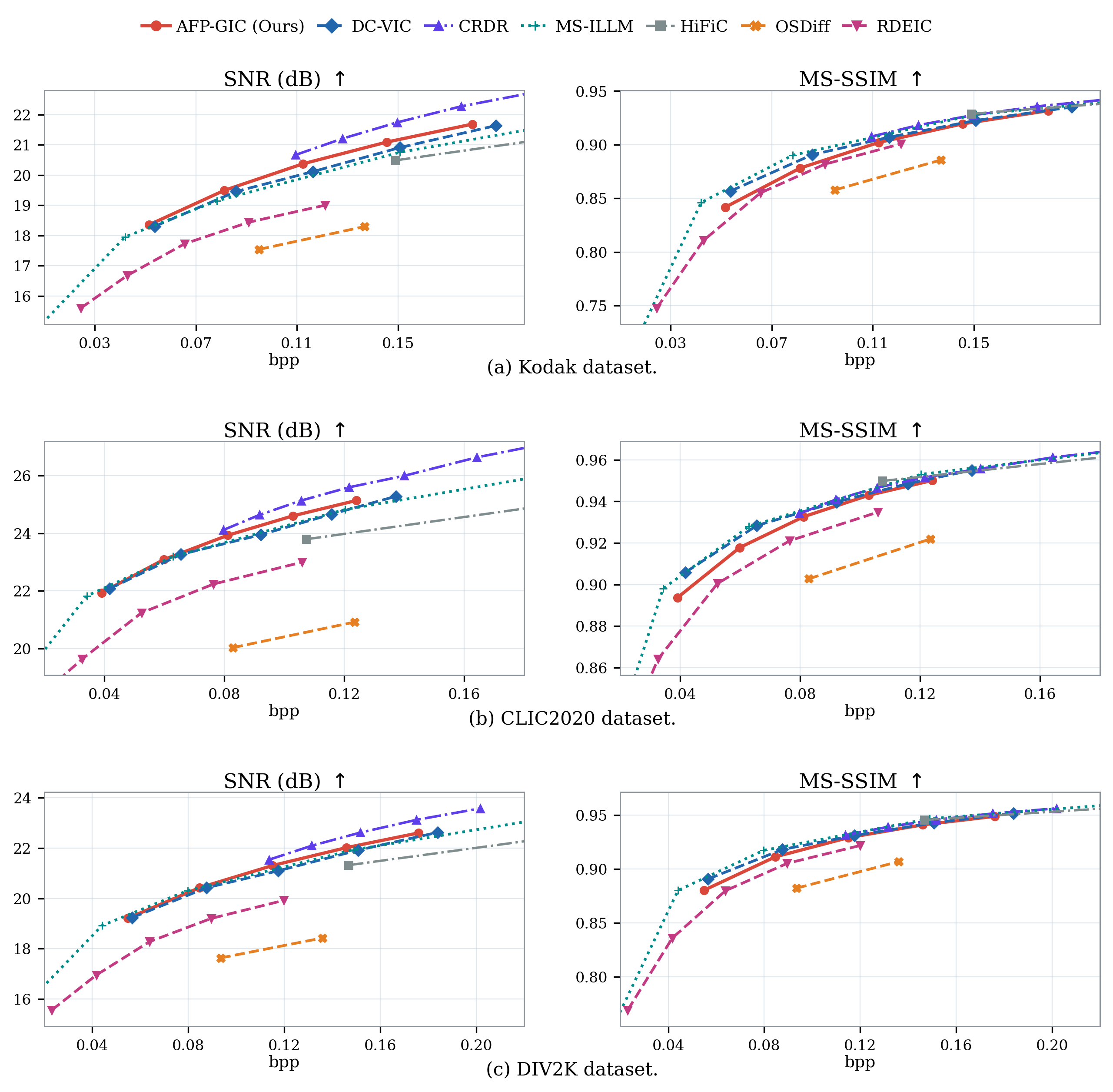}
\captionof{figure}{SNR and MS-SSIM comparisons at the learned-method operating points used in Fig.~3 of the main manuscript on (a) Kodak, (b) CLIC2020, and (c) DIV2K. Higher is better.}
\label{fig:supp_snr_msssim_rd}
\end{strip}

\clearpage

\begin{table}[!t]
\centering
\caption{Bj{\o}ntegaard delta metrics for SNR and MS-SSIM relative to MS-ILLM. ``\(\uparrow\)'' indicates higher is better.}
\label{tab:supp_bd_snr_msssim}
\footnotesize
\setlength{\tabcolsep}{2.5pt}
\renewcommand{\arraystretch}{0.84}
\begin{tabular*}{\columnwidth}{@{\extracolsep{\fill}}lcc@{}}
\hline
\rule{0pt}{1.8ex}Method & BD-SNR (dB) $\uparrow$ & BD-MS-SSIM $\uparrow$ \\
\hline
\multicolumn{3}{l}{\textit{Kodak}} \\
\rule{0pt}{1.8ex}MS-ILLM & 0.0000 & 0.00000 \\
\rule{0pt}{1.8ex}DC-VIC & 0.0659 & -0.00643 \\
\rule{0pt}{1.8ex}CRDR & 0.9689 & -0.00208 \\
\rule{0pt}{1.8ex}HiFiC & -0.3895 & -0.00332 \\
\rule{0pt}{1.8ex}OSDiff & -2.1296 & -0.04195 \\
\rule{0pt}{1.8ex}RDEIC & -1.2087 & -0.02968 \\
\rule{0pt}{1.8ex}AFP-GIC (Ours) & 0.2850 & -0.01244 \\
\hline
\multicolumn{3}{l}{\textit{CLIC2020}} \\
\rule{0pt}{1.8ex}MS-ILLM & 0.0000 & 0.00000 \\
\rule{0pt}{1.8ex}DC-VIC & -0.0634 & -0.00243 \\
\rule{0pt}{1.8ex}CRDR & 0.7225 & -0.00283 \\
\rule{0pt}{1.8ex}HiFiC & -0.9147 & -0.00269 \\
\rule{0pt}{1.8ex}OSDiff & -3.8811 & -0.03513 \\
\rule{0pt}{1.8ex}RDEIC & -1.8024 & -0.02487 \\
\rule{0pt}{1.8ex}AFP-GIC (Ours) & 0.0632 & -0.00711 \\
\hline
\multicolumn{3}{l}{\textit{DIV2K}} \\
\rule{0pt}{1.8ex}MS-ILLM & 0.0000 & 0.00000 \\
\rule{0pt}{1.8ex}DC-VIC & -0.1275 & -0.00470 \\
\rule{0pt}{1.8ex}CRDR & 0.5713 & -0.00336 \\
\rule{0pt}{1.8ex}HiFiC & -0.7207 & -0.00437 \\
\rule{0pt}{1.8ex}OSDiff & -3.1581 & -0.04040 \\
\rule{0pt}{1.8ex}RDEIC & -1.6605 & -0.03159 \\
\rule{0pt}{1.8ex}AFP-GIC (Ours) & 0.0113 & -0.00823 \\
\hline
\end{tabular*}
\end{table}

\begin{table}[!t]
\centering
\caption{Point-wise SNR and MS-SSIM for Fig.~3 on Kodak.}
\label{tab:supp_kodak}
\footnotesize
\setlength{\tabcolsep}{3.0pt}
\renewcommand{\arraystretch}{0.80}
\begin{tabular*}{\columnwidth}{@{\extracolsep{\fill}}lccc@{}}
\hline
\rule{0pt}{1.8ex}Method & bpp & SNR (dB) $\uparrow$ & MS-SSIM $\uparrow$ \\
\hline
\rule{0pt}{1.8ex}AFP-GIC (Ours) & 0.0516 & 18.356 & 0.84168 \\
 & 0.0812 & 19.491 & 0.87817 \\
 & 0.1124 & 20.372 & 0.90213 \\
 & 0.1456 & 21.093 & 0.91945 \\
 & 0.1794 & 21.686 & 0.93170 \\
\addlinespace[0.2pt]
\rule{0pt}{1.8ex}DC-VIC & 0.0537 & 18.300 & 0.85675 \\
 & 0.0860 & 19.457 & 0.89060 \\
 & 0.1164 & 20.115 & 0.90697 \\
 & 0.1508 & 20.910 & 0.92260 \\
 & 0.1888 & 21.638 & 0.93509 \\
\addlinespace[0.2pt]
\rule{0pt}{1.8ex}CRDR & 0.1095 & 20.667 & 0.90758 \\
 & 0.1280 & 21.206 & 0.91819 \\
 & 0.1497 & 21.743 & 0.92754 \\
 & 0.1750 & 22.275 & 0.93573 \\
 & 0.2043 & 22.745 & 0.94256 \\
\addlinespace[0.2pt]
\rule{0pt}{1.8ex}MS-ILLM & 0.0039 & 13.764 & 0.60524 \\
 & 0.0071 & 14.903 & 0.66964 \\
 & 0.0420 & 17.953 & 0.84603 \\
 & 0.0783 & 19.145 & 0.89019 \\
 & 0.1508 & 20.755 & 0.92822 \\
 & 0.2935 & 22.857 & 0.95969 \\
 & 0.4300 & 24.257 & 0.97190 \\
 & 0.7068 & 26.435 & 0.98370 \\
\addlinespace[0.2pt]
\rule{0pt}{1.8ex}HiFiC & 0.1490 & 20.487 & 0.92884 \\
 & 0.3033 & 22.321 & 0.95655 \\
\addlinespace[0.2pt]
\rule{0pt}{1.8ex}OSDiff & 0.0950 & 17.540 & 0.85785 \\
 & 0.1369 & 18.296 & 0.88567 \\
\addlinespace[0.2pt]
\rule{0pt}{1.8ex}RDEIC & 0.0245 & 15.595 & 0.74761 \\
 & 0.0429 & 16.678 & 0.81066 \\
 & 0.0655 & 17.723 & 0.85505 \\
 & 0.0910 & 18.438 & 0.88176 \\
 & 0.1211 & 19.002 & 0.90078 \\
\hline
\end{tabular*}
\end{table}

\begin{table}[!t]
\centering
\caption{Point-wise SNR and MS-SSIM for Fig.~3 on CLIC2020.}
\label{tab:supp_clic2020}
\footnotesize
\setlength{\tabcolsep}{3.0pt}
\renewcommand{\arraystretch}{0.80}
\begin{tabular*}{\columnwidth}{@{\extracolsep{\fill}}lccc@{}}
\hline
\rule{0pt}{1.8ex}Method & bpp & SNR (dB) $\uparrow$ & MS-SSIM $\uparrow$ \\
\hline
\rule{0pt}{1.8ex}AFP-GIC (Ours) & 0.0391 & 21.936 & 0.89368 \\
 & 0.0599 & 23.086 & 0.91778 \\
 & 0.0812 & 23.928 & 0.93265 \\
 & 0.1029 & 24.602 & 0.94298 \\
 & 0.1240 & 25.134 & 0.95008 \\
\addlinespace[0.2pt]
\rule{0pt}{1.8ex}DC-VIC & 0.0417 & 22.073 & 0.90581 \\
 & 0.0655 & 23.266 & 0.92838 \\
 & 0.0922 & 23.946 & 0.93976 \\
 & 0.1158 & 24.662 & 0.94852 \\
 & 0.1372 & 25.278 & 0.95503 \\
\addlinespace[0.2pt]
\rule{0pt}{1.8ex}CRDR & 0.0797 & 24.114 & 0.93413 \\
 & 0.0918 & 24.636 & 0.94082 \\
 & 0.1056 & 25.131 & 0.94659 \\
 & 0.1215 & 25.591 & 0.95158 \\
 & 0.1401 & 25.992 & 0.95569 \\
 & 0.1642 & 26.629 & 0.96119 \\
 & 0.1923 & 27.209 & 0.96576 \\
\addlinespace[0.2pt]
\rule{0pt}{1.8ex}MS-ILLM & 0.0036 & 16.666 & 0.72115 \\
 & 0.0062 & 18.120 & 0.77099 \\
 & 0.0343 & 21.831 & 0.89803 \\
 & 0.0629 & 23.189 & 0.92797 \\
 & 0.1203 & 24.821 & 0.95306 \\
 & 0.2310 & 26.786 & 0.97226 \\
 & 0.3309 & 28.062 & 0.97958 \\
 & 0.5121 & 29.835 & 0.98619 \\
\addlinespace[0.2pt]
\rule{0pt}{1.8ex}HiFiC & 0.1074 & 23.791 & 0.94984 \\
 & 0.2258 & 25.526 & 0.96822 \\
\addlinespace[0.2pt]
\rule{0pt}{1.8ex}OSDiff & 0.0829 & 20.035 & 0.90272 \\
 & 0.1235 & 20.920 & 0.92195 \\
\addlinespace[0.2pt]
\rule{0pt}{1.8ex}RDEIC & 0.0174 & 18.230 & 0.81247 \\
 & 0.0326 & 19.636 & 0.86406 \\
 & 0.0524 & 21.242 & 0.90049 \\
 & 0.0764 & 22.239 & 0.92104 \\
 & 0.1059 & 22.993 & 0.93481 \\
\hline
\end{tabular*}
\end{table}

\begin{table}[!t]
\centering
\caption{Point-wise SNR and MS-SSIM for Fig.~3 on DIV2K.}
\label{tab:supp_div2k}
\footnotesize
\setlength{\tabcolsep}{3.0pt}
\renewcommand{\arraystretch}{0.80}
\begin{tabular*}{\columnwidth}{@{\extracolsep{\fill}}lccc@{}}
\hline
\rule{0pt}{1.8ex}Method & bpp & SNR (dB) $\uparrow$ & MS-SSIM $\uparrow$ \\
\hline
\rule{0pt}{1.8ex}AFP-GIC (Ours) & 0.0549 & 19.214 & 0.88023 \\
 & 0.0846 & 20.420 & 0.91120 \\
 & 0.1151 & 21.303 & 0.92899 \\
 & 0.1459 & 22.016 & 0.94083 \\
 & 0.1760 & 22.601 & 0.94875 \\
\addlinespace[0.2pt]
\rule{0pt}{1.8ex}DC-VIC & 0.0566 & 19.235 & 0.89064 \\
 & 0.0876 & 20.420 & 0.91803 \\
 & 0.1175 & 21.096 & 0.93093 \\
 & 0.1507 & 21.914 & 0.94276 \\
 & 0.1839 & 22.621 & 0.95154 \\
\addlinespace[0.2pt]
\rule{0pt}{1.8ex}CRDR & 0.1137 & 21.538 & 0.93117 \\
 & 0.1315 & 22.092 & 0.93892 \\
 & 0.1517 & 22.614 & 0.94559 \\
 & 0.1751 & 23.121 & 0.95133 \\
 & 0.2017 & 23.575 & 0.95609 \\
\addlinespace[0.2pt]
\rule{0pt}{1.8ex}MS-ILLM & 0.0043 & 14.024 & 0.63852 \\
 & 0.0081 & 15.352 & 0.71212 \\
 & 0.0441 & 18.922 & 0.88015 \\
 & 0.0798 & 20.310 & 0.91720 \\
 & 0.1489 & 21.930 & 0.94653 \\
 & 0.2797 & 23.969 & 0.96897 \\
 & 0.4089 & 25.340 & 0.97778 \\
 & 0.6566 & 27.328 & 0.98590 \\
\addlinespace[0.2pt]
\rule{0pt}{1.8ex}HiFiC & 0.1468 & 21.319 & 0.94542 \\
 & 0.2930 & 23.212 & 0.96622 \\
\addlinespace[0.2pt]
\rule{0pt}{1.8ex}OSDiff & 0.0935 & 17.633 & 0.88224 \\
 & 0.1360 & 18.426 & 0.90676 \\
\addlinespace[0.2pt]
\rule{0pt}{1.8ex}RDEIC & 0.0231 & 15.558 & 0.76917 \\
 & 0.0417 & 16.952 & 0.83589 \\
 & 0.0639 & 18.279 & 0.87984 \\
 & 0.0896 & 19.202 & 0.90516 \\
 & 0.1199 & 19.922 & 0.92168 \\
\hline
\end{tabular*}
\end{table}

\nocite{blau2018perceptiondistortion,blau2019rdp}
\twocolumn[{
\begin{minipage}{\textwidth}
\centering
\captionof{table}{Taxonomy, Characteristics, Strengths, and Limitations of the Evaluated Image Quality Metrics.}
\label{tab:metrics_taxonomy}
\footnotesize
\setlength{\tabcolsep}{4.5pt}
\renewcommand{\arraystretch}{1.18}
\begin{tabularx}{\textwidth}{>{\raggedright\arraybackslash}p{2.3cm} >{\raggedright\arraybackslash}p{2.05cm} >{\raggedright\arraybackslash}p{2.55cm} >{\raggedright\arraybackslash}X >{\raggedright\arraybackslash}X}
\toprule
\textbf{Metric Category} & \textbf{Metric} & \textbf{Input Domain} & \textbf{Key Strengths} & \textbf{Limitations at Very-Low Bitrates} \\
\midrule
\multirow{2}{*}{\shortstack[l]{\textbf{Reference}\\\textbf{Fidelity}}}
& PSNR~\cite{wang2009mse} / SNR & Pixel-domain Full-Reference & Mathematically rigorous; standard benchmark across classical and neural codecs. & Strongly favors pixel alignment; penalizes valid generative textures, causing over-smoothed outputs. \\
\cmidrule{2-5}
& SSIM~\cite{wang2004ssim} / MS-SSIM~\cite{wang2003msssim} & Structural Full-Reference & Better aligns with Human Visual System (HVS) perception of structural contours than MSE. & Highly sensitive to sub-pixel spatial shifts and realistic texture synthesis. \\
\midrule
\multirow{2}{*}{\shortstack[l]{\textbf{Full-Reference}\\\textbf{Perceptual}}}
& LPIPS~\cite{zhang2018lpips} & Deep Feature Full-Reference & Well-aligned with human visual judgment of blur and compression artifacts. & Enforces localized feature alignment; penalizes realistic but unaligned fine details. \\
\cmidrule{2-5}
& DISTS~\cite{ding2022dists} & Deep Feature Full-Reference & Decouples texture and structure; robust to mild spatial shifts and texture swapping. & Remains a full-reference metric; cannot fully reward plausible out-of-reference synthesis. \\
\midrule
\shortstack[l]{\textbf{No-Reference}\\\textbf{Naturalness}}
& NIQE~\cite{mittal2013niqe} & No-Reference (Decoded only) & Reference-free; directly measures image naturalness and statistical plausibility. & Blind to ground-truth content consistency; sensitive to unnatural high-frequency sharpness. \\
\midrule
\shortstack[l]{\textbf{Distributional}\\\textbf{Realism}}
& FID~\cite{heusel2017fid} & Distributional Set-Reference & Evaluates distribution-level generative realism and structural diversity. & Highly sensitive to sample size and crop protocols; unsuitable for single-image assessment. \\
\bottomrule
\end{tabularx}
\end{minipage}
\vspace{1pt}
}]

\section{Taxonomy of Evaluation Metrics}
In the very-low-bitrate generative image compression regime, relying on a single metric is insufficient to characterize codec performance due to the fundamental rate-distortion-perception trade-off~\cite{blau2018perceptiondistortion,blau2019rdp}. Therefore, we employ a multi-faceted evaluation suite spanning four distinct dimensions:

\begin{list}{}{%
  \setlength{\leftmargin}{1.5em}%
  \setlength{\rightmargin}{0pt}%
  \setlength{\itemsep}{0.35em}%
  \setlength{\parsep}{0pt}%
  \setlength{\topsep}{0.35em}%
}
\item[] \textbf{Pixel and Structural Fidelity (PSNR, SNR, SSIM, MS-SSIM):}~\cite{wang2009mse,wang2004ssim,wang2003msssim} These full-reference metrics measure the strict mathematical and structural fidelity against the ground truth. While they ensure that the reconstructed image does not drift from the reference content, optimizing solely for them tends to produce over-smoothed textures under severe bit constraints.

\item[] \textbf{Full-Reference Perceptual Similarity (LPIPS, DISTS):}~\cite{zhang2018lpips,ding2022dists} Evaluated in deep feature spaces, these metrics better reflect human perceptual judgment of visual artifacts and local structural preservation compared to raw pixel differences.

\item[] \textbf{No-Reference Perceptual Naturalness (NIQE):}~\cite{mittal2013niqe} NIQE measures deviations from a natural-scene statistical model without requiring a reference image, providing a reference-free estimate of statistical naturalness.

\item[] \textbf{Distributional Realism (FID):}~\cite{heusel2017fid} FID evaluates distributional differences between reconstructed and reference image sets. In this work, it is used for beta-pair selection and reported as a dataset-level supplementary metric rather than a per-image measure.
\end{list}

\section{Header-Overhead Accounting}
AFP-GIC uses a fixed 48-bit header comprising two 16-bit image dimensions, an 8-bit field encoding the maximum absolute value of the quantized latent \(\hat{\mathbf{y}}\), and an 8-bit operating-point index. For image \(i\), let \(B_i\) denote the total realized bitstream length in bits, and let \(b_i=B_i/(H_iW_i)\) denote the corresponding bpp. The image-wise header fraction is
\begin{equation}
h_i=\frac{48}{B_i}\times100\%
=\frac{48}{b_iH_iW_i}\times100\%.
\end{equation}
For a dataset containing \(N\) images, each entry in Table~9 of the main manuscript is the arithmetic mean \(\bar{h}=N^{-1}\sum_{i=1}^{N}h_i\) over all images at the corresponding operating point.

\section{Paired Uncertainty Analysis on Kodak}
We assess uncertainty in the Kodak comparison using the saved reconstructions of all 24 images at the five operating points. We draw 10,000 paired image-level resamples with replacement, preserving the pairing across methods and operating points. Table~\ref{tab:kodak_bootstrap} reports mean differences (AFP-GIC minus DC-VIC) with bias-corrected and accelerated (BCa) 95\% confidence intervals. These pointwise intervals quantify image-sampling uncertainty, not training-seed variability, and are unadjusted for multiple comparisons. The corresponding operating points have slightly different actual bitrates and are therefore not exactly rate-matched.

For each resample, we also recompute the mean bitrate and quality at each operating point and fit cubic polynomials to quality against log-bpp. Integrating the difference between the fitted curves over their common log-bpp interval and dividing by its width gives the BD difference. The resulting BD-PSNR gain is +0.233 dB (percentile 95\% CI: [+0.176, +0.304] dB), and the BD-NIQE difference is $-0.069$ (percentile 95\% CI: [$-0.154$, $+0.013$]). These comparisons use DC-VIC as the reference, whereas Table~4 of the main manuscript uses MS-ILLM. The BD-PSNR gain is statistically significant, whereas the favorable BD-NIQE estimate does not reach statistical significance at the 95\% confidence level.

\begin{table*}[!t]
\centering
\caption{Paired PSNR and NIQE differences on Kodak with pointwise 95\% BCa confidence intervals. Differences are AFP-GIC minus DC-VIC; positive PSNR and negative NIQE favor AFP-GIC.}
\label{tab:kodak_bootstrap}
\small
\setlength{\tabcolsep}{5pt}
\renewcommand{\arraystretch}{1.18}
\begin{tabular*}{\textwidth}{@{\extracolsep{\fill}}ccccccc@{}}
\toprule
Point & AFP-GIC bpp & DC-VIC bpp & $\Delta$PSNR (dB) & 95\% CI (dB) & $\Delta$NIQE & 95\% CI \\
\midrule
1 & 0.0516 & 0.0537 & $+0.0562$ & $[-0.0127, +0.1300]$ & $+0.0066$ & $[-0.1072, +0.1365]$ \\
2 & 0.0812 & 0.0860 & $+0.0338$ & $[-0.0504, +0.1039]$ & $-0.0613$ & $[-0.1577, +0.0154]$ \\
3 & 0.1124 & 0.1164 & $+0.2573$ & $[+0.1786, +0.3333]$ & $-0.0630$ & $[-0.1684, +0.0190]$ \\
4 & 0.1456 & 0.1508 & $+0.1827$ & $[+0.1210, +0.2379]$ & $-0.1334$ & $[-0.2092, -0.0490]$ \\
5 & 0.1794 & 0.1888 & $+0.0475$ & $[+0.0055, +0.0825]$ & $-0.1294$ & $[-0.2095, -0.0168]$ \\
\bottomrule
\end{tabular*}
\end{table*}

\bibliographystyle{IEEEtran}
\bibliography{AFP_GIC_Supplementary_References}